\documentclass[10pt]{article}
\pdfoutput=1

\usepackage[utf8]{inputenc}
\usepackage{amsmath,amsthm,amssymb}

\usepackage[paper=a4]{typearea}
\areaset{418pt}{609pt}

\usepackage[bottom]{footmisc}
\usepackage{paralist}
\usepackage{mdframed}
\usepackage{color}
\definecolor{citecolor}{rgb}{0.8,0,0}
\definecolor{linkcolor}{rgb}{0,0,0.8}
\definecolor{urlcolor}{rgb}{0,0,0.8}
\usepackage[colorlinks=true,linktocpage=true,hyperfootnotes=false]{hyperref}
\usepackage[capitalise]{cleveref}
\hypersetup{%
  citecolor=citecolor,
  linkcolor=linkcolor,
  urlcolor=urlcolor
}
\usepackage{microtype}
\usepackage[mathic]{mathtools}

\makeatletter
\def\th@plain{%
  \thm@notefont{}
  \itshape 
}
\def\th@definition{%
  \thm@notefont{}
  \normalfont 
}
\makeatother

\newtheorem{theorem}{Theorem}
\Crefname{theorem}{Theorem}{Theorems}

\newtheorem{corollary}[theorem]{Corollary}
\Crefname{corollary}{Corollary}{Corollaries}

\newtheorem{lemma}[theorem]{Lemma}
\Crefname{lemma}{Lemma}{Lemmata}

\theoremstyle{definition}
\newtheorem{definition}[theorem]{Definition}
\Crefname{definition}{Definition}{Definitions}

\theoremstyle{remark}

\Crefname{remark}{Remark}{Remarks}

\theoremstyle{definition}

\Crefname{algorithm}{Algorithm}{Algorithms}

\Crefname{proof}{Proof}{Proofs}

\crefformat{equation}{#2(#1)#3}
\Crefformat{equation}{Equation #2(#1)#3}
\crefrangeformat{equation}{#3(#1)#4--#5(#2)#6}
\Crefrangeformat{equation}{Equations #3(#1)#4 to #5(#2)#6}

\Crefname{section}{Section}{Sections}

\newtheorem*{keywords}{Keywords}

\newcommand{\N}{\ensuremath{\mathbb{N}}}
\newcommand{\Z}{\ensuremath{\mathbb{Z}}}
\newcommand{\R}{\ensuremath{\mathbb{R}}}
\newcommand{\C}{\ensuremath{\mathbb{C}}}

\newcommand{\tz}{\tilde{z}}

\newcommand{\tP}{\tilde{P}}
\newcommand{\ty}{\tilde{y}}

\newcommand{\tO}{\tilde{O}}
\newcommand{\ceil}[1]{\lceil #1 \rceil}
\newcommand{\floor}[1]{\lfloor #1 \rfloor}
\newcommand{\Mea}{\operatorname{Mea}}
\newcommand{\AD}{\textsc{A}\textsc{New}\textsc{Dsc}}
\newcommand{\ignore}[1]{}
\newcommand{\multipoint}[2]{#1[#2]}
\newcommand{\set}[2]{\{ #1\; ;\; #2 \}}
\newcommand{\sset}[1]{\{ #1 \}}
\newcommand{\assign}{\mathop{:=}}

\newcommand{\abs}[1]{\left| #1 \right|}
\newcommand{\norm}[1]{\left\| #1 \right\|}
\newcommand{\onenorm}[1]{\norm{#1}}

\newcommand{\normshort}[1]{\| #1 \|}
\newcommand{\onenormshort}[1]{\normshort{#1}}
\newcommand{\sgn}{\operatorname{sgn}}
\newcommand{\var}{\operatorname{var}}
\newcommand{\Disc}{\mathrm{Disc}}

\newcommand{\cosym}{/\!\!/}

\newcommand{\makeremark}[2]{
  \newcommand{#1}[1]
    {
    {\color{blue}
     $\longrightarrow$ \textsc{#2: }
     ##1
     $\longleftarrow$}
    }
}

\newcommand{\MS}[1]{\textcolor{black}{#1}}
\newcommand{\Kurt}[1]{\textcolor{black}{#1}}

\makeremark{\MSs}{Michael says}
\makeremark{\Kurts}{Kurt says}

\title{Computing Real Roots of Real Polynomials 
}
\author{Michael Sagraloff\thanks{MPI for Informatics, Campus E1 4, 66123 Saarbr\"{u}cken, Germany. emails: \url{msagralo@mpi-inf.mpg.de} and \url{mehlhorn@mpi-inf.mpg.de}. The author ordering deviates from the default alphabetic ordering used in Theoretical Computer Science, because the first author contributed significantly more to the paper than the second author.} \and Kurt Mehlhorn$^*$}

\begin{document}

\maketitle

\begin{abstract}Computing the roots of a univariate polynomial is a fundamental and long-studied problem of computational algebra with applications in mathematics, engineering, computer science, and the natural sciences. 
For isolating as well as for approximating all complex roots, the best algorithm known is based on 
an almost optimal method for approximate polynomial factorization, introduced by Pan in 2002. Pan's factorization algorithm goes back to the splitting circle method from Sch\"onhage in 1982. The main drawbacks of Pan's method are that it is quite involved\footnote{In Victor Pan's own words: ``Our algorithms are quite involved, and their implementation would require a non-trivial work, incorporating numerous known implementation techniques and tricks''. In fact, we are not aware of any implementation of Pan's method.\label{footnote:paniscomplex}} and that all roots have to be computed at the same time. For the important special case, where only the real roots have to be computed, much simpler methods are used in practice; however, they considerably lag behind Pan's method with respect to complexity.

\Kurt{In this paper, we resolve this discrepancy by introducing a hybrid of the Descartes method and Newton iteration, denoted $\AD$, which is simpler than Pan's method, but achieves a run-time comparable to it. Our algorithm computes isolating intervals for
the real roots of any real square-free polynomial, given by an oracle that provides arbitrary good approximations of the polynomial's coefficients. $\AD$ can also be used to only isolate the roots in a given interval and to refine the isolating intervals to an arbitrary small size; it achieves near optimal complexity for the latter task.}
\end{abstract}

\begin{keywords}
    root finding, root isolation, root refinement, approximate arithmetic, certified computation, complexity analysis
  \end{keywords}

\section{Introduction}

Computing the roots of a univariate polynomial is a fundamental problem in computational algebra. Many problems from mathematics, 
engineering, computer science, and the natural sciences can be reduced to solving a system of polynomial equations, which in turn reduces to solving a polynomial equation in one variable by 
means of elimination techniques such as resultants or Gr\"obner Bases. Hence, it is not surprising that this problem has been extensively studied and that numerous approaches have been developed; see~\cite{McNamee:2002,McNamee2007,McNamee-Pan2013,Pan:history} for 
an extensive (historical) treatment. \Kurt{Finding all roots of a polynomial and the  approximate factorization of a polynomial into linear factors are closely related. The most efficient algorithm for approximate factorization is due to Pan~\cite{Pan:alg}; it is based on the splitting circle method of Sch\"onhage~\cite{schonhage:fundamental} and considerably refines it. From an approximate factorization, one can derive arbitrary good approximations of all complex roots as well as corresponding isolating disks; e.g. see~\cite{Pan:survey,MSW-rootfinding2013}. The main drawbacks of Pan's algorithm are that it is quite involved (see Footnote~\ref{footnote:paniscomplex}) and that it necessarily computes all roots, i.e., cannot be used to only isolate the roots in a given region. It has not yet been implemented.} In contrast, simpler approaches, namely Aberth's, Weierstrass-Durand-Kerner's and QR algorithms, found their way into popular packages such as \textsc{MPSolve}~\cite{Bini-Fiorentino} or \texttt{eigensolve}~\cite{DBLP:journals/jsc/Fortune02}, although their excellent empirical behavior has never been entirely verified in theory.

In parallel, there is steady ongoing research on the 
development of dedicated real roots solvers that also allow to search for the roots only in a given interval. Several methods (e.g.~Sturm 
method, Descartes method, continued fraction method, Bolzano method) have been proposed, and there exist many corresponding implementations in computer algebra systems. With respect to computational complexity, all of these methods lag considerably behind the splitting circle approach. \emph{In this paper, we resolve this discrepancy by introducing a hybrid of the Descartes method and Newton iteration, denoted $\AD$ (read approximate-arithmetic-Newton-Descartes). \Kurt{Our algorithm is simpler than Pan's algorithm, is already implemented with very promising results for polynomials with integer coefficients~\cite{Kobel-Rouillier-Sagraloff}, and has a complexity comparable to that of Pan's method.}}

\subsection{Algorithm and Results}

Before discussing the related work in more detail, we first outline our algorithm and provide the main results. Given a square-free univariate polynomial $P$ with real coefficients, the goal is to compute disjoint intervals on the real line such that all real roots are contained in the union of the intervals and each interval contains exactly one real root. The Descartes or Vincent-Collins-Akritas\footnote{Descartes did not formulate an algorithm for isolating the real roots of a polynomial but (only) a method for bounding the number of positive real roots of a univariate polynomial (Descartes' rule of signs). Collins and Akritas~\cite{Collins-Akritas} based on ideas going back to Vincent formulated a bisection algorithm based on Descartes' rule of signs.} method is a simple and popular algorithm for real root isolation. It starts with an open interval guaranteed to contain all real roots and repeatedly subdivides the interval into two open intervals and a split point. The split point is a root if and only if the polynomial evaluates to zero at the split point. For any interval $I$, Descartes' rule of signs (see Section~\ref{sec:DRS:basis}) allows one to compute an integer $v_I$, which bounds the number $m_I$ of real roots in $I$ and is equal to $m_I$, if $v_I \le 1$. The method discards intervals $I$ with $v_I = 0$, outputs intervals $I$ with $v_I = 1$ as isolating intervals for the unique real root contained in them, and splits intervals $I$ with $v_I \ge 2$ further. The procedure is guaranteed to terminate for square-free polynomials, as $v_I = 0$,  if the circumcircle of $I$ \Kurt{(= the one-circle region of $I$)} contains no root of $p$, and $v_I = 1$, if the union of the circumcircles of the two equilateral triangles with side $I$ \Kurt{(the two-circle region of $I$)} contains exactly one root of $I$, see Figure~\ref{fig:Obreshkoff} \Kurt{on page~\pageref{page:Obreshkoff}}.

The advantages of the Descartes method are its simplicity and the fact that it applies to polynomials with  real coefficients. The latter has to be taken with a grain of salt. The method uses the four basic arithmetic operations (requiring only divisions by two) and the sign-test for numbers in the field of coefficients. In particular, if the input polynomial has integer or rational coefficients, the computation stays within the rational numbers. Signs of rational numbers are readily determined. In the presence of non-rational coefficients, the sign-test becomes problematic. 

The disadvantages of the Descartes method are its inefficiency when roots are clustered and its need for exact arithmetic. When roots are clustered, there can be many subsequent subdivision steps, say splitting $I$ into $I'$ and $I''$, where $\min(v_{I'},v_{I''}) = 0$ and $\max(v_{I'},v_{I''}) = v_I$. Such subdivision steps exhibit only linear convergence to the cluster of roots as an interval $I$ is split into equally sized intervals. The need for exact arithmetic stems from the fact that it is crucial for the correctness of the algorithm that sign-tests are carried out exactly. It is known how to overcome each one of the two weaknesses separately; however, it is not known how to overcome them simultaneously. Our main result achieves this. That is, we present an algorithm $\AD$ that overcomes both shortcomings at the same time. Our algorithm applies to real polynomials given through coefficient oracles\footnote{A coefficient oracle provides arbitrarily good approximations of the coefficients.}, and our algorithm works well in the presence of clustered roots. 

We are now ready for a high-level description of the algorithm. The correctness of the Descartes method rests on exact sign computations; however, the exact computation of the sign of a number does not necessarily require the exact computation of the number. The Descartes algorithm uses the sign-test in two situations: It needs to determine whether the polynomial evaluates to zero at the split point, and it determines $v_I$ as the number of sign changes in the coefficient sequence of polynomials $P_I$, where $P_I$ is a polynomial determined by the interval $I$ and the input polynomial $P$. We borrow from~\cite{DBLP:conf/casc/EigenwilligKKMSW05} the idea of carefully choosing split points so as to guarantee that $P$ is relatively large at split points. Our realization of the idea is, however, quite different and is based on a fast method for approximately evaluating a polynomial at many points~\cite{DBLP:journals/jc/Kirrinnis98,DBLP:journals/corr/abs-1304.8069,DBLP:journals/corr/KobelS14,Ritzmann19861,JvH2008}. We describe the details in Section~\ref{sec:multipoint}. The choice of a split point where the polynomial has a large absolute value has a nice consequence for the sign change computation. Namely, the cases $v_I = 0$, $v_I = 1$, and $v_I > 1$ can be distinguished with approximate arithmetic. We give more details in Section~\ref{sec:apxdsc}. 

The recursion tree of the Descartes method may have size $\Omega(n \tau)$. However, there are only $O(n)$ nodes where both children are subdivided further. This holds because the sum of the number of sign changes at the children of a node is at most the number of sign changes at the node. 
In other words, large subdivision trees must have long chains of nodes, where the interval that is split off is immediately discarded. Long chains are an indication of clustered roots. We borrow from~\cite{NewDsc}  the idea of traversing such chains more efficiently by combining Descartes' rule of signs, Newton iteration, and a subdivision strategy similar to Abbott's quadratic interval refinement (QIR for short) method~\cite{abbott-quadratic}. As a consequence, quadratic convergence towards the real roots instead of the linear convergence of pure bisection is achieved in most iterations. Again, our realization of the idea is quite different. For example, Newton iteration towards a cluster of $k$ roots needs to know $k$ (see equation~\ref{newton}). The algorithm in~\cite{NewDsc} uses exact arithmetic and, therefore, can determine the exact number $v_I$ of sign changes in the coefficient sequence of polynomials $P_I$; $v_I$ is used as an estimate for the size of a cluster contained in interval $I$. We cannot compute $v_I$ (in the case $v_I > 1$, we only know this fact and have no further knowledge about the value of $v_I$) and hence have to estimate the size of a cluster differently. We tentatively perform Newton steps from both endpoints of the interval and determine a $k$ such that both attempts essentially lead to the same iterate. Also, we use a variant of quadratic interval refinement because it interpolates nicely between linear and quadratic convergence. Details are given in Section~\ref{sec:Algorithm}. This completes the high-level description of the algorithm.

A variant of our algorithm using randomization instead of multi-point evaluation for the choice of split point has already been implemented for polynomials with integer coefficients~\cite{Kobel-Rouillier-Sagraloff}. First experiments are quite promising. The implementation seems to be competitive with the fastest existing real root solvers for all instances and much superior for some hard instances with clustered roots, e.g., Mignotte polynomials.\footnote{Mignotte~\cite{Mignotte}[Section 4.6] considers the polynomial $x^n - 2 ( a x - 1)^2$, where $n \ge 3$ and $a \ge 10$ is an integer. He shows that the polynomial has two real roots in the interval $(1/a - h,1/a + h)$, where $h = a^{-(n + 2)/2}$.} 

The worst-case bit complexity of our algorithm essentially matches the bit complexity of the algorithms that are based on Pan's algorithm. More specifically, we prove the following theorems:

\begin{theorem}\label{thm: main} Let $P(x) = P_n x^n + \ldots + P_1 x^1 + P_0 \in \R[x]$ be a real, square-free polynomial of degree $n$ with\footnote{If $P(x)$ is arbitrary, it suffices to determine an integer $t$ with $2^t/4 \le P_n \le 2^t$ and to consider $2^{-t} P(x)$.}  $1/4 \le P_n \le 1$. The algorithm $\AD$ determines isolating intervals
for all real roots of $P$ with a number of bit operations bounded by\footnote{The $\tilde{O}$-notation suppresses polylogarithmic factors, i.e., $\tilde{O}(T) = O(T (\log T)^k)$, where $k$ is any fixed integer.}
\begin{align}\label{main:complexity bound:1}
&\tilde{O}(n\cdot(n^2+n \log \Mea(P)+\sum_{i=1}^n \log M(P'(z_i)^{-1}))) \\
                              &\qquad= \tilde{O}(n(n^2 + n \log \Mea(P) + \log M(\Disc(P)^{-1}))).\nonumber
\end{align}
The coefficients of $P$ have to be approximated with quality\footnote{Let $L \ge 1$ be an integer and let $z$ be real. We call $\tz = s \cdot 2^{-(L+1)}$ with $s \in \Z$  an \emph{approximation of $z$ with quality $L$} if $\abs{z - \tz} \le 2^{-L}$.}
\[
\tilde{O}(n+\tau_{P}+\max_{i}(n\log M(z_{i})  + \log M(P'(z_{i})^{-1}) )).\]
Here $M(x) := \max(1, \abs{x})$, $z_1$ to $z_n$ are the roots of $P$, $\Mea(P):=|P_n|\cdot\prod_{i=1}^{n} M(|z_i|)$ denotes the \emph{Mahler Measure} of $P$, $\Disc(P) := P_n^{2n - 2} \prod_{1 \le i < j \le n} (z_j - z_i)^2$ is the \emph{discriminant} of $P$, $\tau_P := M( \abs {\log(\max_i \abs{P_i})})$,
 and $P'$ is the derivative of $P$. 
\end{theorem}

For polynomials with integer coefficients, the bound can be stated more simply. 

\begin{theorem}\label{thm: main integer case}
For a square-free polynomial $P \in\Z[x]$ with integer coefficients of absolute value $2^\tau$ or less, the algorithm $\AD$ computes isolating intervals for all real roots of $P$ with $\tilde{O}(n^3+n^2\tau)$ bit operations. \Kurt{If $P$ has only $k$ non-vanishing coefficients, the bound improves to $\tilde{O}(n^2(k + \tau))$.}
\end{theorem}

For general real polynomials, the bit complexity of the algorithm $\AD$ matches the bit complexity of the best algorithm known (\cite{MSW-rootfinding2013}). For polynomials with integer coefficients, the bit complexity of the best algorithm known (\cite[Theorem 3.1]{Pan:survey}) is $\tilde{O}(n^{2}\tau)$, however, for the price of using $\Omega(n^{2}\tau)$ bit operations for every input.\footnote{More precisely, Pan's algorithm uses $O(n(\log^2n)(\log^2 n+\log b))$ arithmetic operations carried out with a precision $b$ of size $\Omega(n(\tau+\log n))$, and thus $O(n(\log^2 n+\log \tau)\mu(n(\tau+\log n))$ bit operations, where $\mu(b)$ denotes the cost for multiplying two $b$-bit integers. A straight forward, but tedious, calculation yields the bound $O((n^3+n^2\tau)(\log^6n)(\log^2 (n\tau)))$ for the bit complexity of our method, when ignoring $\log \log $-factors. This is weaker than Pan's method, however, we have neither tried to optimize our algorithm in this direction nor to consider possible amortization effects in the analysis when considering $\log$-factors.}
Both algorithms are based on Pan's approximate factorization algorithm~\cite{Pan:alg}, which is quite complex, and always compute all complex roots.

Our algorithm is simpler and has the additional advantage that it can be used to isolate the real roots in a given interval instead of isolating all roots. Moreover, the complexities stated in the theorems above are worst-case complexities, and we expect a better behavior for many instances. \Kurt{We have some theoretical evidence for this statement. For sparse integer polynomials with only $k$ non-vanishing coefficients, the complexity bound reduces from $\tilde{O}(n^2(n + \tau))$ to $\tilde{O}(n^2(k + \tau))$ (Theorem~\ref{thm: main integer case}).}
\Kurt{Also, if we restrict the search for roots in an interval $I_0$, then only the roots contained in the one-circle region $\Delta(I_0)$ of $I_0$ have to be considered in the 
complexity bound (\ref{main:complexity bound:1}) in Theorem~\ref{thm: main}.  More precisely, the first summand $n^3$ can be replaced by $n^2\cdot m$, where $m$ denotes the number of roots contained in $\Delta(I_0)$, and the last summand $\sum_{i=1}^n \log M(P'(z_i)^{-1})$ can be replaced by $\sum_{i:z_i\in \Delta(I_0)} \log M(P'(z_i)^{-1})$. We can also bound the size of the subdivision tree in terms of the number of sign changes in the coefficient sequence of the input polynomial (Theorem~\ref{thm:treesize}). In particular, if $P$ is a sparse integer polynomial, e.g., a Mignotte polynomial, with only $(\log(n\tau))^{O(1)}$ non-vanishing coefficients, our algorithm generates a tree of size $(\log(n\tau))^{O(1)}$. Our algorithm generates a tree of size $(\log(n\tau))^{O(1)}$, whereas bisection methods, such as the classical Descartes method, generate a tree of size $\Omega(n\tau)$, and the continued fraction method~\cite{Collins} generates a tree of size $\Omega(n)$.  }

A modification of our algorithm can be used to refine roots once they are isolated. 

\begin{theorem}\label{main: refinement}
Let $P = P_n x^n + \ldots +P_0 \in \R[x]$ be a real, square-free polynomial with $1/4 \le \abs{P_n} \le 1$, and let $\kappa$ be a positive integer. Isolating intervals of size less than $2^{-\kappa}$ for all real roots can be computed with
\begin{align*} \tilde{O}(n\cdot(\kappa+n^2+n \log \Mea(P)+  \log M(\Disc(P)^{-1})))\end{align*}
bit operations using coefficient approximations of quality
$$\tilde{O}(\kappa+n+\tau_{P}+\max_{i}(n\log M(z_{i})  + \log M(P'(z_{i})^{-1}) )).
$$
For a square-free polynomial $P$ with integer coefficients of size less than $2^\tau$,  isolating intervals of size less than $2^{-\kappa}$ for all real roots can be computed with $\tilde{O}(n(n^2+n\tau+\kappa))$ bit operations.
\end{theorem}

The complexity of the root refinement algorithm is $\tO(n \kappa)$ for large $\kappa$. This is optimal up to logarithmic factors, as the size of the output is $\Omega(n \kappa)$. The complexity matches the complexity shown in~\cite{MSW-rootfinding2013}, and when considered as a function in $\kappa$ only, it also matches the complexity as shown in ~\cite{DBLP:journals/corr/abs-1104-1362} and as sketched in~\cite{pantsi:ISSAC13}.

\subsection{Related Work}\label{sec: related work}

Isolating the roots of a polynomial is a fundamental and well-studied problem. One is either interested in isolating all roots, or all real roots, or all roots in a certain subset of the complex plane. A related problem is the approximate factorization of a polynomial, that is, to find $\tz_1$ to $\tz_n$ such that $\onenormshort{P(x) - P_n \prod_{1 \le i \le n} (x - \tz_i)}$ is small. Given the number of distinct complex roots of a polynomial $P$, one can derive isolating disks for all roots from a sufficiently good approximate factorization of $P$; see~\cite{MSW-rootfinding2013}. In particular, this approach applies to polynomials that are known to be square-free.

Many algorithms for approximate factorization and root isolation are known, see~\cite{McNamee:2002,McNamee2007,McNamee-Pan2013,Pan:history} for surveys. The algorithms can be roughly split into two groups: There are iterative methods for simultaneously approximating all roots (or a  single root if a sufficiently good approximation is already known); there are subdivision methods that start with a region containing all the roots of interest, subdivide this region according to certain rules, and use inclusion- and exclusion-predicates to certify that a region contains exactly one root or no root. Prominent examples of the former group are the Aberth-Ehrlich method (used for \textsc{MPSolve}~\cite{Bini-Fiorentino}) and the Weierstrass-Durand-Kerner method. These algorithms work well in practice and are widely used. However, a complexity analysis and global convergence proof is missing. Prominent examples of the second group are the Descartes method~\cite{Collins-Akritas,eigenwillig-phd,DBLP:conf/casc/EigenwilligKKMSW05,rouillier-zimmermann:roots:04}, the Bolzano method~\cite{DBLP:journals/jsc/BurrK12,Yap-Sagraloff-Bolzano}, the Sturm method~\cite{du-sharma-yap:sturm:07}, the continued fraction method~\cite{akritas-strzebonski:comparison:05,sharma,te-cf:08}, and the splitting circle method~\cite{schonhage:fundamental,Pan:alg}. 

The splitting circle method was introduced by 
Sch\"{o}nhage~\cite{schonhage:fundamental} and later considerably refined by Pan~\cite{Pan:alg}. Pan's algorithm computes an 
approximate factorization and can also be used to isolate all complex roots of a polynomial. For integer polynomials, it isolates all roots with $\tO(n^2 \tau)$ bit operations. It also serves as the key routine in a recent 
algorithm~\cite{MSW-rootfinding2013} for complex root isolation, which achieves a worst case complexity similar to the one stated in 
our main theorem. There exists a ``proof of concept'' implementation of the splitting circle method in the computer algebra 
system Pari/GP~\cite{Gourdon}, whereas we are not aware of any implementation of Pan's method itself.

The Descartes, Sturm, and continued fraction methods isolate only the real roots. They are popular for their simplicity, ease of implementation, and  practical efficiency. The papers~\cite{snc-benchmarks09,Kamath,rouillier-zimmermann:roots:04} report about implementations and experimental comparisons. The price for the simplicity is a considerably larger worst-case complexity. We concentrate on the Descartes method. 

The standard Descartes method has a complexity of $\tO(n^4 \tau^2)$ for isolating the real roots of an integer polynomial of degree $n$ with coefficients bounded by $2^\tau$ in absolute value, see~\cite{Eigenwillig-Sharma-Yap}. The size of the recursion tree is $O(n (\tau + \log n))$, and $\tO(n)$ arithmetic operations on numbers of bitsize $O(n^2(\tau+\log n))$ need to be performed at each node. For $\tau = \Omega(\log n)$, these bounds are tight, that is, there are examples where the recursion tree has size $\Omega(n\tau)$ and the numbers to be handled grow to integers of length $\Omega(n^2 \tau)$ bits. 

Johnson and Krandick~\cite{Johnson-Krandick} and Rouillier and Zimmermann~\cite{rouillier-zimmermann:roots:04} suggested the use of approximate arithmetic to speed up the Descartes method. They fall back on exact arithmetic when sign computations with approximate arithmetic are not conclusive. 
Eigenwillig et al.~\cite{DBLP:conf/casc/EigenwilligKKMSW05} were the first to describe a Descartes method that has no need for exact arithmetic. It works for polynomials with real coefficients given through coefficient oracles and 
isolates the real roots of a square-free real polynomial $P(x) = P_nx^n+\ldots+P_0$ with root separation\footnote{The root separation of a polynomial is the minimal distance between two roots} $\rho$, coefficients $\abs{P_n} \ge 1$, and $\abs{P_i} \le 2^{\tau}$, with an expected cost of $O(n^4 (\log(1/\rho) +\tau)^2)$ bit operations. For polynomials with integer coefficients, it constitutes no improvement. Sagraloff~\cite{Sagraloff2014DSC} gave a variant of the Descartes method that, when applied to integer polynomials, uses approximate arithmetic with a working precision of only $\tO(n\tau)$ bits. This leads to a bit complexity of $\tO(n^3 \tau^2)$; the recursion tree has size $O(n (\tau + \log n))$, there are $\tO(n)$ arithmetic operations per node, and arithmetic on numbers of length $\tO(n \tau)$ bits is required. 

As already mentioned before, the recursion tree of the Descartes method may have size $\Omega(n \tau)$, and there are only $O(n)$ nodes where both children are subdivided further. Thus, large subdivision trees must have
long chains of nodes, where one child is immediately discarded.  
Sagraloff~\cite{NewDsc}  showed how to traverse such chains more efficiently using a hybrid of bisection and Newton iteration.
His method reduces the size of the recursion tree to $O(n \log (n\tau))$, which is optimal up to logarithmic factors.\footnote{As there might be $n$ real roots, $n$ is a trivial lower bound on the worst-case tree size.} It only applies to polynomials with integral coefficients, uses exact rational arithmetic, and achieves a bit complexity of $\tO(n^3 \tau)$. In essence, the size of the recursion tree is $O(n\log(n\tau))$, there are $\tO(n)$ arithmetic operations per node, and arithmetic is on numbers of amortized length $\tO(n \tau)$ bits. 
Other authors have also shown how to use Newton iteration for faster convergence after collecting enough information in a slower initial phase. 
For instance, Renegar~\cite{DBLP:journals/jc/Renegar87} combines 
the Schur-Cohn test with Newton iteration. His algorithm makes crucial use of the fact that each (sufficiently small) cluster 
consisting of $k$ roots of a polynomial $P$ induces the existence of a single nearby ordinary root of the $(k-1)$-th 
derivative $P^{(k-1)}$, and thus, one can apply Newton iteration to $P^{(k-1)}$ in order to efficiently compute a good 
approximation of this root (and the cluster) if the multiplicity $k$ of the cluster is known. Several methods have been proposed to compute 
$k$, such as approximating the winding number around the perimeter of a disk by discretization of the contour, 
numerically tracking a homotopy path near a cluster, or the usage of Rouch\'e's and Pellet's Theorem. For a more detailed discussion and 
more references, we refer the reader to Yakoubsohn's paper~\cite{Yakoubsohn2000}. Compared to these methods, our approach 
is more light-weight in the sense that we consider a trial and error approach that yields less information in the 
intermediate steps without making sacrifices with respect to the speed of convergence. More precisely, we use a variant of 
Newton iteration for multiple roots in order to guess the multiplicity of a cluster (provided that such a cluster exists), however, we actually never verify our guess. Nevertheless, our analysis shows that, if there exists a well-separated cluster of roots, then our method yields the correct 
multiplicity. From the so obtained "multiplicity", we then compute a better approximation of the cluster (again assuming the existence of a cluster) and finally aim to verify via Descartes' Rule of Signs that we have not missed any root. 
If the latter cannot be verified, we fall back to bisection.   

The bit complexity of our new algorithm is $\tilde{O}(n^3 + n^2 \tau)$ for integer polynomials. Similar as in~\cite{NewDsc}, the size of the recursion tree is $O(n\log(n\tau))$ due to the combination of bisection and Newton steps. The number of arithmetic operations per node is $\tO(n)$ and
arithmetic is on numbers of amortized length $\tO(n + \tau)$ bits (instead of $\tO(n\tau)$ as in~\cite{NewDsc})  due to the use of approximate multipoint evaluation and approximate Taylor shift.\smallskip

Root refinement is the process of computing better approximations once the roots are isolated. In~\cite{DBLP:journals/corr/abs-1104-1362,DBLP:journals/corr/KobelS14,DBLP:journals/jc/Kirrinnis98,MSW-rootfinding2013,pantsi:ISSAC13}, algorithms have been proposed which scale like $\tO(n \kappa)$ for large $\kappa$. The former two algorithms are based on the splitting circle approach and compute approximations of all complex roots. The latter two solutions are dedicated to approximate only the real roots. They combine a fast convergence method (i.e., the secant method and Newton iteration, respectively) with approximate arithmetic and efficient multipoint evaluation; however, there are details missing in~\cite{pantsi:ISSAC13} when using multipoint evaluation. In order to achieve complexity bounds comparable to the one stated in Theorem~\ref{main: refinement}, the methods from~\cite{DBLP:journals/corr/abs-1104-1362,pantsi:ISSAC13} need as input isolating intervals whose size is comparable to the separation of the corresponding root, that is, the roots must be "well isolated". This is typically achieved by using a fast method, such as Pan's method, for complex root isolation first. Our algorithm does not need such a preprocessing step.

Very recent work~\cite{sagraloff-sparse-2014} on isolating the real roots of a sparse integer polynomial $P\in\Z[x]$ makes crucial use of a 
slight modification of the subroutine Newton-Test as proposed in Section~\ref{sec:Newton-Test and Boundary-Test}. There, it is used to refine an isolating 
interval $I$ for a root of $P$ in a number of arithmetic operations that is nearly linear in the number of roots that are 
close to $I$ and polynomial in $m\cdot\log (n\cdot \tau_P)$, where $n:=\deg P$ and $m$ denotes the number of non-vanishing coefficients of $P$. This eventually yields the first real root isolation algorithm that needs only a number of arithmetic operations over the rationals that is polynomial in the input size of the sparse representation of $P$. Furthermore, for very sparse polynomials (i.e.~$m=O(\log^c(n\tau_P))$ with $c$ a constant), the algorithm from~\cite{sagraloff-sparse-2014} uses only $\tilde{O}(n\tau_P)$ bit operations to isolate all real roots of $P$ and is thus near-optimal.

\subsection{Structure of Paper and Reading Guide}

We introduce our new algorithm in Section~\ref{sec:Algorithm} and analyze its complexity in Section~\ref{sec:analysis}. We first derive a bound on the size of the subdivision tree (Section~\ref{sec:size-subdivision-tree}) and then a bound on the bit complexity (Section~\ref{sec:bit-complexity}). Section~\ref{sec:Root Refinement} discusses root refinement. Section~\ref{sec:basics} provides background material, which we recommend to go over quickly in a first reading of the paper. We provide many references to Section~\ref{sec:basics} in Sections~\ref{sec:Algorithm} and~\ref{sec:analysis} so that the reader can pick up definitions and theorems as needed.

\section{The Basics}\label{sec:basics}

\subsection{Setting and Basic Definitions}\label{sec:setting-definitions}

We consider a square-free polynomial 
\begin{equation}\label{def:P}                                 
P(x)=P_{n}x^{n}+\ldots +P_{0}\in\R[x],  \quad\text{where $n \ge 2$ and $1/4 \le P_n \le 1$}. 
\end{equation}
We fix the following notations. 

\begin{definition}\label{def:basic definitions}\mbox{}\par
\begin{compactenum}[(1)]
\item $M(z):=\max(1, |z|)$ for all $z\in \mathbb{C}$.
\item $\onenorm{P} :=\norm{P}_{1}:= \abs{P_0} + \ldots + \abs{P_n}$ denotes the $1$-norm of $P$, and $\norm{P}_{\infty}:=\max_{i}|P_{i}|$ denotes the infinity-norm of $P$. 
\item $\tau_{P}:=M(\log\norm{P}_\infty)$.
\item $z_{1},\ldots,z_n\in\C$ are the complex roots of $P$. 
\item For each root $z_{i}$, we define the \emph{separation of $z_{i}$} as the value
$\sigma_i:=\sigma(z_{i},P):=\min_{j\in\{1,\ldots,n\}\backslash i}|z_{i}-z_{j}|$. The \emph{separation of $P$} is defined as $\sigma_P:=\min_{i}\sigma(z_i,P)$.
\item $\Gamma_P:=M(\log\max_{i} |z_i|)$ denotes the \emph{logarithmic root bound} of $P$, and
\item $\Mea(P):=|P_n|\cdot\prod_{i=1}^{n} M(|z_i|)$ denotes the \emph{Mahler Measure} of $P$. 
\item For an interval $I=(a,b)$, $m(I):=\frac{a+b}{2}$ denotes the \emph{midpoint} and $w(I):=b-a$  the \emph{width of $I$}. The open disk in complex space with center $m(I)$ and radius $\frac{w(I)}{2}$ is denoted by $\Delta(I)$. We call $\Delta(I)$ the \emph{one-circle region of $I$}.\footnote{The choice of name will become clear when we discuss Descartes' rule of signs in Section~\ref{sec:DRS:basis}; see also Figure~\ref{fig:Obreshkoff}.}
\item $\mathcal{M}(I)$ denotes the set of roots of $P$ which are contained in $\Delta(I)$.
\item A dyadic fraction is any rational of the form $s\cdot 2^{-\ell}$ with $s \in \Z$ and $\ell \in \Z_{\ge 0}$. \end{compactenum}
\end{definition}

We assume the existence of an oracle that provides arbitrary good approximations of the  polynomial $P$. 
Let $L \ge 1$ be an integer and let $z$ be a real. We call $\tz = s \cdot 2^{-(L+1)}$ with $s \in \Z$ an \emph{(absolute) $L$-approximation} of $z$ or an \emph{approximation of $z$ with quality $L$} if $\abs{z - \tz} \le 2^{-L}$. We call a polynomial $\tP=\tP_{n}x^{n}+\ldots+\tP_{0}$ an
\emph{(absolute) $L$-approximation} of $P$ or an \emph{approximation of quality $L$} if every coefficient of $\tP$ is an approximation of quality $L$ of the corresponding coefficient of $P$. 
We assume that we can obtain such an approximation $\tP$ at $O(n(L+\tau_P))$ cost. This is the cost of reading the coefficients of $\tP$.

We have $\tau_P\le M(\log (2^n\cdot \Mea( P)))\le M(n+n\Gamma_P)=n(1+\Gamma_P)\le 2n\Gamma_P$. According to~\cite[Theorem 1]{MSW-rootfinding2013} (or~\cite[Section 6.1]{Sagraloff2014DSC}), we can compute an integer approximation $\tilde{\Gamma}_{P}$ of $\Gamma_{P}$ with 
\begin{align}\label{def:Gammatilde}
\Gamma_{P}+1\le \tilde{\Gamma}_{P}\le\Gamma_{P}+8\log n+1   
\end{align}
with $\tO(n^{2}\Gamma_{P})$ many bit operations. From $\tilde{\Gamma}_{P}$, we can then immediately derive a $\Gamma=2^{\gamma}$, with $\gamma:=\lceil \log \tilde{\Gamma}_{P}\rceil\in\N_{\ge 1}$, such that
\begin{align}\label{def:Gamma}
\Gamma_{P}+1\le \tilde{\Gamma}_{P}\le \Gamma\le 2\cdot\tilde{\Gamma}_{P} \le 2\cdot(\Gamma_{P}+8\log n+1).
\end{align}
Thus, $2^{\Gamma}=2^{2^{\gamma}}$ is an upper bound for the modulus of all roots (in fact, we have $2^{\Gamma}\ge \max_{i}|z_{i}|+1$ for all $i=1,\ldots,n$), and $\Gamma=O(\Gamma_{p}+\log n)$.

\subsection{Approximate Polynomial Evaluation}\label{sec:multipoint}

We introduce the notions \emph{multipoint} (Definition~\ref{definition:multipoint}) and \emph{admissible point} (Definition~\ref{admissible point}). A point $x^*$ in a set $X$ is admissible if $\abs{P(x^*)} \ge \frac{1}{4}\cdot\abs{P(x)}$ for all $x \in X$. We show how to efficiently compute an admissible point in a multipoint (Corollary~\ref{cor:multipoint}) and derive a lower bound on the value of $P$ at such a point. Corollary~\ref{cor:multipoint} is our main tool for choosing subdivision points.

\begin{theorem}\label{lem:singlepolyeval}
Let $P$ be a polynomial as defined in (\ref{def:P}), $x_{0}$ be a real point, and $L$ be a positive integer. 
\begin{compactenum}[(a)]
\item The algorithm stated in the proof of part (a) computes an approximation $\tilde{y}_{0}$ of $y_{0}:=P(x_{0})$ with $|y_{0}-\tilde{y}_{0}|\le 2^{-L}$ with
\[
\tilde{O}(n(\tau_{P}+n\log M(x_{0})+L)).
\]
bit operations using approximations of the coefficients of $P$ and the point $x_0$ of quality $O(\tau_{P}+n\log M(x_0)+L+\log n)$. 

\item Suppose $y_{0}\neq 0$. The algorithm stated in the proof of part (b) computes an integer $t$ with $2^{t-1}\le |y_{0}|\le 2^{t+1}$ with 
\[\tilde{O}(n(\tau_{P}+n\log M(x_{0})+\log M(y_{0}^{-1})))\]
 bit operations. The computation uses fixed-precision arithmetic with a precision of $O(\tau_{P}+n\log M(x_{0})+M(y_{0}^{-1})+\log n)$ bits.
\end{compactenum}
\end{theorem}

\begin{proof}
Part (a) follows directly from~\cite[Lemma 3]{DBLP:journals/corr/abs-1104-1362}, where it has been shown that we can compute a desired approximation $\tilde{y}_{0}$ via the Horner scheme and fixed-precision interval arithmetic, with a precision of $O(\tau_{P}+n\log M(x_0)+L+\log n)$ bits. 

For (b), we consider $L=1,2,4,8,\ldots$ and compute absolute $L$-bit approximations $\tilde{y}_{0}$ of $y_{0}$ until we obtain an approximation $\tilde{y}_0$ with $\abs{\ty_0} \ge 2^{2-L}$. Since $\ty_0$ is an absolute $L$-bit approximation of $y_0$, $\abs{\ty_0 - y_0} \le 2^{-L} \le \abs{\ty_0}/4$. 
Since $\ty_0$ is a dyadic fraction, we can determine $t \in \Z$ with $\abs{t - \log \abs{\ty_0}} \le 1/2$. 
Then, $2^{t-1}   \le (3/4) 2^{-1/2} 2^t \le (3/4) \abs{\ty_0} \le \abs{y_0} \le (5/4) \abs{\ty_0} \le (5/4) 2^{1/2} 2^t \le 2^{t+1}$. 
Obviously, we succeed if $L\ge 2\log M(y_{0}^{-1})$, and since we double $L$ in each step, we need at most $O(\log\log M(y_{0}^{-1}))$ many steps. Up to logarithmic factors, the total cost is dominated by the cost of the last iteration, which is bounded by $\tilde{O}(n(\tau_{P}+n\log M(x_0)+\log M(y_{0}^{-1})))$ bit operations according to Part (a). 
\end{proof}

It has been shown~\cite{DBLP:journals/jc/Kirrinnis98,DBLP:journals/corr/abs-1304.8069,DBLP:journals/corr/KobelS14,Ritzmann19861,JvH2008} that the cost for approximately evaluating a polynomial of degree $n$ at $N=O(n)$ points is is of the same order as the cost of approximately evaluating it at a single point.

\begin{theorem}[\cite{DBLP:journals/jc/Kirrinnis98,DBLP:journals/corr/abs-1304.8069,DBLP:journals/corr/KobelS14}]\label{lem:multpolyeval}
Let $P$ be a polynomial as in (\ref{def:P}), let $x_{1},\ldots,x_{N}$ be real points with $N=O(n)$, and let $L$ be a positive integer. The algorithm in~\cite{DBLP:journals/jc/Kirrinnis98,DBLP:journals/corr/abs-1304.8069,DBLP:journals/corr/KobelS14} computes approximations $\tilde{y}_{i}$ of $y_{i}:=P(x_{i})$ with $|y_{i}-\tilde{y}_{i}|\le 2^{-L}$, $i=1,\ldots,N$, 
with 
\[
\tilde{O}(n(n+\tau_{P}+n\log M(\max_{i}|x_{i}|)+L))
\]
bit operations using approximations of the coefficients of $P$ as well as the points $x_{i}$ of quality $O(\tau_{P}+n\log M(\max_{i}|x_i|)+L+n\log n)$.
\end{theorem} 

We frequently need to select a point $x_i$ from a given set $X=\{x_{1},\ldots,x_{N}\}$ of points at which $|P(x_{i})|$ is close to maximal.

\begin{definition}\label{admissible point} Let $X:=\{x_{1},\ldots,x_{N}\}$ be a set of $N$ real points. We call a point $x^{*}\in X$ \emph{admissible with respect to $X$} (or just \emph{admissible} if there is no ambiguity) if $|P(x^{*})|\ge \frac{1}{4}\cdot\max_{i}|P(x_{i})|$.
\end{definition}

Fast approximate multipoint evaluation allows us to find an admissible point efficiently.\medskip

\begin{mdframed}[frametitle={{\bf Algorithm: Admissible Point}}]\label{alg:admissable point}
{\color{black}
\noindent{\bf Input:}
A polynomial $P(x)$ as in (\ref{def:P}), and a set $X = \{x_1,\ldots,x_N\}$ of points. \smallskip

\noindent{\bf Guarantee:} $\lambda:=\max_{i}|P(x_i)|>0$.\smallskip

\noindent{\bf Output:} An admissible point $x^*\in X$ and an integer $t$ with $2^{t-1}\le \lambda\le 2^{t+1}$. \medskip

\begin{itemize}
\item[(1)] $L:=1/2$.
\item[(2)] repeat
\begin{itemize}
\item[(2.1)] $L:=2\cdot L$ 
\item[(2.2)] Compute approximations $\tilde{\lambda}_i$ of quality $L$ 
for the values $P(x_i)$ for $1 \le i \le N$.
\item[(2.3)] $\tilde{\lambda}:=\max_{i}|\tilde{\lambda}_i|$.
\end{itemize} 
until  $\tilde{\lambda}\le 2^{2-L}$.
\item[(3)] Let $i_0$ be an index with $\tilde{\lambda}:=|\tilde{\lambda}_{i_0}|$ and $t$ be an integer with $|t-\log|\tilde{\lambda}||\le \frac{1}{2}$.
\item[(4)] return $x_{i_0}$ and $t$.
\end{itemize}}
\end{mdframed}
\bigskip

\ignore{
\hrule\medskip
\noindent{\bf Algorithm Admissible Point} (from a set $X = \{x_1,\ldots,x_N\}$):
We consider $L=1,2,4,8,\ldots$ and approximate all values $|P(x_{i})|$ with quality $L$ until, for at least one $i$, we obtain an approximation $2^{t_{i}}$ with $t_{i}\in\Z$ and $2^{t_{i}-1}\le |P(x_{i})|\le 2^{t_{i}+1}$. Now, let $i_{0}$ be such that $t_{i_{0}}$ is maximal; then, it follows that $2^{t_{i_{0}}-1}\le \lambda:=\max_{i}|P(x_{i})|\le 2^{t_{i_{0}}+1}$.
}
An argument similar to the one as in the proof of Part (b) in Lemma~\ref{lem:singlepolyeval} now yields the following result:

\begin{lemma}\label{lem:apxmultipointeval}
Let $X:=\{x_{1},\ldots,x_{N}\}$ be a set of $N=O(n)$ real points. The algorithm {\bf Admissible Point} applied to $X$ selects an admissible point $x^{*}\in X$ and an integer $t$ with 
\[
2^{t-1}\le |P(x^{*})|\le \lambda:=\max_{i}|P(x_{i})| \le 2^{t+1} 
\]
using $\tilde{O}(n(n+\tau_{P}+n\log M(\max_{i}|x_{i}|)+\log M(\lambda^{-1})))$ bit operations. It requires approximations of the coefficients of  $P$ and the points $x_{i}$ of quality $O(n+\tau_{P}+n\log M(\max_{i}|x_{i}|)+\log M(\lambda^{-1}))$. 
\end{lemma}

We will mainly apply the Lemma in the situation where $X$ is a set of $N=2\cdot \lceil n/2\rceil+1$ equidistant points. In this situation, we can prove a lower bound on 
$\lambda$ in the Lemma above in terms of the separations of the roots $z_{i}$, the absolute values of the derivatives $P'(z_{i})$, and the number of roots contained in a neighborhood of the points $X$. 

\begin{definition}\label{definition:multipoint}
For a real point $m$ and a real positive value $\epsilon$, the \emph{$(m,\epsilon)$-multipoint} $\multipoint{m}{\epsilon}$ is defined as
\begin{align}\label{def:multipoint}
\multipoint{m}{\epsilon} := \set{m_{i}:=m+(i-\lceil n/2\rceil)\cdot \epsilon}{i=0,\ldots,2\cdot \lceil n/2\rceil}.
\end{align}
\end{definition}

\begin{lemma}\label{lem:boundonmax} Let $m$ be a real point, let $\epsilon$ be a real positive value, and let $K$ be a positive real with $K\ge 2\cdot \ceil{n/2}$. If the disk $\Delta:=\Delta_{K\cdot \epsilon}(m)$ with radius $K\cdot \epsilon$ and center $m$ contains at least two roots of $P$, then each admissible point $m^{*}\in \multipoint{m}{\epsilon}$ satisfies
\[
|P(m^{*})|>2^{-4n-1}\cdot K^{-\mu(\Delta)-1}\cdot \sigma(z_i,P)\cdot |P'(z_i)|\quad\text{for all roots }z_i\in\Delta,
\] 
where $\mu(\Delta)$ denotes the number of roots of $P$ contained in $\Delta$.
\end{lemma}

\begin{proof}
Since the number of points $m_{i} \in \multipoint{m}{\epsilon}$ is larger than the number of roots of $P$ and since their pairwise distances are $\epsilon$, there exists a point $m_{i_{0}} \in \multipoint{m}{\epsilon}$ whose distance to all roots of $P$ is at least $\epsilon/2$. We will derive a lower bound on $\abs{P(m_{i_0})}$. 
Let $z_{i}$ be any root in $\Delta$. For any different root $z_{j}\in \Delta$, we have $|z_i-z_{j}|/|m_{i_0}-z_{j}|<2K\epsilon/(\epsilon/2)=4K$, and, for any root $z_{j}\notin\Delta$, we have $|z_i-z_{j}|/|m_{i_0}-z_{j}|\le 2K/(K-\lceil n/2\rceil)\le 4$. Hence, it follows that
\begin{align*}
|P(m_{i_0})|&=|P_{n}|\cdot |m_{i_0}-z_i|\cdot \prod_{j\neq i}|m_{i_{0}}-z_{j}|>|P_{n}|\cdot \frac{\epsilon}{2}\cdot (4K)^{-\mu(\Delta)}\cdot 4^{-n+\mu(\Delta)}\cdot \prod_{j\neq i}|z_i-z_{j}|\\
&=|P'(z_i)|\cdot \frac{\epsilon}{2n}\cdot 4^{-n}\cdot K^{-\mu(\Delta)}=2^{-(\log n+1)-2n}\cdot\epsilon\cdot K^{-\mu(\Delta)}\cdot |P'(z_i)|\\
&>2^{-2n-\log n-2}\cdot K^{-\mu(\Delta)-1}\cdot \sigma(z_i,P)\cdot |P'(z_i)|,
\end{align*}
where we used $\sigma(z_i,P)<2K\epsilon$. Hence, for each admissible point $m^{*}\in\multipoint{m}{\epsilon}$, it follows that $|P(m^{*})|\ge \frac{|P(m_{i_{0}})|}{4}\ge 2^{-4n-1}\cdot K^{-\mu(\Delta)-1}\cdot \sigma(z_i,P)\cdot |P'(z_i)|$.
\end{proof}

We summarize the discussion of this section in the following corollary.

 \begin{corollary}\label{cor:multipoint} Let $m$ be a real point, let $\epsilon$ be a real positive value, and let $K$ be a positive real with $K\ge 2\cdot \ceil{n/2}$, and assume that the disk $\Delta:=\Delta_{K\cdot \epsilon}(m)$ contains at least two roots of $P$. Then, for each admissible point $m^{*}\in \multipoint{m}{\epsilon}$,
  \begin{align}\label{sizeofsubdivpoint}
 |P(m^{*})|> 2^{-4n-1}\cdot K^{-\mu(\Delta)-1}\cdot \sigma(z_i,P)\cdot |P'(z_i)|\quad\text{for all roots }z_i\in\Delta.\end{align}
The algorithm {\bf Admissible Point} applied to $\multipoint{m}{\epsilon}$ selects an admissible point from the set with
 \begin{align}
\tilde{O}(n(\mu(\Delta)\cdot \log K+n+\tau_{P}+n\log M(|m|+n\epsilon)+\log M(\max_{z_i\in\Delta} (\sigma(z_i,P)\cdot |P'(z_i)|))).
\end{align}
bit operations. It requires approximations of the 
coefficients of $P$ and the points $m_{i}$ of quality $$O(\mu(\Delta)\cdot \log K+n+\tau_{P}+n\log M(|m|+n\epsilon)+\log M(\max_{z_i\in\Delta} (\sigma(z_i,P)\cdot |P'(z_i)|)).$$
 \end{corollary}

Corollary~\ref{cor:multipoint} is a key ingredient of our root isolation algorithm. We will appeal to it whenever we have to choose a subdivision point. Assume, in an ideal world with real arithmetic at unit cost, we choose a subdivision point $m$. The polynomial $P$ may take a very small value at $m$, and this would lead to a high bit complexity. Instead of choosing $m$ as the subdivision point, we choose a nearby admissible point $m^{*}\in\multipoint{m}{\epsilon}$ and are guaranteed that  $|P(m^*)|$ has at least the value stated in (\ref{sizeofsubdivpoint}). The fact that $|P|$ is reasonably large at $m^*$ will play a crucial role in the analysis of our algorithm, cf.~Theorem~\ref{main:theorem}.

\subsection{Descartes' Rule of Signs in Monomial and in Bernstein Basis}\label{sec:DRS:basis}

\begin{figure}[t]
\begin{center}
\includegraphics[width=0.65\textwidth]{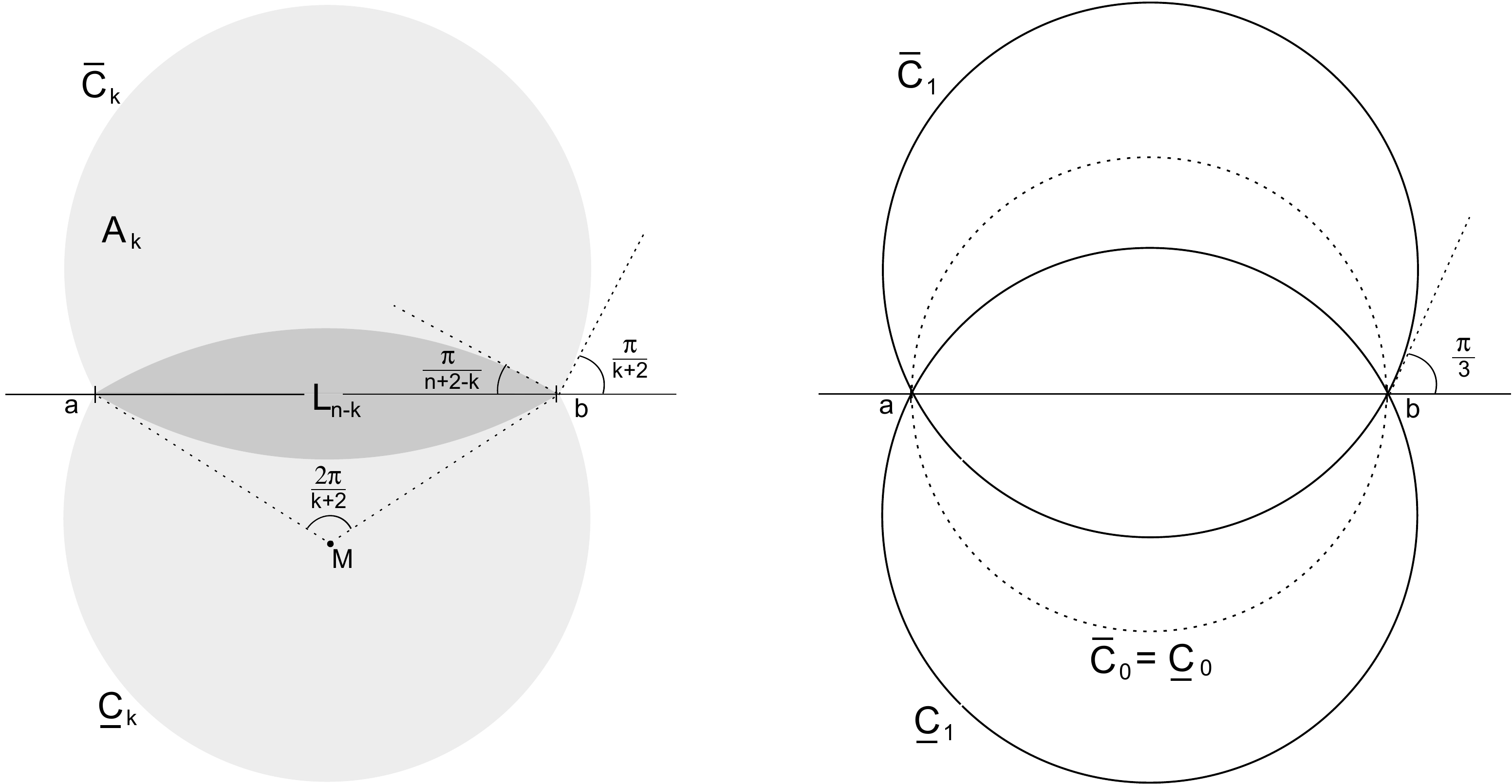}\end{center}
\caption{\label{fig:Obreshkoff}\label{page:Obreshkoff}For any $k$, $0 \le k \le n$, the \emph{Obreshkoff disks} 
$\overline{C}_k$ and $\underline{C}_k$ for $I=(a,b)$ have the endpoints of $I$ on their
boundaries; their centers see the line segment $(a,b)$ under the angle $\frac{2\pi}{k+2}$.  
The \emph{Obreshkoff lens} $L_k$ is the interior of $\overline{C}_k \cap
\underline{C}_k$, and the \emph{Obreshkoff area} $A_k$ is the interior of $\overline{C}_k \cup
\underline{C}_k$. Any point (except $a$ and $b$) on the boundary of $A_k$ sees $I$ under the
angle $\frac{\pi}{k+2}$, and any point (except
$a$ and $b$) on the boundary of
$L_k$ sees $I$ under the angle $\pi - \frac{\pi}{k+2}$. We have $L_n \subseteq \ldots \subseteq L_1 \subseteq L_0$ and $A_0
\subseteq A_1 \subseteq \ldots \subseteq A_n$. The cases $k=0$ and $k=1$ are of special interest: The circles
$\overline{C}_0$ and $\underline{C}_0$ coincide. They have their centers at the
midpoint of $I$. The circles $\overline{C}_1$
and $\underline{C}_1$ are the circumcircles of the two equilateral triangles
having $I$ as one of their edges. We call $A_0=\Delta(I)$ and $A_1$ the \emph{one-circle} and the \emph{two-circle region} for $I$, respectively.}
\end{figure}

This section provides a brief review of Descartes' rule of signs. We remark that most of what follows in this section has already been presented (in more detail) elsewhere (e.g.~in~\cite{eigenwillig-phd,DBLP:conf/casc/EigenwilligKKMSW05,NewDsc}); however, for the sake of a self-contained representation, we have decided to reiterate the most important results which are needed for our algorithm and its analysis.

In order to estimate the number $m_{I}$ of roots of $P$ contained in an interval $I=(a,b)\subseteq \mathcal{I}=(-2^{\Gamma},2^{\Gamma})$, we use Descartes' rule of signs: \textit{For a polynomial $F(x)=\sum_{i=0}^N f_i x^i\in\R[x]$, the number $m$ of positive real roots of $F$ is bounded by the number $v$ of sign variations\footnote{Zero entries are not considered. For instance, $\var(-1,0,0,2,0,-1)=\var(-1,2,-1)=2$.} in its coefficient sequence $(f_0,\ldots,f_N)$ and, in addition, $v\equiv m \text{ }\operatorname{mod}\text{ }2$}. We can apply this rule to the polynomial $P$ and the interval $I$ by considering a M\"obius transformation $x\mapsto \frac{ax+b}{x+1}$ that maps $(0,+\infty)$ one-to-one onto $I$. Namely, let
\begin{align}
P_I(x):=\sum_{i=0}^n p_{I,i} \cdot x^i:=(x+1)^n\cdot P\left(\frac{ax+b}{x+1}\right),\label{poly:PI}
\end{align}
and let $v_{I}:=\var(P,I):=\var(p_{I,0},\ldots,p_{I,n})$ be defined as \emph{the number of sign variations in the coefficient sequence $(p_{I,0},\ldots,p_{I,n})$ of 
$P_I$}. Then,
$v_{I}$ is an upper bound for $m_{I}$ (i.e.~$v_{I}\ge m_{I}$) and $v_{I}$ has the same parity as $m_{I}$ (i.e.~$v_{I}\equiv m_{I} \text{ }\operatorname{mod}\text{ }2$). Notice that the latter two properties imply that $v_{I}=m_{I}$ if $v_{I}\le 1$.

The following theorem states that the number $v_{I}$ is closely related to the number of roots located in specific neighborhoods of the interval $I$.

\begin{theorem}[\cite{Obreschkoff:book,Obrechkoff:book-english}]
\label{Obreshkoff} Let $I=(a,b)$ be an open interval and
$v_{I}= \var(P,I)$. If the Obreshkoff lens 
$L_{n - k}$ (see Figure~\ref{fig:Obreshkoff} for the definition of $L_{n - k}$) contains at least $k$ roots (counted with
multiplicity) of $P$, then $v _{I}\ge k$. If the Obreshkoff area $A_k$ contains no more than $k$ roots (counted with multiplicity) of $P$, then $v_I \le k$. In
particular,
\[
\#\text{ of roots of }P\text{ in }L_n\le    v_{I}=\var(P,I) \le \#\text{ of roots
of }P\text{ in }A_n.
\]
\end{theorem}

We remark that the special cases $k=0$ and $k=1$ appear as the one-circle and the two-circle theorems in the literature (e.g.~\cite{Alesina-Galuzzi,eigenwillig-phd}). Theorem~\ref{Obreshkoff} implies that  if the one-circle region $A_{0}=\Delta(I)$ of $I$ contains a root $z_{i}$ with separation $\sigma(z_i,P)>2w(I)=2(b-a)$, then this root must be real and $v_{I}=1$. Namely, the condition on $\sigma(z_i,P)$ guarantees that the two-circle region $A_{1}$ contains $z_{i}$ but no other root of $P$. If the one-circle region contains no root, then $v_I=0$. Hence, it follows that each interval $I$ of width $w(I)<\sigma_P/2$ yields $v_{I}=0$ or $v_{I}=1$. 
In addition, we state the variation diminishing property of the function $\var(P,I)$; e.g.,~see~\cite[Corollary 2.27]{eigenwillig-phd} for a self-contained proof:

\begin{theorem}[\cite{Schoenberg}]\label{subad}
Let $I$ be an interval and $I_1$ and $I_2$ be two disjoint subintervals of $I$. Then,
\[\var(P,I_1) + \var(P,I_2) \le \var(P,I).
\]
\end{theorem}

In addition to the above formulation of Descartes' rule of signs in the monomial basis, we provide corresponding results for the representation of $P(x)$ in terms of the Bernstein basis $B_{0}^{n},\ldots,B_{n}^{n}$ with respect to $I=(a,b)$, where 
\begin{align}
B_{i}^{n}(x):=B_{i}^{n}[a,b](x):=\binom{n}{i}\frac{(x-a)^{i}(b-x)^{n-i}}{(b-a)^{n}},\quad 0\le i\le n.
\end{align}
If $P(x)=\sum_{i=0}^{n}b_{i}B_{i}^{n}[a,b](x)$, we call $B=(b_{0},\ldots,b_{n})$ the Bernstein 
representation of $P$ with respect to $I$. For the first and the last coefficient, we have $b_{0}=P(a)$ 
and $b_{n}=P(b)$. The following Lemma provides a direct correspondence between the coefficients 
of the polynomial $P_{I}$ from (\ref{poly:PI}) and the entries of $B$. For a self-contained proof, we refer to~\cite{eigenwillig-phd}.

\begin{lemma}\label{lem:MonomBernstein}
Let $I=(a,b)$ be an interval, $P(x)=\sum_{i=0}^{n}b_{i}B_{i}^{n}[a,b](x)$ be the Bernstein representation of $P$ with respect to $I$, and $P_{I}(x)=\sum_{i=0}^n p_{I,i} \cdot x^i$ as in (\ref{poly:PI}). It holds that 
\begin{align}
p_{I,i}=b_{n-i}\cdot\binom{n}{i}\quad\text{for all }i=0,\ldots,n.
\end{align}
In particular, $v_{I}$ coincides with the number of sign variations in the sequence $(b_{0},\ldots,b_{n})$.
\end{lemma}

In essence, the above lemma states that, when using Descartes' Rule of Signs, it makes no difference whether we consider the Bernstein basis representation of $P$ with respect to $I$ or the 
polynomial $P_{I}$ from (\ref{poly:PI}). This will turn out to be useful in the next section, where we review results from~\cite{DBLP:conf/casc/EigenwilligKKMSW05} which allow us to treat the cases $v_{I}=0$ and $v_{I}=1$ by  using approximate arithmetic.

\subsection{Descartes' Rule of Signs with Approximate Arithmetic}\label{sec:apxdsc}

We introduce the $0$-Test and $1$-Test for intervals $I$ with the following properties.\smallskip

\begin{compactenum}[(1)]
\item If $\var(P,I) = 0$ ($\var(P,I) = 1$), then the $0$-Test ($1$-Test) for $I$ succeeds.
\item If the $0$-Test ($1$-Test) for $I$ succeeds, $I$ contains no (exactly one) root of $P$.
\item The $0$-Test and the $1$-Test for $I$ can be carried out efficiently with approximate arithmetic, see Corollaries~\ref{cor:cost0test} and~\ref{cor:cost1test}. 
\end{compactenum}

\subsubsection{The case $\var(P,I)=0$}\label{var0apx}
Consider the following Lemma, which follows directly from~\cite[Lemma 5]{DBLP:conf/casc/EigenwilligKKMSW05} and its proof:

\begin{lemma}\label{lem:bitstreamresult1}
Let  $P(x)=\sum_{i=0}^{n}b_{i}B_{i}^{n}[a,b](x)$ be the Bernstein representation of $P$ with respect to the interval $I=(a,b)$, and let $m$ be a subdivision point contained in $[m(I)-\frac{w(I)}{4},m(I)+\frac{w(I)}{4}]$. The Bernstein representations of $P$ with respect to $I'=(a,m)$ and $I''=(m,b)$ are given by $P(x)=\sum_{i=0}^{n}b_{i}'B_{i}^{n}[a,m](x)$ and $P(x)=\sum_{i=0}^{n}b_{i}''B_{i}^{n}[m,b](x)$, respectively. 
Suppose that $\var(P,I)=0$. Then, $\var(P,I')=\var(P,I'')=0$, and
\[
|b_{i}'|,|b_{i}''|>\min(|P(a)|,|P(b)|)\cdot 4^{-(n+1)}\quad\text{for all }i=0,\ldots,n. 
\]
\end{lemma}

Combining the latter result with Lemma~\ref{lem:MonomBernstein} now yields:

\begin{corollary}\label{cor:var0test}
Let $I$, $I'$, $I''$ be intervals as in Lemma~\ref{lem:bitstreamresult1}, and let $L$ be an integer with
\begin{align}\label{def:LI0}
L\ge L_{I,0}:=\log M(\min(|P(a)|,|P(b)|)^{-1})+ 2(n+1)+1.
\end{align}
Suppose that $\sum_{i=0}^n \tilde{p}_{I',i} \cdot x^i$ and $\sum_{i=0}^n \tilde{p}_{I'',i} \cdot x^i$ are absolute 
$L$-bit approximations of $P_{I'}$ and $P_{I''}$, respectively. 
If $\var(P,I)=0$, then $\var(\tilde{p}_{I',0},\ldots,\tilde{p}_{I',n})=\var(\tilde{p}_{I'',0},\ldots,\tilde{p}_{I'',n})=0$, and $|\tilde{p}_{I',i}|,|\tilde{p}_{I'',i}|>2^{-L}$ for all $i=0,\ldots,n$.
\end{corollary}

\begin{proof}
Suppose that $\var(P,I)=0$, then Lemma~\ref{lem:bitstreamresult1} yields $|b_{i}'|,|b_{i}''|>\min(|P(a)|,|P(b)|)\cdot 4^{-(n+1)}\ge 2\cdot 2^{-L_{I,0}}$ 
for all $i=0,\ldots,n$, and, in addition, all coefficients $b_{i}'$ and $b_{i}''$ have the same sign. Since 
$p_{I',i}=\binom{n}{i}b'_{n-i}$ and $p_{I'',i}=\binom{n}{i}b''_{n-i}$, it follows that the coefficients 
$p_{I',i}$ and $p_{I'',i}$ also have absolute value larger than $2\cdot 2^{-L_{I,0}}$. Thus, $|\tilde{p}_{I',i}|,|\tilde{p}_{I'',i}|>2^{-L_{I,0}}$ since 
$|p_{I',i}-\tilde{p}_{I',i}|\le 2^{-L_{0}}$ and $|p_{I'',i}-\tilde{p}_{I'',i}|\le 2^{-L}\le 2^{-L_{I,0}}$ for all $i$. In addition, all coefficients 
$\tilde{p}_{I',i}$ and $\tilde{p}_{I'',i}$ have the same sign because this holds for their exact counterparts.
\end{proof}

The above corollary allows one to discard an interval $I$ by using approximate arithmetic with a 
precision directly related to the absolute values of $P$ at the endpoints of $I$. More precisely, we consider the following exclusion test that applies to intervals $I=(a,b)$ with $P(a)\neq 0$ and $P(b)\neq 0$. \MS{Comments explaining the rationale behind our choices are typeset in italic and start with the symbol $\cosym$.}\medskip

\begin{mdframed}[frametitle={{\bf Algorithm: $\mathbf{0}$-Test}}]
{\color{black}
\noindent{\bf Input:} A polynomial $P(x)$ as in (\ref{def:P}) and an interval $I=(a,b)$ with $P(a)\neq 0$ and $P(b)\neq 0$.\smallskip

\noindent{\bf Output:} True or False. In case of True, $I$ contains no root of $P$. In case of False, it holds that $\var(P,I)>0$.\smallskip

\begin{itemize}
\item[(1)] Compute integers $t_a$ and $t_b$ with $2^{t_a-1}\le |P(a)|\le 2^{t_a+1}$ and $2^{t_b-1}\le |P(b)|\le 2^{t_b+1}$ using the algorithm {\bf Admissible Point} with input $X=\{a\}$ and $X=\{b\}$, respectively.
\item[ ]\noindent$\cosym$\textit{Notice that, according to Lemma~\ref{lem:singlepolyeval}, we can compute $t_{a}$ and $t_{b}$ with a number of bit operations bounded by $\tilde{O}(n(\tau_{P}+n\log M(a)+n\log M(b)+\log M(P(a)^{-1})+\log M(P(b)^{-1})))$.}
\item[(2)] $L:=M(-\min(t_{a}-1,t_{b}-1))+2(n+1)+1$
\item[ ]\noindent$\cosym$\textit{From the definition of $L_{I,0}$ in (\ref{def:LI0}), it follows that $L_{I,0}\le L\le L_{I,0}+2$.} 
\item[(3)] $I':=(a,m(I))$ and $I'':=(m(I),b)$  
\item[(4)] Compute absolute $L$-bit approximations $\tilde{P}_{I'}=\sum_{i=0}^n \tilde{p}_{I',i} \cdot x^i$ and $\tilde{P}_{I''}:=\sum_{i=0}^n \tilde{p}_{I'',i} \cdot x^i$ of the polynomials $P_{I'}$ and $P_{I''}$, respectively.
\item[ ]\noindent$\cosym$\textit{For an efficient solution of this step, consider the algorithm from the proof of Lemma~\ref{lem:computingpI}.}
\item[(5)] If $\var(\tilde{p}_{I',0},\ldots,\tilde{p}_{I',n})=\var(\tilde{p}_{I'',0},\ldots,\tilde{p}_{I'',n})=0$, and $|\tilde{p}_{I',j}|>2^{-L}$ and $|\tilde{p}_{I'',j}|>2^{-L}$ for all $j=0,\ldots,n$, then return True.
\item[ ]\noindent$\cosym$\textit{We conclude that $\var(p_{I',0},\ldots,p_{I',n})=\var(p_{I'',0},\ldots,p_{I'',n})=0$. Since $\tilde{p}_{I'',0}=P(m(I))$, we also have $P(m(I))\neq 0$, and thus $I$ contains no root of $P$.}
\item[(6)] return False.
\end{itemize}}
\end{mdframed}
\bigskip
\ignore{
\hrule\medskip
\noindent{\bf Algorithm $0$-Test:}
Compute approximations $2^{t_{a}}$ and $2^{t_{b}}$ for $|P(a)|$ and $|P(b)|$ with $t_{a},t_{b}\in\Z$, and $2^{t_a-1}\le |P(a)|\le 2^{t_{a}+1}$, and $2^{t_b-1}\le |P(b)|\le 2^{t_{b}+1}$.\footnote{According to Lemma~\ref{lem:singlepolyeval}, we can compute $t_{a}$ and $t_{b}$ with a number of bit operations bounded by $\tilde{O}(n(\tau_{P}+n\log M(a)+n\log M(b)+\log M(P(a)^{-1})+\log M(P(b)^{-1})))$.} 
It follows that $L_{I,0}\le L:=M(-\min(t_{a}-1,t_{b}-1))+2(n+1)+1\le L_{I,0}+2$. Now, for $I'=(a,m(I))$ and $I''=(m(I),b)$, compute absolute $L$-bit
approximations $\tilde{P}_{I'}=\sum_{i=0}^n \tilde{p}_{I',i} \cdot x^i$ and $\tilde{P}_{I''}:=\sum_{i=0}^n \tilde{p}_{I'',i} \cdot x^i$ of the polynomials $P_{I'}$ and $P_{I''}$, respectively. If all approximate coefficients $\tilde{p}_{I',i}$ and $\tilde{p}_{I'',i}$ have the same sign (i.e.~$\var(\tilde{p}_{I',0},\ldots,\tilde{p}_{I',n})=\var(\tilde{p}_{I'',0},\ldots,\tilde{p}_{I'',n})=0$) and if all of them have absolute 
value larger then $2^{-L}$, then $\var(P,I')=\var(P,I'')=0$ and, thus, $I$ 
contains no root of $P$.\footnote{Notice that $P(m)\neq 0$ because $|P(m)|=|p_{I',n}|>0$.} In this case, we say that the $0$-Test succeeds.\medskip
\hrule\bigskip}

 It remains to provide an efficient method to compute an absolute $L$-bit approximation of a polynomial $P_{I}$ as required in the $0$-Test:

\begin{lemma}\label{lem:computingpI}
Let $I=(a,b)$ be an interval, and let $L$ be a positive integer. The algorithm stated in the proof computes an absolute $L$-bit approximation $\tilde{P}_{I}(x)=\sum_{i=0}^{n}\tilde{p}_{I,i}$ of $P_{I}(x)=\sum_{i=0}^{n}p_{I,i}x^{i}$ with
\[
\tilde{O}(n(n+\tau_{P}+n\log M(a)+n\log M(b)+L))
\]
bit operations. It requires approximations of the coefficients of $P$ and the endpoints of $I$ of 
quality $O(n+\tau_{P}+n\log M(a)+n\log M(b)+L)$.
\end{lemma}

\begin{proof}
The computation of $P_{I}$ decomposes into four steps: First, we substitute $x$ by $a+x$, which yields the 
polynomial $P_{1}(x):=P(a+x)$. Second, we substitute $x$ by $w(I)\cdot x$ in order to 
obtain $P_{2}(x):=P_{1}(a+w(I)\cdot x)$. Third, the coefficients of $P_{2}$ are reversed (i.e.~the $i$-th 
coefficient is replaced by the $(n-i)$-th coefficient), which yields the polynomial $P_{3}(x)=x^{n}P_{2}(1/x)=x^{n}P(a+w(I)/x)$. In the last 
step, we compute the polynomial $P_{4}(x):=P_{3}(x+1)=(x+1)^{n}P(a+w(I)/(x+1))=P_{I}(x)$.

Now, for the computation of an absolute $L$-bit approximation $\tilde{P}_{I}$, we proceed as follows: Let $L_{1}$ be a positive integer, which will be specified later. 
According to~\cite[Theorem 14]{DBLP:journals/corr/abs-1304.8069} (or~\cite[Theorem 8.4]{schonhage:fundamental}), we can compute an absolute 
$L_{1}$-bit approximations $\tilde{P}_{1}$ of $P_{1}$ with $\tilde{O}(n(n+\tau_{P}+n\log M(a)+L_{1}))$ bit 
operations, where we used that the coefficients of $P$ have absolute value of size $2^{\tau_{P}}$ or less.
For this step, the coefficients of $P$ as well as the endpoint $a$ have to be approximated with quality $\tilde{O}(n+\tau_{P}+n\log M(a)+L_{1})$.
The coefficients of $P_{1}$ have absolute value less than $2^{n+\tau_{P}}M(a)^{n}$, and thus, the coefficients of $\tilde{P}_{1}$ have absolute value less than $2^{n+\tau_{P}}M(a)^{n}+1<2^{n+1+\tau_{P}}M(a)^{n}$. Computing $w(I)^i$ for all $i=0,\ldots,n$ with quality $L_1$ takes $\tilde{O}(n(n\log M(w(I))+L_1))=\tilde{O}(n(n\log M(a)+n\log M(b)+L_1))$ bit operations. This yields an approximation $\tilde{P}_2$ of $P_2$ with quality $L_2:=L_1-n-1-\tau_P-n \log M(a)$.

The coefficients of $\tilde{P_{2}}$ have absolute value less than $2^{n+1+\tau_{P}}M(a)^{n}M(w(I))^{n}$. 
Reversing the coefficients of $\tilde{P}_{2}$ trivially yields an absolute $L_{2}$-bit approximation $\tilde{P}_{3}$ of $P_{3}$. 
For the last step, we again apply~\cite[Theorem 14]{DBLP:journals/corr/abs-1304.8069} to show that we can compute an absolute $L$-bit approximation of $P_{4}=P_{3}(x+1)$ from an $L_{3}$-bit 
approximation of $P_{3}$, where $L_{3}$ is an integer of size $\tilde{O}(L+n+\tau_{P}+n\log M(a)+n\log M(w(I)))$. The cost for this computation is bounded by $\tilde{O}(nL_{3})$ bit operations. Hence, it suffices to start with an integer
$L_{1}$ of size $\tilde{O}(L+n+\tau_{P}+n\log M(a)+n\log M(w(I)))$. This shows the claimed bound for the needed input precision, where we use that $w(I)\le |a|+|b|$.
The bit complexity for each of the two Taylor shifts (i.e.~$x\mapsto a+x$ and $x\mapsto x+1$) as 
well as for the approximate scaling (i.e.~$x\mapsto w(I)\cdot x$) is bounded by $\tilde{O}(n(n+\tau_{P}+n\log M(a)+n\log M(b)+L))$ bit operations.
\end{proof}

The above lemma (applied to the intervals $I'=(a,m(I))$ and $I''=(m(I),b)$) now directly yields a bound on the bit complexity for the $0$-Test:

\begin{corollary}\label{cor:cost0test}
For an interval $I=(a,b)$, the $0$-Test requires
\begin{align}\label{bound:confirm0}
\tilde{O}(n(n+\tau_{P}+n\log M(a)+n\log M(b)+\log M(\min(|P(a)|,|P(b)|)^{-1}))
\end{align}
 bit operations using approximations of the coefficients of $P$ and the endpoints of $I$ of quality $O(n+\tau_{P}+n\log M(a)+n\log M(b)+\log M(\min(|P(a)|,|P(b)|)^{-1}))$. 
 \end{corollary}

 \subsubsection{The case $\var(P,I)=1$}

We  need the following result, which follows directly from~\cite[Lemma 6]{DBLP:conf/casc/EigenwilligKKMSW05} and its proof.
 
 \begin{lemma}\label{lem:bitstreamresult2}

With the same definitions as in Lemma~\ref{lem:bitstreamresult1}, suppose that $\var(P,I)=1$ and $P(m) \not= 0$. Then,
\[
|b_{i}'|,|b_{i}''|>\min(|P(a)|,|P(b)|,|P(m)|)\cdot 16^{-n}\quad\text{for all }i=0,\ldots,n. 
\]
Furthermore, $\var(P,I')=1$ (and $\var(P,I'')=0$) or $\var(P,I'')=1$ (and $\var(P,I')=0$).
\end{lemma}

Again, combining the latter result with Lemma~\ref{lem:MonomBernstein} yields the following result, whose proof is completely analogous to the proof of Corollary~\ref{cor:var1test}.

\begin{corollary}\label{cor:var1test}
With the same definitions as in Lemma~\ref{lem:bitstreamresult1} and Lemma~\ref{lem:bitstreamresult2}, \MS{let $L$ be an integer with
\begin{align}
L\ge L_{I,1}:=\log M(\min(|P(a)|,|P(b)|,|P(m)|)^{-1})+ 4n+1,
\end{align}
and let $\sum_{i=0}^n \tilde{p}_{I',i} \cdot x^i$ and $\sum_{i=0}^n \tilde{p}_{I'',i} \cdot x^i$ be absolute 
$L$-bit approximations of $P_{I'}$ and $P_{I''}$, respectively.} 
Suppose that $\var(P,I)=1$ and $P(m) \not= 0$. \MS{Then, it follows that $|\tilde{p}_{I',i}|,|\tilde{p}_{I'',i}|>2^{-L}$} for all $i=0,\ldots,n$, and, in addition, $\var(\tilde{p}_{I',0},\ldots,\tilde{p}_{I',n})=1$ (and $\var(\tilde{p}_{I'',0},\ldots,\tilde{p}_{I'',n})=0$) or $\var(\tilde{p}_{I',0},\ldots,\tilde{p}_{I',n})=1$ (and $\var(\tilde{p}_{I'',0}\ldots,\tilde{p}_{I'',n})=0$).
\end{corollary}
Based on the above Corollary, we can now formulate the \emph{$1$-Test}, which applies to intervals $I=(a,b)$ with $P(a)\neq 0$ and $P(b)\neq 0$:\medskip

\begin{mdframed}[frametitle={{\bf Algorithm: $\mathbf{1}$-Test}}]
{\color{black}
\noindent{\bf Input:}
A polynomial $P(x)$ as in (\ref{def:P}) and an interval $I=(a,b)$ with $P(a)\neq 0$ and $P(b)\neq 0$.\smallskip

\noindent{\bf Output:} True or False. In case of True, the algorithm also returns an interval $I'$, with $I'\subset I$, and such that $\frac{1}{4}\cdot w(I)\le w(I')\le \frac{3}{4}\cdot w(I)$, $\var(f,I)=1$, and $I\backslash I'$ contains no root. In case of False, it holds that $\var(P,I)\neq 1$.\smallskip
\begin{itemize}
\item[(1)] Compute integers $t_a$ and $t_b$ with $2^{t_a-1}\le |P(a)|\le 2^{t_a+1}$ and $2^{t_b-1}\le |P(b)|\le 2^{t_b+1}$ using the algorithm {\bf Admissible Point}.
\item[(2)] For $\epsilon:=w(I)\cdot 2^{-\ceil{\log n+2}}\le \frac{w(I)}{4n}$, compute an admissible point $m\in \multipoint{m(I)}{\epsilon}$ and an integer $t$ with $2^{t+1}\ge \max_{i}|P(m_{i})|\ge |P(m^*)|\ge 2^{t-1}$ using the algorithm {\bf Admissible Point}.
\item[(3)] $L:=M(-\min(t_{a}-1,t_{b}-1,t-1))+4n+2$
\item[(4)] $I':=(a,m)$ and $I'':=(m,b)$
\item[ ]\noindent$\cosym$\textit{Notice that each point from $\multipoint{m(I)}{\epsilon}$ is contained in $[m(I)-\frac{w(I)}{4},m(I)+\frac{w(I)}{4}]$ since $\lceil n/2\rceil\cdot 2^{-\lceil \log n+2 \rceil}<\frac{1}{4}$. In addition, we have $L_{I,1}\le L \le L_{I,1}+2$. Thus, we can use Corollary~\ref{cor:var1test}.}
\item[(5)] Compute absolute $L$-bit approximations $\tilde{P}_{I'}=\sum_{i=0}^n \tilde{p}_{I',i} \cdot x^i$ and $\tilde{P}_{I''}:=\sum_{i=0}^n \tilde{p}_{I'',i} \cdot x^i$ of the polynomials $P_{I'}$ and $P_{I''}$, respectively.
\item[(6)] If $|\tilde{p}_{I',j}|\le 2^{-L}$ or $|\tilde{p}_{I'',j}|>2^{-L}$ for some $j=0,\ldots,n$ then return False
\item[ ]\noindent$\cosym$\textit{Corollary~\ref{cor:var1test} implies that $\var(P,I)\neq 1$.}
\item[(7)] If $\var(\tilde{p}_{I',0},\ldots,\tilde{p}_{I',n})=1$ and $\var(\tilde{p}_{I'',0},\ldots,\tilde{p}_{I'',n})=0$ then return True and $I'$
\item[ ]\noindent$\cosym$\textit{Corollary~\ref{cor:var1test} implies that $\var(P,I')=1$ and $\var(P,I'')=0$.}
\item[(8)] If $\var(\tilde{p}_{I',0},\ldots,\tilde{p}_{I',n})=0$ and $\var(\tilde{p}_{I'',0},\ldots,\tilde{p}_{I'',n})=1$ then return True and $I''$ 
\item[ ]\noindent$\cosym$\textit{Corollary~\ref{cor:var1test} implies that $\var(P,I')=0$ and $\var(P,I'')=1$.}
\item[(9)] return False
\item[ ]\noindent$\cosym$\textit{Corollary~\ref{cor:var1test} implies that $\var(P,I)\neq 1$.}
\end{itemize}
}
\end{mdframed}
\bigskip

In completely analogous manner as for the $0$-Test, we can estimate the cost for the $1$-Test:

\begin{corollary}\label{cor:cost1test}
Let $I=(a,b)$ be an interval, let $\multipoint{m(I)}{\epsilon}$ be the multipoint defined in the $1$-Test, 
and let $\lambda \assign \max \set{|P(x)|}{x \in \multipoint{m(I)}{\epsilon}}$. Then, the $1$-Test applied to $I$ requires
\begin{align}\label{bound:confirm1}
\tilde{O}(n(n+\tau_{P}+n\log M(a)+n\log M(b)+\log M(\min(|P(a)|,|P(b)|,\lambda)^{-1})))
\end{align}
bit operations using approximations of the coefficients of $P$ and the endpoints of $I$ of quality $O(n+\tau_{P}+n\log M(a)+n\log M(b)+\log M(\min(|P(a)|,|P(b)|,\lambda)^{-1}))$.
 \end{corollary}

\subsection{Useful Inequalities}

\begin{lemma}\label{lem:usefulineq}
Let $p = \sum_{0 \le i \le n} p_ix^i = p_n \prod_{1 \le i \le n} (x- z_i) \in \R[z]$. Then, 
\begin{align}
\Mea(p) &\le \norm{p}_2 \le \norm{p}_1 \le (n+1) 2^{\tau_p}  \label{upper bound on measure}\\
\sigma_p &\ge \sqrt{|\Disc(p)|} \norm{p}_2^{-n + 1} n^{-(n+2)/2} \label{lower bound on separation}\\
\abs{\Disc(p)} &\le n^n (\Mea(p))^{2n - 2} \le n^n \norm{p}_2^{2n - 2}\label{upper bound on discriminant}\\
\log |p'(z_i)|  &=O(\log n+\tau_{p}+n\log M(z_{i})) \label{upper bound on derivative at root}\\
\sum_i\log M(p'(z_i)^{-1}) &= O(n \tau_p + n^2 + n \log \Mea(p) + |\log \Disc(p)|^{-1}) \label{upper bound on sum of inverse derivatives}\\
\Mea(p(x - z_i)) &\le 2^{\tau_p} 2^{n+1} M(z_i)^n \label{measure of shifted polynomial}\\
\tau_p &= O(n + \log \Mea(p))    \label{taup in terms of measure}.
\end{align}
\end{lemma}
\begin{proof} \cite[Lemma 4.14]{yap-fundamental} establishes (\ref{upper bound on measure}). \cite[Corollary 6.29]{yap-fundamental} establishes (\ref{lower bound on separation}) and (\ref{upper bound on discriminant}). For (\ref{upper bound on derivative at root}), observe 
\[  \log |p'(z_i)|  = \log ( \sum_{1 \le k \le n} \abs{p_k k z_i^{k-1}}) \le \log (n \cdot 2^{\tau_p} n M(z_i)^n) =O(\log n + \tau_{p}+n\log M(z_{i})).\]
(\ref{upper bound on sum of inverse derivatives}) follows from
\begin{align*}
\sum_i\log M(p'(z_i)^{-1}) &= \log \prod_i  M(p'(z_i)^{-1}) 
                                        = \log \frac{\prod_i M(p'(z_i))}{\prod_i |p'(z_i)|}\\
                                        &= \log \frac{ \abs{p_n}^{n-2} \prod_i M(p'(z_i))}{|\Disc(p)|}
                                        = O(n \tau_p + n^2 + n \log \Mea(p) + \log |\Disc(p)|^{-1}).
\end{align*}
For (\ref{measure of shifted polynomial}), we use $\Mea(p) \le \norm{p}_1$ and $p(x - z_i) = \sum_{0 \le k \le n} ( x^k \sum_{k \le j \le n} p_j {j \choose k} (-z_i)^{j - k})$, and hence, 
\begin{align*}
\norm{p(x - z_i)}_1 &\le \sum_{0 \le k \le n} \sum_{k \le j \le n} p_j {j \choose k} M(z_i)^{j - k} 
\le 2^{\tau_p} M(z_i)^n \sum_{0 \le j \le n} \sum_{k \le j} {j \choose k} \\
&\le 2^{\tau_p} M(z_i)^n \sum_{0 \le j \le n} 2^j \le 2^{\tau_p} M(z_i)^n 2^{n+1}. 
\end{align*}
For (\ref{taup in terms of measure}), we first recall that $\tau_P = \log M(\norm{p}_\infty)$. The coefficient $p_i$ is given by
\[    p_i = p_n \cdot\sum_{I \subseteq \sset{1,\ldots,n}, \abs{I} = n - i} \prod_{i \in I} z_i .\]
Thus 
\[ \abs{p_i} \le \abs{p_n} {n \choose n - i} \frac{\Mea(p)}{\abs{p_n}} \le 2^n \Mea(p).\]
\end{proof}

\section{The Algorithm}\label{sec:Algorithm}

We are now ready for our algorithm $\AD$\footnote{Our algorithm is an \emph{approximate arithmetic} variant of the algorithm \textsc{New}\textsc{Dsc} presented in~\cite{NewDsc}. 
\textsc{New}\textsc{Dsc} combines the classical Descartes method and Newton iteration. It uses exact rational arithmetic and only applies to polynomials with rational coefficients. Pronounce $\AD$ as either ``approximate arithmetic Newton-Descartes'' or ``a new Descartes''.} for isolating the real roots of $P$.
We maintain a list $\mathcal{A}$ of active intervals,\footnote{In fact, $\mathcal{A}$ is a list of pairs $(I,N_I)$, where $I$ is an interval and $N_I\in\N$ a power of two. For the high level introduction, the reader may think of $\mathcal{A}$ of a list of intervals only.} a list $\mathcal{O}$ of isolating intervals, and the invariant that the intervals in $\mathcal{O}$ are isolating and that each real root of $P$ is contained in either an active or an isolating interval. We initialize $\mathcal{O}$ to the empty set and $\mathcal{A}$ to the interval $\mathcal{I}=(-2^{\Gamma},2^{\Gamma})$, where $\Gamma=2^\gamma$ is defined as in (\ref{def:Gamma}). This interval contains all real roots of $P$. Our actual initialization procedure is more complicated, see Section~\ref{sec:Initialization}, but this is irrelevant for the high level introduction to the algorithm. 

In each iteration, we work on one of the active intervals, say $I$. We first apply the $0$-Test and the $1$-Test to $I$; see Section~\ref{sec:apxdsc} for a discussion of these tests.  If the $0$-Test succeeds, we discard $I$. This is safe, as a successful $0$-Test implies that $I$ contains no real root. If the $1$-Test succeeds, we add $I$ to the set of isolating intervals. This is safe, as a successful $1$-Test implies that $I$ contains exactly one real root. If neither $0$- or $1$-Test succeeds, we need to subdivide $I$. 

Classical bisection divides $I$ into two equal or nearly equal sized subintervals. This works fine, if the roots contained in $I$ spread out nicely, as then a small number of subdivision steps suffices to separate the roots contained in $I$. This works poorly if the roots contained in $I$ form a cluster of nearby roots, as then a larger number of subdivision steps are needed until $I$ is shrunk to an interval whose width is about the diameter of the cluster. 

In the presence of a cluster $\mathcal{C}$ of roots (i.e., a set of $k:=|\mathcal{C}|\ge 2$ nearby roots that are ``well separated'' from all other roots),
straight bisection converges only linearly, and it is much more efficient to obtain a good approximation of $\mathcal{C}$ by using Newton iteration. 
More precisely, if we consider a point $\xi$, whose distance $d$ to the cluster $\mathcal{C}$ is considerably larger than the diameter of the cluster, and whose distance to all remaining roots is considerably larger than $d$, then the distance from the point
\begin{align}\label{newton}
\xi':=\xi-k\cdot\frac{P(\xi)}{P'(\xi)}
\end{align}
to the cluster $\mathcal{C}$ is much smaller than the distance from $\xi$ to $\mathcal{C}$.\footnote{The following derivation gives intuition for the behavior of the Newton iteration. Consider $P(x) = (x - \alpha)^k g(x)$, where $\alpha$ is not a root of $g$, and consider the iteration $x_{n+1} = x_n - k \frac{P(x_n)}{P'(x_n)}$. Then, 
\begin{align*}
x_{n+1} - \alpha &= x_n - \alpha - k \frac{(x_n - \alpha)^k g(x_n)}{k (x_n - \alpha)^{k-1} g(x_n) + (x_n - \alpha)^k g'(x_n)} \\
&= (x_n - \alpha)(1 - \frac{k g(x_n)}{k g(x_n) + (x_n - \alpha) g'(x_n)} = (x_n - \alpha)^2 \frac{g'(x_n)}{k g(x_n) + (x_n - \alpha) g'(x_n)},
\end{align*}
and hence, we have quadratic convergence in an interval around $\alpha$. }. The distance $d'$ of $\xi'$ to the cluster is approximately $d^{2}$ if $d<1$. Thus, we can achieve quadratic convergence to the cluster $\mathcal{C}$ by iteratively applying (\ref{newton}). Unfortunately, when running the subdivision algorithm, we neither know whether there actually exists a cluster $\mathcal{C}$ nor do we know its size or diameter. Hence, the challenge is to make the above insight applicable to a computational approach.

We overcome these difficulties as follows. First, we estimate $k$. For this, we consider two choices for $\xi$, say $\xi_1$ and $\xi_2$. Let $\xi'_i$, $i = 1,2$, be the Newton iterates. For the correct value of $k$, we should have $\xi'_1 \approx \xi'_2$. Conversely, we can estimate $k$ by solving $\xi'_1 = \xi'_2$ for $k$. 
Secondly, we use quadratic interval refinement~\cite{abbott-quadratic}. With every active interval $I = (a,b)$, we maintain a number $N_I = 2^{2^{n_{I}}}$, where $n_I\ge 1$ is an integer. We call $n_I$ the \emph{level} of interval $I$. We hope to refine $I$ to an interval $I' = (a',b')$ of width $w(I)/N_I$. We compute candidates for the endpoints of $I'$ using Newton iteration, that is, we compute a point inside $I'$ and then obtain the endpoints of $I'$ by rounding. We apply the $0$-Test to $(a,a')$ and to $(b',b)$. If both $0$-Tests succeed, we add $(I',N_I^2)$ to the set of active intervals. Observe that, in a regime of quadratic convergence, the next Newton iteration should refine $I'$ to an interval of width $w(I')/N_I^2$. If we fail to identify $I'$, we bisect $I$ and add both subintervals to the list of active intervals (with $N_I$ replaced by $\max(4, \sqrt{N_I})$. 

The details of the Newton step are discussed in Section~\ref{sec:Newton-Test and Boundary-Test}, where we introduce the Newton-Test and the Boundary-Test. The Boundary-Test treats the special case that the subinterval $I'$ containing all roots in $I$ shares an endpoint with $I$, and there are roots outside $I$ and close to $I$. 

There is one more ingredient to the algorithm. We need to guarantee that $P$ is large at interval endpoints. Therefore, instead of determining interval endpoints as described above, we instead take an admissible point chosen from an appropriate multipoint.

We next give the details of the algorithm \AD:\medskip

\begin{mdframed}[frametitle={{\bf Algorithm: \AD}}]
{\color{black}
\noindent{\bf Input:}
A polynomial $P(x)$ as in (\ref{def:P}).\smallskip

\noindent{\bf Output:}
Disjoint isolating intervals $I_1,\ldots,I_m$ for all real roots of $P$ with $\var(P,I_j)=1$ for all $j=1,\ldots,n$.\smallskip

\begin{itemize}
\item[(1)] $\mathcal{A}:=\{(\mathcal{I}_{k},4)\}_{k=0,\ldots,2\gamma+2}$, with $\mathcal{I}_{k}$ as computed by Algorithm {\bf Initialization}, and $\mathcal{O}:=\emptyset$.
\item[ ]\noindent$\cosym$\textit{Algorithm {\bf Initialization} is defined in Section~\ref{sec:Initialization}. It computes a set of open and disjoint intervals $\mathcal{I}_k$, such that all real roots of $P$ are covered by the union of these intervals. We remark that $\AD$ works for any set of intervals with this property, however, the choice of Section~\ref{sec:Initialization} simplifies the complexity analysis of the algorithm.}
\item[(2)] while $\mathcal{A}\neq \emptyset$ do
\begin{itemize}
\item[(2.1)] Choose an arbitrary pair $(I,N_{I})$ from $\mathcal{A}$, with $I=(a,b)$, and remove $(I,N_I)$ from $\mathcal{A}$ 
\item[(2.2)] If the algorithm {\bf $\mathbf{0}$-Test} (with input $P$ and $I$) returns True, then go to Step (2.1)
\item[ ]\noindent$\cosym$\textit{In this case, we have $I$ contains no root of $P$, and thus $I$ can be discarded.}
\item[(2.3)] If the algorithm {\bf $\mathbf{1}$-Test} (with input $P$ and $I$) returns True and an interval $I'$, then add $I'$ to $\mathcal{O}$ and go to Step (2.1)
\item[ ]\noindent$\cosym$\textit{If the {\bf $\mathbf{1}$-Test} returns an interval $I'$, then $I'$ is isolating and contains all roots of $P$ that are contained in $I$. Hence, we can store $I'$ as isolating and discard $I\backslash I'$.}
\item[(2.4)] If the algorithm {\bf Boundary-Test} or the algorithm {\bf Newton-Test} returns True and an interval $I'$, then add $(I',N_{I'}):=(I',N_I^2)$ to $\mathcal{A}$, and go to Step (2.1)
\item[ ]\hfill\emph{(quadratic step)}\medskip
\item[ ]\noindent$\cosym$\textit{The {\bf Boundary-Test} and the {\bf Newton-Test} are defined in Section~\ref{sec:Newton-Test and Boundary-Test}. The interval $I'$ returned by one of these tests contains all roots that are contained in $I$, and thus we can proceed with $I'$. Further notice that $\frac{w(I)}{8N_I}\le w(I')\le \frac{w(I)}{N_I}$.}
\item[(2.5)] Compute an admissible point $m^* \in \multipoint{m(I)}{\frac{w(I)}{2^{\ceil{2 + \log n}}}}$ using the algorithm {\bf Admissible Point} and add $(I',N_{I'})$ and $(I'',N_{I''})$ to $\mathcal{A}$, where $I'=(a,m^*)$ and $I''=(m^*,b)$, and $N_{I'}:=N_{I''}:=\max(4,\sqrt{N_{I}})$.
\item[ ]\hfill\emph{(linear step)}\medskip
\item[ ]\noindent$\cosym$\textit{This step correspond to the classical bisection step, where the interval $I$ is split into two equally sized subintervals. Here, we make sure that $|P|$ takes a reasonably large value at the splitting point. In addition, we have $\frac{w(I)}{4}\le \min(w(I'),w(I''))\le \max(w(I'),w(I''))\frac{3w(I)}{4}$.}
\end{itemize}
\item[(3)] return $\mathcal{O}$
\end{itemize}
}
\end{mdframed}
\bigskip
If we succeed in Step (2.4), we say that the subdivision step from $I$ to $I'$ is \emph{quadratic}. In a \emph{linear} step, we just split $I$ into two intervals of approximately the same size (i.e.,~of size in between $\frac{w(I)}{4}$ and $\frac{3w(I)}{4}$). From the definitions of our tests, the exactness of the algorithm follows immediately. In addition, since any interval $I$ of width $w(I)<\frac{\sigma_P}{2}$ satisfies $\var(P,I)=0$ or $\var(P,I)=1$, either the $0$-Test or the $1$-Test succeeds for $I$. This proves termination (i.e.,~Step (2.2) or Step (2.3) succeeds) because, in each iteration, an interval $I$ is replaced by intervals of width less than or equal to $\frac{3w(I)}{4}$.

\subsection{Initialization}\label{sec:Initialization}

Certainly, the most straight-forward initialization is to start with the interval $\mathcal{I}=(-2^{\Gamma},2^{\Gamma})$. In fact, this is also what we recommend doing in an actual implementation. However, in order to simplify the analysis of our algorithm, we proceed slightly differently. 
We first split $\mathcal{I}$ into disjoint intervals $\mathcal{I}_{k}=(s_{k}^{*},s_{k+1}^{*})$, with $k=0,\ldots,2\cdot\gamma+2$ and $\gamma=\log \Gamma$, such that for each interval, $P$ is large at the endpoints of the interval, and $\log M(x)$ is essentially constant within the interval. More precisely, the following conditions are fulfilled for all $k$:
\begin{align}\label{cond:start1}
\min(|P(s_{k}^{*})|,|P(s_{k+1}^{*})|)&>2^{-8n\log n},\quad\text{and}\\
\max_{x\in \mathcal{I}_k}\log M(x)&\le 2\cdot (1+\min_{x\in \mathcal{I}_k}\log M(x)). \nonumber
\end{align}

The intervals $(-2^{2^\gamma},-2^{2^{\gamma-1}})$, $(-2^{2^{\gamma - 1}}, -2^{2^{\gamma - 2}})$, \ldots, $(-2^{2^0},0)$, $(0,+2^{2^0})$, \ldots, $(-2^{2^{\gamma-1}},-2^{2^\gamma})$ satisfy the second condition. In order to also satisfy the first, consider the points 
\begin{align}\label{definition of sk}
s_{k}:= \begin{cases} -2^{2^{\gamma-k}} &\mbox{for } k=0,\ldots,\gamma \\
0 &\mbox{for } k=\gamma+1\\
+2^{2^{k-\gamma}} &\mbox{for } k=\gamma+2,\ldots,2\gamma+2
 \end{cases} 
\end{align}
and corresponding multipoints $M_k \assign \multipoint{s_{k}}{2^{-\lceil 2\log n\rceil}}$. In $M_k$, there exists at least one point   
with distance at least $\frac{1}{8n}$ or larger to all roots of $P$. Thus, $|P(s_{k}^{*})|\ge |P_{n}|\cdot\left(\frac{1}{8n}\right)^{n}\ge 2^{-4n-n\log n}>2^{-8n\log n}$, where $s_k^*$ is an admissible point in $M_k$.\medskip

\begin{mdframed}[frametitle={{\bf Algorithm: Initialization}}]
{\color{black}
\noindent{\bf Input:}
A polynomial $P(x)$ as in (\ref{def:P}) and a $\gamma$ as defined in (\ref{def:Gamma}).\medskip

\noindent{\bf Output:}
Disjoint open intervals $\mathcal{I}_{k}:=(s_{k}^{*},s_{k+1}^{*})$, with $k=0,\ldots,2\gamma+2$, such that $\bigcup_{k}\mathcal{I}_k$ contains all real roots of $P$ and the condition in (\ref{cond:start1}) is fulfilled.\medskip

 \begin{itemize}
\item For $k=0,\ldots,2\gamma+2$, define $s_k$ as in (\ref{definition of sk}).
\item For $k=0,\ldots,2\gamma+2$, compute an admissible point $s_k^* \in M_k \assign \multipoint{s_{k}}{2^{-\lceil 2\log n\rceil}}$ using the algorithm {\bf Admissible Point}.
\item Return the intervals
\begin{align}\label{def:startintervals}
\mathcal{I}_{k}:=(s_{k}^{*},s_{k+1}^{*})\quad\text{for }k=0,\ldots,2\gamma+2.
\end{align}
\end{itemize}}
\end{mdframed}
\bigskip
\ignore{
\hrule \nopagebreak  \medskip
\noindent{\bf Algorithm \AD-Initialization:} Let $s_k$'s be as in (\ref{definition of sk}). For each compute an admissible point $s_k^* \in M_k \assign \multipoint{s_{k}}{2^{-\lceil 2\log n\rceil}}$ and initialize $\cal A$ with the intervals 
\begin{align}\label{def:startintervals}
\mathcal{I}_{k}:=(s_{k}^{*},s_{k+1}^{*})\quad\text{for }k=0,\ldots,2\gamma+2.
\end{align}\vspace{-1em} \par
\hrule \bigskip}
It is easy to check that the second condition in (\ref{cond:start1}) is fulfilled for the intervals computed by the above algorithm.

\subsection{The Newton-Test and the Boundary-Test}\label{sec:Newton-Test and Boundary-Test}

The Newton-Test and the Boundary-Test are the key to quadratic convergence. The Newton-test
receives an interval $I=(a,b)\subseteq \mathcal{I}$ and an integer $N_{I}=2^{2^{n_{I}}}$, where $n_I \ge 1$ is an integer. In case of success, the test returns an interval $I'$ with $\frac{w(I)}{8N_{I}}\le w(I')\le \frac{w(I)}{N_{I}}$ that contains all roots that are contained in $I$. Success is guaranteed  if there is a subinterval $J$ of $I$ of width at most $2^{-13}\cdot\frac{w(I)}{N_{I}}$ whose one-circle region contains all roots that are contained in the one-circle region of $I$ and if the disk with radius $2^{\log n+10}\cdot N_{I}\cdot w(I)$ and center $m(I)$ contains no further root of $P$, see Lemma~\ref{newtonsuccess} for a precise statement. Informally speaking, the Newton-Test is guaranteed to succeed if the roots in $I$ cluster in a subinterval significantly shorter than $w(I)/N_I$, and roots outside $I$ are far away from $I$. \MS{In the following description of the Newton-Test, we inserted comments that explain the rationale behind our choices.} For this rationale, we assume the existence of a cluster $\mathcal{C}$ of $k$ roots centered at some point $\xi\in I$ with diameter $d(\mathcal{C})\ll w(I)$ and that there exists no other root in a large neighborhood of the one-circle region $\Delta(I)$ of $I$. The formal justification for the Newton-Test will be given in Lemma~\ref{newtonsuccess}.\medskip

\begin{mdframed}[frametitle={{\bf Algorithm: Newton-Test}}]
{\color{black}
\noindent{\bf Input:}
An Interval $I=(a,b)\subset \R$, an integer $N_I=2^{2^{n_I}}$ with $n_I\in\N$, and a polynomial $P\in\R[x]$
as defined in (\ref{def:P}) \smallskip

\noindent{\bf Output:}
True or False. In case of True, it also returns an interval $I'\subset I$, with $\frac{w(I)}{8N_I}\le w(I')\le\frac{w(I)}{N_I}$, that contains all real roots of $P$ that are contained in $I$.\smallskip

\begin{itemize}
\item[(1)] Let $\xi_{1}:=a+\frac{1}{4}\cdot w(I)$, $\xi_{2}:=a+\frac{1}{2}\cdot w(I)$, $\xi_{3}:=a+\frac{3}{4}\cdot w(I)$, and $\epsilon := 2^{-\ceil{5 + \log n}}$, and compute admissible points 
\begin{equation}\label{Newtonmultipoint1}
\xi_{j}^{*}\in \multipoint{\xi_{j}}{\epsilon \cdot w(I)},
\end{equation}
for $j=1,2,3$, using the Algorithm {\bf Admissible Point} from Section~\ref{sec:multipoint}.\medskip
\item[ ]\noindent$\cosym$\textit{At least two of the three points $\xi_{j}$ (say $\xi_{1}$ and $\xi_{2}$) have a distance from $\mathcal{C}$ that is large compared to the diameter of $\mathcal{C}$. In addition, their distances to all remaining roots are also large, and thus the points $\xi_{1}':=\xi_{1}-k\cdot v_{1}$ and $\xi_{2}':=\xi_{2}-k\cdot v_{2}$, with $v_{j_1}:=\frac{P(\xi^{*}_{j_{1}})}{P'(\xi^{*}_{j_{1}})}$ and $v_{j_2}:=\frac{P(\xi^{*}_{j_{2}})}{P'(\xi^{*}_{j_{2}})}$, obtained from considering one Newton step, with $\xi=\xi^*_{j_1}$ and $\xi=\xi^*_{j_2}$, have much smaller distances to $\mathcal{C}$ than the points $\xi_{1}$ and $\xi_{2}$. Notice that $k$ is not known to the algorithm at this point.\medskip} 
\item[(2)] For each of the three distinct pairs $(j_1,j_2)$ of indices $j_{1},j_{2}\in\{1,2,3\}$, with $j_{1}<j_{2}$, do:
\begin{itemize}
\item[(2.1)] For $L=2,4,8,\ldots$ , compute 
approximations $A_{j_{1}}$, $A_{j_{2}}$, $A'_{j_{1}}$, and $A'_{j_{2}}$ of $P(\xi^{*}_{j_{1}})$, $P(\xi^{*}_{j_{2}})$, $P'(\xi^{*}_{j_{1}})$, and 
$P'(\xi^{*}_{j_{2}})$ of quality $L$, respectively, until, for some $L=L_1$, it holds that
\begin{align}\label{condition1}
\frac{|A_{j_1}|-2^{-L}}{|A'_{j_1}|+2^{-L}}>w(I)\quad\text{or}\quad\frac{|A_{j_2}|-2^{-L}}{|A'_{j_2}|+2^{-L}}>w(I)
\end{align} 
or
\begin{align}\label{condition2}
|A_{j_1}|,|A_{j_2}|>2^{-L+1}\text{ and }|A'_{j_1}|,|A'_{j_2}|>2^{-L+1}
\end{align} 
Then, if (\ref{condition1}) holds, discard the pair $(j_1,j_2)$. Otherwise, proceed with Step (2.2).\medskip
\item[ ]\noindent$\cosym$\textit{The values $\tilde{v}_{j_1}:=\frac{A_{j_1}}{A'_{j_1}}$ and $\tilde{v}_{j_2}:=\frac{A_{j_2}}{A'_{j_2}}$ are approximations of $v_{j_1}$ and $v_{j_2}$, respectively. Condition (\ref{condition1}) implies that either $|v_{j_1}|>w(I)$ or $|v_{j_2}|>w(I)$. The proof of Lemma~\ref{newtonsuccess} shows that the existence of $\mathcal{C}$ implies that there is a pair $(j_1,j_2)$, for which $v_{j_1}$ and $v_{j_2}$ have distance $\ll \frac{w(I)}{N_I}$ to $\mathcal{C}$ and $|v_{j_1}|,|v_{j_2}|<\frac{w(I)}{k}\le w(I)$. Hence, such a pair cannot be discarded in step (2.1). Further notice that $L_1\le L_1^*:=\log \max\{M(|P(\xi^{*}_{j_{1}})|^{-1}),M(|P(\xi^{*}_{j_{1}})|^{-1})\}+\log M(w(I))+3$. Namely, for $L\ge L_1^*$, either (\ref{condition2}) holds or $$\frac{|A_{j}|-2^{-L}}{|A'_{j}|+2^{-L}}>\frac{|P(\xi^{*}_{j})|-2^{-L+1}}{2^{-L+1}+2^{-L}}\ge
\frac{8\cdot M(w(I))\cdot 2^{-L}-2^{-L+1}}{3\cdot 2^{-L}}\ge w(I)$$ for $j=j_1$ or $j=j_2.$
In addition, if (\ref{condition2}) holds for $L=L_1$, then (\ref{condition2}) also holds for any $L\ge 2\cdot L_1$. This also implies that, for $j=j_1,j_2$, $|\tilde{v}_{j}|$ is a $4$-approximation of $|v_{j}|$, that is, 
$|\tilde{v}_{j}|$ differs from $|v_{j}|$ by a factor in $(1/4,4)$. Since $\frac{|A_j|+2^{-L}}{|A'_j|-2^{-L}}$ differs from $\tilde{v}_j$ by a factor in $(1,4)$, it follows that $|\tilde{v}_{j}|<64w(I)$ for any $L\ge 2 L_1$ and $j=j_1,j_2$.\medskip
}
\item[(2.2)] For $L=2\cdot L_1,4\cdot L_1,8\cdot L_1,\ldots$ , compute 
approximations $A_{j_{1}}$, $A_{j_{2}}$, $A'_{j_{1}}$, and $A'_{j_{2}}$ of $P(\xi^{*}_{j_{1}})$, $P(\xi^{*}_{j_{2}})$, $P'(\xi^{*}_{j_{1}})$, and 
$P'(\xi^{*}_{j_{2}})$ of quality $L$, respectively, until, for some $L=L_2$, it holds that
\begin{align}\label{condition3}
\delta_{j_1},\delta_{j_2}<\min\left\{\frac{w(I)}{2^5\cdot n},\frac{w(I)}{2^{14}\cdot N_I}\right\},\text{ with }\delta_{j_1}:=\frac{|A_{j_1}|+|A'_{j_1}|}{2^{L-2}\cdot|A'_{j_1}|^2},\text{ }\delta_{j_2}:=\frac{|A_{j_2}|+|A'_{j_2}|}{2^{L-2}\cdot|A'_{j_2}|^2}.
\end{align} 
If the condition
\begin{align}\label{condition4}
\left|\tilde{v}_{j_1}-\tilde{v}_{j_2}\right|+\delta_{j_1}+\delta_{j_2}\ge\frac{w(I)}{n},\quad\text{with}\quad \tilde{v}_{j_1}:=\frac{A_{j_1}}{A'_{j_1}}\quad\text{and}\quad\tilde{v}_{j_2}:=\frac{A_{j_2}}{A'_{j_2}},
\end{align}
is fulfilled, then proceed with Step (2.3). Otherwise, discard the pair $(j_1,j_2)$.\medskip
\item[ ]\noindent$\cosym$\textit{A straight-forward computation shows that $|\tilde{v}_{j_1}-v_{j_1}|<\delta_{j_1}$ and $|\tilde{v}_{j_2}-v_{j_2}|<\delta_{j_2}$, where we use that $|A_j|$ and $|A'_j|$ are relative $2$-approximations of $|P(\xi^*_j)|$ and $|P'(\xi^*_j)|$, for $j=j_1,j_2$, respectively. Hence, if condition (\ref{condition4}) does not hold, then $|v_{j_1}-v_{j_2}|< \frac{w(I)}{n}$. Due to the proof of Lemma~\ref{newtonsuccess}, the existence of $\mathcal{C}$ yields that $|v_{j_1}-v_{j_2}|>\frac{w(I)}{k}$ for some pair $(j_1,j_2)$, and thus (\ref{condition4}) must be fulfilled for such a pair. Notice that the inequality in (\ref{condition3}) holds for $L=L_2$, with an $L_2$ of size 
\begin{align*}
L_2&=O(L_1+\log \max(M((32n)/w(I)),N_I))\\
&=O(\log \max\{n,M(|P(\xi^{*}_{j_{1}})|^{-1}),M(|P(\xi^{*}_{j_{2}})|^{-1}),M(w(I)),M(w(I)^{-1}),N_I\}),
\end{align*}
where we use that $\tilde{v}_{j}<64w(I)$ and $2^{L}\cdot |A_{j}'|\ge 2^{L}\cdot \frac{1}{4}\cdot 2^{-L_1+1}=2^{L-L_1-1}$ for all $L\ge 2\cdot L_1$ and $j=j_1,j_2$.\medskip  
}
\item[(2.3)] Compute
\begin{align}\label{def:lambdaapx}
\tilde{\lambda}_{j_{1},j_{2}}:=\xi_{j_{1}}^{*}+\frac{\xi^{*}_{j_{2}}-\xi^{*}_{j_{1}}}{\tilde{v}_{j_{1}}-\tilde{v}_{j_{2}}}\cdot \tilde{v}_{j_{1}}
\end{align}
If $\tilde{\lambda}_{j_1,j_2} \not\in \bar{I}=[a,b]$, discard the pair $(j_1,j_2)$. Otherwise, compute $\ell_{j_{1},j_{2}}:=\floor{\frac{\tilde{\lambda}_{j_{1},j_{2}}-a}{w(I)/(4N_I)}}$, which is an integer contained in $\{0,\ldots,4N_{I}\}$. Further define
\[
I_{j_{1},j_{2}}:=(a_{j_{1},j_{2}},b_{j_{1},j_{2}}):=(a+\max(0,\ell_{j_{1},j_{2}}-1)\cdot\frac{w(I)}{4N_{I}},a+\min(4N_{I},\ell_{j_{1},j_{2}}+2)\cdot\frac{w(I)}{4N_{I}}).
\]
If $a_{j_{1},j_{2}}=a$, set $a_{j_{1},j_{2}}^{*}:=a$, and if $b_{j_{1},j_{2}}=b$, set $b_{j_{1},j_{2}}:=b$. For all other values for $a_{j_{1},j_{2}}$ and $b_{j_{1},j_{2}}$, use Algorithm {\bf Admissible Point} from Section~\ref{sec:multipoint} to compute admissible points 
\begin{equation}\label{Newtonmultipoint2}
a_{j_{1},j_{2}}^{*}\in \multipoint{a_{j_{1},j_{2}}}{\epsilon\cdot \frac{w(I)}{N_I}}
\quad\text{and}\quad
b_{j_{1},j_{2}}^{*}\in \multipoint{b_{j_{1},j_{2}}}{\epsilon\cdot \frac{w(I)}{N_I}}.
\end{equation}
Define $I_{j_{1},j_{2}}^{*}:=(a_{j_{1},j_{2}}^{*},b_{j_{1},j_{2}}^{*})$.\medskip 
\item[ ]\noindent$\cosym$\textit{The Newton iteration (\ref{newton}) with $\xi=\xi_{j}^{*}$ for a $k$-fold root produces $\xi_j' = \xi_j^* - k v_j$. Equating $\xi'_{j_1} = \xi'_{j_2}$ yields $-k = \frac{\xi^{*}_{j_{2}}-\xi^{*}_{j_{1}}}{v_{j_{1}}-v_{j_{2}}}$. Then, $\xi'_{j_1}$ and $\xi'_{j_2}$ are given by 
\begin{align}\label{def:lambda}
\lambda_{j_{1},j_{2}}:=\xi_{j_{1}}^{*}+\frac{\xi^{*}_{j_{2}}-\xi^{*}_{j_{1}}}{v_{j_{1}}-v_{j_{2}}}\cdot v_{j_{1}}.
\end{align}
A straight-forward computation shows that $|\tilde{\lambda}_{j_1,j_2}-\lambda_{j_1,j_2}|<\frac{w(I)}{32N_I}$, where we use inequality (\ref{condition3}) and the fact that $|\tilde{v}_{j_1}|,|\tilde{v}_{j_2}|<64w(I)$, and $|\tilde{v}_{j_1}|,|\tilde{v}_{j_2}|>\frac{w(I)}{2n}$ due to (\ref{condition4}). If $\tilde{\lambda}_{j_1,j_2}$ is contained in $I$, we (conceptually) subdivide $I$ into $4N_I$ subintervals and determine the subinterval that contains $\tilde{\lambda}_{j_1,j_2}$. Extending the interval on both sides by $\frac{w(I)}{4N_I}$ yields an interval $I_{j_1,j_2}$, which contains $\lambda_{j_1,j_2}$. Finally, replacing the endpoints $a_{j_1,j_2}$ and $b_{j_1,j_2}$ by nearby admissible points yields an interval $I^*_{j_1,j_2}$ with $\frac{w(I)}{8n}\le w(I^*_{j_1,j_2})\le \frac{w(I)}{n}$. Lemma~\ref{newtonsuccess} then shows that the existence of $\mathcal{C}$ guarantees that $I^*_{j_1,j_2}$ contains all roots of $P$ that are contained in $I$.\medskip 
}
\item[(2.4)] Run the {\bf $\mathbf{0}$-Test} from Section~\ref{var0apx} with input $I'_{\ell}:=(a,a_{j_{1},j_{2}}^{*})$ and $I'_{r}:=(b_{j_{1},j_{2}}^{*},b)$. If it succeeds on both intervals, return {\bf True} and the interval $I':=I^*_{j_1,j_2}$. Otherwise, discard the pair $(j_1,j_2)$.\medskip
\item[ ]\noindent$\cosym$\textit{For intervals $I'_{\ell}$ or $I'_{r}$, which are empty (i.e.,~$I'_{\ell}=(a,a)$ or $I'_{r}=(b,b)$), nothing needs to be done. If the {\bf $\mathbf{0}$-Test} succeeds on $I'_{\ell}$ as well as on $I'_{r}$, then neither interval contains a root of $P$. Hence, $I^*_{j_1,j_2}$ contains all roots of $P$ that are contained in $I$.\medskip}
\end{itemize}
\item[(3)] If each of the three pairs $(j_1,j_2)$ is discarded in one of the above steps, return {\bf False}.
\end{itemize}
}
\end{mdframed}
\bigskip

We next derive a sufficient condition for the success of the Newton-Test. 

\begin{lemma}\label{newtonsuccess}
Let $I=(a,b)$ be an interval, $N_{I}=2^{2^{n_{I}}}$ with $n_{I}\in\Z_{\ge 1}$, and $J=(c,d)\subseteq I$ be a subinterval of width $w(J)\le 2^{-13}\cdot\frac{w(I)}{N_{I}}$. Suppose that the one-circle region of $\Delta(J)$ contains $k$ roots $z_{1},\ldots,z_{k}$ of $P$, with $k\ge 1$, and that the disk with radius $2^{\log n+10}\cdot N_{I}\cdot w(I)$ and center $m(I)$ contains no further root of $P$. Then, the 
Newton-Test succeeds.
\end{lemma}

\begin{proof}
We first show that, for at least two of the three points $\xi_{j}^{*}$, $j=1,2,3$, the inequality $\abs{ m(J)-(\xi_{j}^{*}-k\cdot\frac{P(\xi_{j}^{*})}{P'(\xi_{j}^{*})})}<\frac{w(I)}{128N_{I}}$ holds: There exist at least two points (say $\xi:=\xi_{j_{1}}^{*}$ and $\bar{\xi}:=\xi_{j_{2}}^{*}$ with $j_{1}<j_{2}$) whose distances to any root from $z_{1},\ldots,z_{k}$ are larger than $\frac{|\xi-\bar{\xi}|}{2}-\frac{w(J)}{2}>\frac{3}{32}w(I)-\frac{1}{2}w(J)\ge 512N_Iw(J)$. In addition, the distances to any of the remaining roots $z_{k+1},\ldots,z_{n}$ are larger than $2^{10}n N_{I}w(I)-w(I)\ge 512\cdot n N_{I}w(I)$. Hence, with $m:=m(J)$, it follows that
\begin{align} \nonumber
\left|\frac{1}{k}\cdot\frac{(\xi-m)P'(\xi)}{P(\xi)}-1\right|&=\left|\frac{1}{k}\sum_{i=1}^k\frac{\xi-m}{\xi-z_i}+\frac{1}{k}\sum_{i>k} \frac{\xi-m}{\xi-z_i}-1\right|=\frac{1}{k}\left|\sum_{i=1}^k\frac{z_i-m}{\xi-z_i}+\sum_{i>k} \frac{\xi-m}{\xi-z_i}\right|\\ \nonumber
&\le \frac{1}{k}\sum_{i=1}^k\frac{|z_i-m|}{|\xi-z_i|}+\frac{1}{k}\sum_{i>k} \frac{|\xi-m|}{|\xi-z_i|}<\frac{w(J)}{512nN_{I}w(J)}+\frac{(n-k)\cdot w(I)}{512 knN_{I}w(I)}\\ \nonumber
&\le \frac{1}{256N_{I}},
\end{align}
where we used that $\frac{P'(\xi)}{P(\xi)}=\sum_{i=1}^n (\xi-z_i)^{-1}$.
This yields the existence of an $\epsilon\in\R$ with $|\epsilon|<\frac{1}{256N_I}\le \frac{1}{1024}$ and $\frac{1}{k}\cdot\frac{(\xi-m)P'(\xi)}{P(\xi)}=1+\epsilon$. We can now derive the following bound on the distance between the approximation $\xi'=\xi-k\cdot\frac{P(\xi)}{P'(\xi)}$ obtained by the Newton iteration and $m$:
\begin{align}\nonumber
\left|m-\xi'\right|=|m-\xi|\cdot \left|1-\frac{1}{\frac{1}{k}\cdot\frac{(\xi-m)P'(\xi)}{P(\xi)}}\right|=|m-\xi|\cdot \left|1-\frac{1}{1+\epsilon}\right|=\left|\frac{\epsilon\cdot (m-\xi)}{1+\epsilon}\right|<\frac{w(I)}{128N_{I}}.
\end{align}
In a completely analogous manner, we show that $\abs{\bar{\xi}-k\cdot \frac{P(\bar{\xi})}{P'(\bar{\xi})}-m}<\frac{w(I)}{128N_I}$.

Let $v_{j_{1}}=\frac{P(\xi)}{P'(\xi)}$ and $v_{j_{2}}=\frac{P(\bar{\xi})}{P'(\bar{\xi})}$ be defined as in the Newton-Test. 
Then, from the above considerations, it follows that $|(\xi-k\cdot v_{j_{1}})-(\bar{\xi}-k\cdot v_{j_{2}})|<\frac{w(I)}{64N_{I}}$. Hence, since $|\xi-\bar{\xi}|>\frac{3w(I)}{16}$ and $1\le k\le n$, we must have $|v_{j_{1}}-v_{j_{2}}|>\frac{w(I)}{8k}$. Furthermore, it holds that $|k\cdot v_{j_{1}}|<w(I)$ since, otherwise, the point $\xi-k\cdot v_{j_{1}}$ is not contained in $(\xi-w(I),\xi+w(I))$, which 
contradicts the fact that $|\xi-k\cdot v_{j_{1}}-m|<\frac{w(I)}{128N_{I})}$ and $m\in I$. An analogous argument yields that $|k\cdot v_{j_{2}}|<w(I)$. \MS{Hence, none of the two inequalities in (\ref{condition1}) are fulfilled, whereas the inequality in (\ref{condition4}) must hold.} In the next step, we show that $\lambda:=\lambda_{j_{1},j_{2}}$ as defined in (\ref{def:lambda}) is actually a good approximation of $\xi-k\cdot v_{j_{1}}$: There exist $\epsilon$ and $\bar{\epsilon}$, both of magnitude less than $\frac{w(I)}{128N_{I}}$, such that $\xi-k\cdot v_{j_{1}}=m+\epsilon$ and $\bar{\xi}-k\cdot v_{j_{2}}=m+\bar{\epsilon}$. This yields
\[
\lambda=\xi+\frac{\bar{\xi}-\xi}{v_{j_{1}}-v_{j_{2}}}\cdot v_{j_{1}}=\xi+\left(\frac{\bar{\epsilon}-\epsilon+k\cdot(v_{j_{2}}-v_{j_{1}})}{v_{j_{1}}-v_{j_{2}}}\right)\cdot v_{j_1}=\xi+k\cdot v_{j_{1}}+\frac{(\bar{\epsilon}-\epsilon)v_{j_{1}}}{v_{j_{1}}-v_{j_{2}}}.
\]
The absolute value of the fraction on the right side is smaller than $\frac{w(I)\cdot |v_{j_{1}}|}{64N_{I}}\cdot\frac{8k}{w(I)} \le \frac{w(I)}{8N_{I}}$, and thus $|\xi-k\cdot v_{j_{1}}-\lambda|<\frac{w(I)}{8N_{I}}$.
\MS{Hence, with $\tilde{\lambda}_{j_1,j_2}$ as defined in (\ref{def:lambdaapx}), we have} 
$$|m-\tilde{\lambda}_{j_{1},j_{2}}|\le |m-(\xi-k\cdot v_{j_{1}})|+|(\xi-k\cdot v_{j_{1}})-\lambda|+|\lambda-\tilde{\lambda}_{j_{1},j_{2}}|<\frac{w(I)}{128N_{I}}+\frac{w(I)}{8N_{I}}+\frac{w(I)}{32N_{I}}<\frac{3w(I)}{16N_{I}}.$$
From the definition of the interval $I_{j_{1},j_{2}}$, we conclude that $J\subseteq I_{j_{1},j_{2}}$. Furthermore, each endpoint of $I_{j_{1},j_{2}}$ is either an endpoint of $I$, or its distance to both endpoints of $J$ is larger than $\frac{w(I)}{16N_{I}}-\frac{w(J)}{2}>\frac{w(I)}{32N_{I}}>\frac{w(J)}{2}$.
This shows that the interval $I'=I_{j_{1},j_{2}}^{*}$ contains $J$. Hence, the Newton-Test succeeds since the one-circle regions of $I'_{\ell}$ and $I'_{r}$ contain no roots of $P$.
\end{proof}

The Newton-Test is our main tool to speed up convergence to clusters of roots without actually knowing that there exists a cluster. 
However, there is one special case that has to be considered separately: Suppose that there exists a cluster $\mathcal{C}\subseteq\Delta(I)$ of roots whose center is close to one of the endpoints of $I$. If, in addition, $\mathcal{C}$ is not well separated from other roots that are located outside of $\Delta(I)$, then the above lemma does not apply.
For this reason, we
introduce the \emph{Boundary-Test}, which checks for clusters near the endpoints of an interval $I$. Its input is the same as for the Newton-Test. In case of success, it either returns an interval $I'\subseteq I$, with $\frac{w(I)}{4N_{I}}\le w(I')\le \frac{w(I)}{N_{I}}$, which contains all real roots that are contained in $I$, or it proves that $I$ contains no root.\medskip

\begin{mdframed}[frametitle={{\bf Algorithm: Boundary-Test}}]
{\color{black}

\noindent{\bf Input:}
An Interval $I=(a,b)\subset \R$, an integer $N_I=2^{2^{n_I}}$ with $n_I\in\N$, and a polynomial $P\in\R[x]$
as defined in (\ref{def:P}) \smallskip

\noindent{\bf Output:}
True or False. In case of True, it also returns an interval $I'\subset I$, with $\frac{w(I)}{8N_I}\le w(I')\le\frac{w(I)}{N_I}$, that contains all real roots of $P$ that are contained in $I$.\smallskip

\begin{itemize}
\item[(1)] Let $m_{\ell}:=a+\frac{w(I)}{2N_{I}}$, $m_{r}:=b-\frac{w(I)}{2N_{I}}$, and $\epsilon := 2^{-\ceil{2 + \log n}}$. Use algorithm {\bf Admissible Point} to compute admissible points 
\begin{equation}\label{Boundarymultipoint}
m_{\ell}^{*}\in  \multipoint{m_{\ell}}{\epsilon \cdot\frac{w(I)}{N_I}} \quad\text{and}\quad m_{r}^{*}\in \multipoint{m_{r}}{\epsilon\cdot\frac{w(I)}{N_I}}, 
\end{equation}
\item[(2)] If the $0$-Test returns true for the interval $I_{\ell}:=(m_{\ell}^*,b)$, then return $(a,m_{\ell}^ *)$.
\item[(3)] If the $0$-Test returns true for the interval $I_{r}:=(a,m_{r}^*)$, then return $(m_{r}^ *,b)$.
\item[(4)] return False
\end{itemize}
}
\end{mdframed}
\bigskip
\ignore{
\hrule \nopagebreak \medskip
\noindent{\bf Algorithm Boundary-Test:}
Let $m_{\ell}:=a+\frac{w(I)}{2N_{I}}$ and $m_{r}:=b-\frac{w(I)}{2N_{I}}$, and let $\epsilon := 2^{-\ceil{2 + \log n}}$. Compute 
admissible points 
\begin{equation}\label{Boundarymultipoint}
m_{\ell}^{*}\in  \multipoint{m_{\ell}}{\epsilon \cdot\frac{w(I)}{N_I}} \quad\text{and}\quad m_{r}^{*}\in \multipoint{m_{r}}{\epsilon\cdot\frac{w(I)}{N_I}}, 
\end{equation}
and run the $0$-Test for the intervals $I_{\ell}:=(m_{\ell},b)$ and $I_{r}:=(a,m_{r})$. If both $0$-Tests succeed, $I$ contains no root, and thus, we return this result.
If it succeeds only for $I_{\ell}$, then $I':=(a,m_{\ell})$ contains all roots that are contained in $I$, and we return $I'$. If it succeeds only for $I_{r}$, we return $I'=(m_{r},b)$. Notice that from our definition of $m_{\ell}^{*}$ and $m_{r}^{*}$, it follows that both intervals $I_{\ell}$ and $I_{r}$ have width in between $\frac{w(I)}{4N_{I}}$ and $\frac{w(I)}{N_{I}}$. If neither $0$-Test succeeds, the Boundary-Test fails. \nopagebreak
\medskip
\hrule\bigskip
}
Clearly, if all roots contained in $\Delta(I)$ have distance less than $\frac{w(I)}{4N_{I}}$ to one of the two endpoints of $I$, the Boundary-Test for $I$ is successful, as the one-circle region of either $I_{\ell}$ or $I_{r}$ contains no root of $P$.

\section{Complexity Analysis}\label{sec:analysis}

We bound the size of the subdivision tree in Section~\ref{sec:size-subdivision-tree} and the bit complexity in Section~\ref{sec:bit-complexity}.

\subsection{Size of the Subdivision Tree}\label{sec:size-subdivision-tree}

We use $\mathcal{T}$ to denote the subdivision forest which is induced by our algorithm $\AD$. More 
precisely, in this forest, we have one tree for each interval $\mathcal{I}_{k}$, with $k=0,\ldots,2\log\Gamma+2$,  as defined in (\ref{def:startintervals}). Furthermore, an interval $I'$ is a child of some $I\in \mathcal{T}$ if and only if it 
has been created by our algorithm when processing $I$. We have $\frac{w(I)}{8N_{I}}\le w(I')\le \frac{w(I)}{N_{I}}$ in a quadratic step and $\frac{1}{4}\cdot w(I)\le w(I')\le \frac{3}{4}\cdot w(I)$ in a linear step. An interval in $\mathcal{T}$ has two, one, or zero children. Intervals with zero children are called \emph{terminal}. Those are precisely the intervals for which either the $0$-Test or the $1$-Test is successful. 
Since each interval $I\neq \mathcal{I}$ with $\var(P,I)\le 1$ is terminal, it follows that, for each non-terminal interval $I$, the one circle region $\Delta(I)$ contains at least one root and the 
two-circle region of $I$ contains at 
least two roots of $P$. Thus, all non-terminal nodes have width larger than or equal to $\sigma_P/2$.

In order to estimate the size of $\mathcal{T}$, we estimate for each $\mathcal{I}_k$ the size of the tree $\mathcal{T}_k$ rooted at it.  If $\mathcal{I}_k$ is terminal, $\mathcal{T}_k$ consists only of the root. So, assume that $\mathcal{I}_k$ is non-terminal. Call a non-terminal $I\in\mathcal{T}_k$ \emph{splitting} if either $I$ is the root of $\mathcal{T}_k$, or $\mathcal{M}(I')\neq \mathcal{M}(I)$ for all children $I'$ of $I$ (recall that $\mathcal{M}(I)$ denotes the set of roots of $P$ contained in the one circle region $\Delta(I)$ of $I$), or if all children of $I$ are terminal. By the argument in the preceding paragraph, $\mathcal{M}(I) \not= \emptyset$ for all splitting nodes. A splitting node $I$ is called \emph{strongly splitting} if there exists a root $z \in\mathcal{M}(I)$ that is not contained in any of the one-circle regions of its children.
The number of splitting nodes in $\mathcal{T}_k$ is bounded by 
$2\abs{\mathcal{M}_k}$ since there are at most $\abs{\mathcal{M}_k}$ splitting nodes all of whose children are terminal, since at most $\abs{\mathcal{M}_k}-1$ splitting nodes all of whose children have a smaller set of roots in the one-circle region of the associated interval, and since there is one root. 
For any splitting node, consider the path of non-splitting nodes ending in it, and let $s_{\max}$ be the maximal length of such a path (including the splitting node at which the path ends and excluding the splitting node at which the path starts). Then, the number of non-terminal nodes in $\mathcal{T}_k$ is bounded by $1 + s_{\max} \cdot (2 \abs{\mathcal{M}_k} - 1)$, and the total number of non-terminal nodes in the subdivision forest is $O(\log \Gamma + n \cdot s_{\max})$. Hence, the same bound also applies to the number of all nodes in $\mathcal{T}$.

The remainder of this section is concerned with proving that 
\[s_{\max}=O(\log n+\log(\Gamma+\log M(\sigma_P^{-1}))).\] 
The proof consists of three parts. 
\begin{compactenum}[(1)]
\item We first establish lower and upper bounds for the width of \emph{all} (i.e.,~also for terminal) intervals $I\in\mathcal{T}$ and the corresponding numbers $N_{I}$ (Lemma~\ref{lemma:bounds1}). 
\item We then study an abstract version of how interval sizes and interval levels develop in quadratic interval refinement (Lemma~\ref{lemma:sequence}). 
\item In a third step, we then derive the bound on $s_{\max}$  (Lemma~\ref{lemma:bound on smax}). 
\end{compactenum}

\begin{lemma}\label{lemma:bounds1}
For each interval $I\in\mathcal{T}\backslash{\mathcal{I}}$, we have $$2^{\Gamma}\ge w(I)\ge 2^{-4\Gamma-6}\sigma_{P}^{5}\quad\text{and}\quad4\le N_I\le 2^{4(\Gamma+1)}\cdot \sigma_P^{-4}.$$
\end{lemma}

\begin{proof}
The inequalities $2^{\Gamma}\ge w(I)$ and $N_I\ge 4$ are trivial. For $N_I>16$, there exist 
two ancestors (not necessarily parent and grandparent) $J'$ and $J$ of $I$, with $I\subseteq J'\subseteq J$, such that $w(I)\le w(J')/N_{J'}$ and $w(J')\le w(J)/N_{J}$, and $N_{I}=N_{J'}^{2}=N_{J}^{4}$. Hence, it follows that $J'$ is a non-terminal interval of width less than or equal to $2^{\Gamma}/N_{J}=2^{\Gamma}N_{I}^{-1/4}$.
Since each non-terminal interval has width $\sigma_{P}/2$ or more, the upper bound on $N_{I}$ follows. For the claim on the width of $I$, we remark that the parent interval $K$ of $I$ has width $\sigma_{P}/2$ or more and that $\frac{w(K)}{8N_{J}}\le w(I)\le \frac{3}{4}w(K)$. 
\end{proof}

We come to the evolution of interval sizes and levels in quadratic interval refinement. The following Lemma has been introduced in~\cite[Lemma 4]{NewDsc} in a slightly weaker form:

\begin{lemma}\label{lemma:sequence}
Let $w$, $w'\in\R^+$ be two positive reals with $w>w'$, and let $m\in\mathbb{N}_{\ge 1}$ be a positive integer. We recursively define the sequence $(s_i)_{i\in\N_{\ge 1}}:=((x_i,n_i))_{i\in\N_{\ge 1}}$ as follows: Let $s_1=(x_{1},n_1):=(w,m)$, and
$$s_{i+1}=\left(x_{i+1},n_{i+1}\right):=\begin{cases}
  \left(\epsilon_{i}\cdot x_{i},n_{i}+1\right)\text{ with an } \epsilon_{i}\in [0,\frac{1}{N_{i}}],   & \text{if }\frac{x_{i}}{N_{i}}\ge w'\\
  \left(\delta_{i}\cdot x_{i},\max(1,n_{i}-1)\right)\text{ with a } \delta_{i}\in [0,\frac{3}{4}],  & \text{if }\frac{x_{i}}{N_{i}}<w',
\end{cases}
$$
where $N_i:=2^{2^{n_{i}}}$ and $i\ge 1$. Then, the smallest index $i_0$ with $x_{i_0}\le w'$ is bounded by $8(n_1+\log\log \max(4,\frac{w}{w'}))$.
\end{lemma}

\begin{proof}
The proof is similar to the proof given in~\cite{NewDsc}. However, there are subtle differences, and hence, we give the full proof. We call an index $i$ \emph{strong} (\textbf{S}) if $x_i/N_i\ge w'$ and \emph{weak} (\textbf{W}), otherwise. If $w/4<w'$, then each $i\ge 1$ is weak, and thus, $i_0\le 6$ because of $(3/4)^{5}<1/4$. 

So assume $w/4\ge w'$ and let $k'$ be the smallest weak index. We split the sequence $1,2,\ldots, i_0$ into three parts, namely (1) the prefix $1,\ldots,k'-1$ of strong indices, (2) the subsequence $k',\ldots,i_0-6$ starting with the first weak index and containing all indices but the last 6, and (3) the tail $i_0-5,\ldots,i_0$. The length of the tail is 6. 

We will show that the length of the prefix of strong indices is bounded by $k\in\mathbb{N}_{\ge 1}$ where $k$ is the unique integer with $$2^{-2^{k+1}}<w'/w\le 2^{-2^{k}}.$$ Then, $k\le\log\log\frac{w}{w'}$. Intuitively, this holds since we square $N_i$ in each strong step, and hence, after $O(\log \log w/w')$ strong steps we reach a situation where a single strong step  guarantees that the next index is weak. In order to bound the second subsequence, we split it into subsubsequences of maximal length containing no two consecutive weak indices. We will show that the subsubsequences have length at most five and that each such subsubsequence (except for the last) has one more weak index than strong indices. Thus, the value of $n$ at the end of a subsubsequence is one smaller than at the beginning of the subsubsequence, and hence, the number of subsubsequences is bounded by $n_1$. 
We turn to the bound on the length of the prefix of strong indices. 

\noindent\emph{Claim 1:}  $k'\le k+1$.

\noindent Suppose that the first $k$ indices are strong. Then, $x_{i+1}\le 2^{-2^{m+i}}x_{i}$ for $i=1,\ldots,k$, and hence, 
$$x_{k+1}\le w\cdot 2^{-(2^{m}+2^{m+1}+\cdots+2^{m+k-1})}=
w\cdot 2^{-2^{m}(2^{k}-1)}\le 4w\cdot 2^{-2^{k+1}}<4w',
$$
and $n_{k+1}\ge 2$. Thus, $x_{k+1}/N_{k+1}<w'$, and $k+1$ is weak.\smallskip

Let us next consider the subsequence $\mathcal{S}=k',k'+1,\ldots,i_0-6$. 

\noindent\emph{Claim 2:} $\mathcal{S}$ contains no subsequence of type \textbf{SS} or \textbf{SWSWS}.

\noindent Consider any weak index $i$ followed by a strong index $i+1$. Then, $N_{i+2} \ge N_i$ and $x_{i+2} \le x_i$, and hence, 
$x_{i+2}/N_{i+2} \le x_i/N_i < w'$. Thus, $i+2$ is weak. Since $\mathcal{S}$ starts with a weak index, the first part of our claim follows. For the second part, assume that $i$, $i+2$ are strong, and $i+1$ and $i+3$ are weak. Then, $N_i = N_{i+2} = N_{i+4}$, $N_{i+1}= N_i^2$, $x_{i+1} \le x_i/N_i$, $x_{i+3} \le x_{i+2}/N_{i+2}$, $x_{i+4} < x_{i+3}$, $x_{i+3} < x_{i+1}$, and hence, 
\[ \frac{x_{i+4}}{N_{i+4}}<\frac{x_{i+2}}{N_{i+2} N_{i+4}} \le \frac{x_{i+1}}{N_{i+2}^2}=\frac{x_{i+1}}{N_{i+1}} < w'.\]
Thus, $i+4$ is weak. \smallskip

\noindent\emph{Claim 3:} If $i$ is weak and $i\le i_0-6$, then $n_i\ge 2$.

\noindent Namely, if $i$ is weak and $n_i=1$, then $x_i/4=x_i/N_i<w'$, and thus, $x_{i_0-1}<w'$ because $(3/4)^{5}<1/4$. This contradicts the definition of $i_0$.\smallskip

We now partition the sequence $\mathcal{S}$ into maximal subsequences $\mathcal{S}_1,\mathcal{S}_2,\ldots,\mathcal{S}_r$, such that each $\mathcal{S}_j$, $j=1,\ldots,r$, contains no two consecutive weak elements. Then, according to our above results, each $\mathcal{S}_j$, with $j<r$, is of type \textbf{W}, \textbf{WSW}, or \textbf{WSWSW}. The last subsequence $\mathcal{S}_r$ is of type \textbf{W}, \textbf{WS}, \textbf{WSW}, \textbf{WSWS}, or \textbf{WSWSW}. Since $n_{i}\ge 2$ for all weak $i$ with $i\le i_{0}-6$, the number $n_i$ decreases by one after each $S_j$, with $j<r$. Thus, we must have $r\le n_1+k'-2$ since $n_{k'}=n_1+k'-1$, $n_{r - 1} = n_{k'} - (r-1)$, and $n_{r-1} \ge 2$. Since the length of each $\mathcal{S}_j$ is bounded by $5$, it follows that 
$$i_0= i_0-6 +6\le k'+5r+6\le k'  + 5(n_1+k'-2)+6\le 5(n_{1}+k)+k+2<8(n_{1}+k).$$
\end{proof}

We are now ready to derive an upper bound on $s_{\max}$. 

\begin{lemma}\label{lemma:bound on smax} The maximal length $s_{\max}$ of any path between splitting nodes  is bounded by $O((\log n+\log(\Gamma+ \log M(\sigma_{P}^{-1}))))$.
\end{lemma}
\begin{proof} Consider any path in the subdivision forest ending in a splitting node and otherwise containing only non-splitting nodes. Let 
$I_{1}:=(a_{1},b_{1}):=I$ to $I_{s}=(a_{s},b_{s})$  be the corresponding sequence of intervals. Then, the one-circle regions $\Delta(I_{j})$ of all intervals in the sequence contain exactly the same set of roots of $P$, and this set is non-empty. We show $s=O(\log n+\log(\Gamma+ \log M(\sigma_{P}^{-1})))$. We split the sequence into three parts:
\begin{compactenum}[(1)]
\item Let $s_1\in\{1,\ldots,s\}$ be the smallest index with $a_{s_1}\neq a_1$ and $b_{s_1}\neq b_1$. The first part consists of intervals $I_1$ to $I_{s_1 - 1}$.  
We may assume $a = a_1 = a_2 = \ldots = a_{s_1 - 1}$. We will show $s_1 = O(\log (\Gamma+\log M(\sigma_{P}^{-1}))$.
\item Let $s_2 \ge s_1$ be minimal such that either $s_2 = s$ or $w(I_{s_2})\le 2^{-13-\log n}w(I_{s_{1}})/N_{I_{s_2}}$. We will show $s_2 - s_1 = 
O(\log n+\log (\Gamma+\log M(\sigma_{P}^{-1}))$. The second part consists of intervals $I_{s_1}$ to intervals $I_{s_2 - 1}$. 
\item The third part consists of the remaining intervals $I_{s_2}$ to $I_s$. If $s_2 = s$, this part consists of a single interval. If $s_s < s$, we have 
$w(I_j)\le 2^{-13-\log n}w(I_{s_{1}})/N_{I_j}$ for all $j\ge s_2$. If $I_{j+1}$ comes from $I_j$ by a linear step, this is obvious because $w(I_{j+1}) \le w(I_j)$ and $N_{I_{j+1}} \le N_{I_j}$. If it is generated in a quadratic step, we have $w(I_{j+1}) \le w(I_j)/N_{I_j}$ and $N_{I_{j+1}} = N_{I_j}^2$. 
\end{compactenum}

In order to derive a bound on $s_1$, we appeal to Lemma~\ref{lemma:sequence}. 
If $w(I_j)/N_{I_j}\ge 4\cdot w(I_{s+1})$ for some $j$, then according to the remark following the definition of the Boundary-Test, the subdivision step from $I_{j}$ to $I_{j+1}$ is quadratic. 
However, it might also happen that the step from $I_{j}$ to $I_{j+1}$ is quadratic, and yet, $w(I_j)/N_{I_j}< 4\cdot w(I_{s+1})$. If such a $j$ exists, then let $j_0$ be the minimal such $j$; otherwise, we define $j_0=s$. In either case, $s = j_0 + O(1)$. This is clear if $s = j_0$. If $j_0 < s$, the step from $I_{j_0}$ to $I_{j_o + 1}$ is quadratic, and hence, $w(I_{j_0+1}) \le w(I_{j_0})/N_{I_{j_0}} < 4 \cdot w(I_{s+1})$, and hence, a constant number of steps suffices to reduce the width of $I_{j_0 +1}$ to the width of $I_{s+1}$. For $j=1,\ldots,j_0-1$, the sequence $(w(I_j),n_{I_j})$ coincides with a sequence $(x_j,n_j)$ as defined in Lemma~\ref{lemma:sequence}, where $w:=w(I_1)$, $w':=4w(I_{s+1})$, and $n_1=m:=n_{I_1}$. Namely, if $w(I_j)/N_{I_j}\ge w'$, we have $w(I_{j+1})\le w(I_j)/N_{I_j}$ and $n_{I_{j+1}}=1+n_{I_j}$, and otherwise, we have $w(I_{j+1})\le \frac{3}{4}\cdot w(I_j)$ and $n_{I_{j+1}}=\max(1,n_{I_j}-1)$. Hence, according to Lemma~\ref{lemma:sequence}, it follows that $j_0$ (and thus also $s$) is bounded by $$8(n_{I_1}+\log\log\max(4,w(I_1)/w(I_s)))=O(\log (\Gamma+\log M(\sigma_{P}^{-1})),$$
where we used the bounds for $N_{I_1}$, $w(I_{1})$, and $w(I_s)$ from Lemma~\ref{lemma:bounds1}.

We come to the bound on $s_2$. Observe first that $\min(|a_{1}-a_{s_1}|,|b_1-b_{s_1}|)\ge \frac{1}{8}w(I_{s_1})$.
Obviously, there exists an $s_1'=s_1+O(\log n)$, such that $w(I_{j})\le 2^{-13-\log n}
w(I_{s_1})$ for all $j\ge s_1'$. Furthermore, $N_{I_j}\le N_{\max} := 2^{O(\Gamma+\log M(\sigma_{P}^{-1}))}$ for all $j$ according to Lemma~\ref{lemma:bounds1}.
Thus, if the sequence $I_{s_1'},$ $I_{s_1'+1},\ldots$ starts with $m_{\max}:=\max(5,\log\log 
N_{\max}+1)$ or more consecutive linear subdivision steps, then $N_{I_{j'}}=4$ and $w(I_{j'})\le \frac{w(I_{s_1'})}{4}\le 2^{-13-\log n}\cdot \frac{w(I_{s_1})}{4}=2^{-13-\log n}\cdot 
w(I_{s_1})\cdot N_{I_{j'}}^{-1}$ for some $j'\le s_1'+ m_{\max}$. Otherwise, there exists a 
$j'$ with $s_1'\le j'\le s_1'+m_{\max}$, such that the step from $I_{j'}$ to $I_{j'+1}$ is 
quadratic. Since the length of a sequence of consecutive quadratic subdivision steps is also 
bounded by $m_{\max}$, there must exist a $j''$ with $j'+1\le j''\le j'+m_{\max}+1$ such that 
the step from $I_{j''-1}$ to $I_{j''}$ is quadratic, whereas the step from $I_{j''}$ to 
$I_{j''+1}$ is linear. Then, $N_{I_{j''+1}}=\sqrt{N_{I_{j''}}}=N_{I_{j''-1}}$, and 
$$w(I_{j''+1})\le \frac{3}{4}w(I_{j''})\le \frac{3w(I_{j''-1})}{4N_{I_{j''-1}}}\le 2^{-13-\log n}\cdot\frac{w(I_{s_1})}{N_{I_{j''+1}}}.$$ Hence, in any case, there exists an $s_2\le 
s_1+2m_{\max}+1$ with $w(I_{s_2})\le 2^{-13-\log n}\cdot w(I_{s_1})/N_{I_{s_2}}$. 

We next bound $s - s_2$. We only need to deal with the case that $s_2 < s$, and hence, $w(I_j)\le 2^{-13-\log n}N_{I_j}^{-1}w(I_{s_{1}})$ for all $j\ge s_2$. 
For $j \ge s_2$, we also have 
\begin{align}
|x-z_i|>2^{\log n+10} N_{I_j} w(I_j)\text{ for all }z_{i}\notin\mathcal{M}(I_{j})\text{ and all }x\in I_j.\label{outerbound}
\end{align}
From (\ref{outerbound}) and Lemma~\ref{newtonsuccess}, we conclude that the step from $I_{j}$ to $I_{j+1}$ is quadratic if $j\ge s_{2}$ and $w(I_{s})\le 2^{-13}\cdot\frac{w(I_{j})}{N_{I_{j}}}$. Again, it might also happen that there exists a $j\ge s_2$ such that the step from $I_{j}$ to $I_{j+1}$ is quadratic, and yet, $w(I_{s})> 2^{-13}\cdot\frac{w(I_{j})}{N_{I_{j}}}$. If this is the case, then we define $s_3$ to be the minimal such index; otherwise, we set $s_3:=s$. Clearly, $s = s_3 +O(1)$. We can now again apply Lemma~\ref{lemma:sequence}. 
The sequence $(w(I_{s_2+i}),n_{I_{s_2+i}})_{i=1,\ldots,s_3-s_2}$ coincides with a sequence $(x_i,n_i)_{1\le i\le s_3-s_2}$ as defined in Lemma~\ref{lemma:sequence}, where $n_1=m=n_{I_{s_2+1}}$ and $w':=2^{13}\cdot w(I_s)$. Namely, if $w(I_{s_2+i})\cdot N_{I_{s_2+i}}^{-1}\ge w'$, then $w(I_{s_2+i+1})\le w(I_{s_2+i})\cdot N_{s_2+i}^{-1}$ and $n_{I_{s_2+i+1}}=1+n_{I_{s_2+i}}$, whereas we have $w(I_{s_2+i+1})\le \frac{3}{4}w(I_{s_2+i})$ and 
$n_{I_{s_2+i+1}}=\max(n_{I_{s_2+i}}-1,1)$ for $w(I_{s_2+i})\cdot N_{I_{s_2+i}}^{-1}< w'$. It follows that $s_3-s_2$ is bounded by $8(n_1+\log\log\max(4,w(I_{s_2+1})/w'))=O(\log(\Gamma+\log M(\sigma_P^{-1})))$. 
\end{proof}

\MS{The following theorem now follows immediately from Lemma~\ref{lemma:bound on smax} and our considerations from the beginning of Section~\ref{sec:size-subdivision-tree}. For the bound on the size of the subdivision forest when running the algorithm on a square-free polynomial $P$ of degree $n$ and with integer coefficients of bit-size less than $\tau$, we use that $\gamma=O(\log\Gamma)$ and $\log n+\log(\Gamma+\log M(\sigma_{P}^{-1}))=O(\log (n\tau))$, which is due to Lemma~\ref{lem:usefulineq}.}

\begin{theorem}\label{thm:treesize} Let $K = \log n+\log(\Gamma+\log M(\sigma_{P}^{-1}))$. 
The size $\abs{\mathcal{T}}$ of the subdivision forest is 
\[ O\left( \sum_{k=0}^{2\gamma+2} \left(1 + \abs{\mathcal{M}(\mathcal{I}_{k})} \cdot K \right) \right)
=O(n K ),\]
where $\mathcal{I}_{k}$ are the intervals as defined in (\ref{def:startintervals}), and $\mathcal{M}(I_{k})$ denotes the set of all roots contained in the one-circle region $\Delta(\mathcal{I}_{k})$ of $\mathcal{I}_{k}$. \MS{In the case, where the input polynomial has integer coefficients of bit-size less than $\tau$, the above bound simplifies to
\[
O\left( \sum_{k=0}^{2\gamma+2} \left(1 + \abs{\mathcal{M}(\mathcal{I}_{k})} \cdot \log(n\tau) \right) \right)
=O(n\log(n\tau)).
\]}
\end{theorem}

\begin{figure}[t]
\begin{center}
\includegraphics[width=0.6\textwidth]{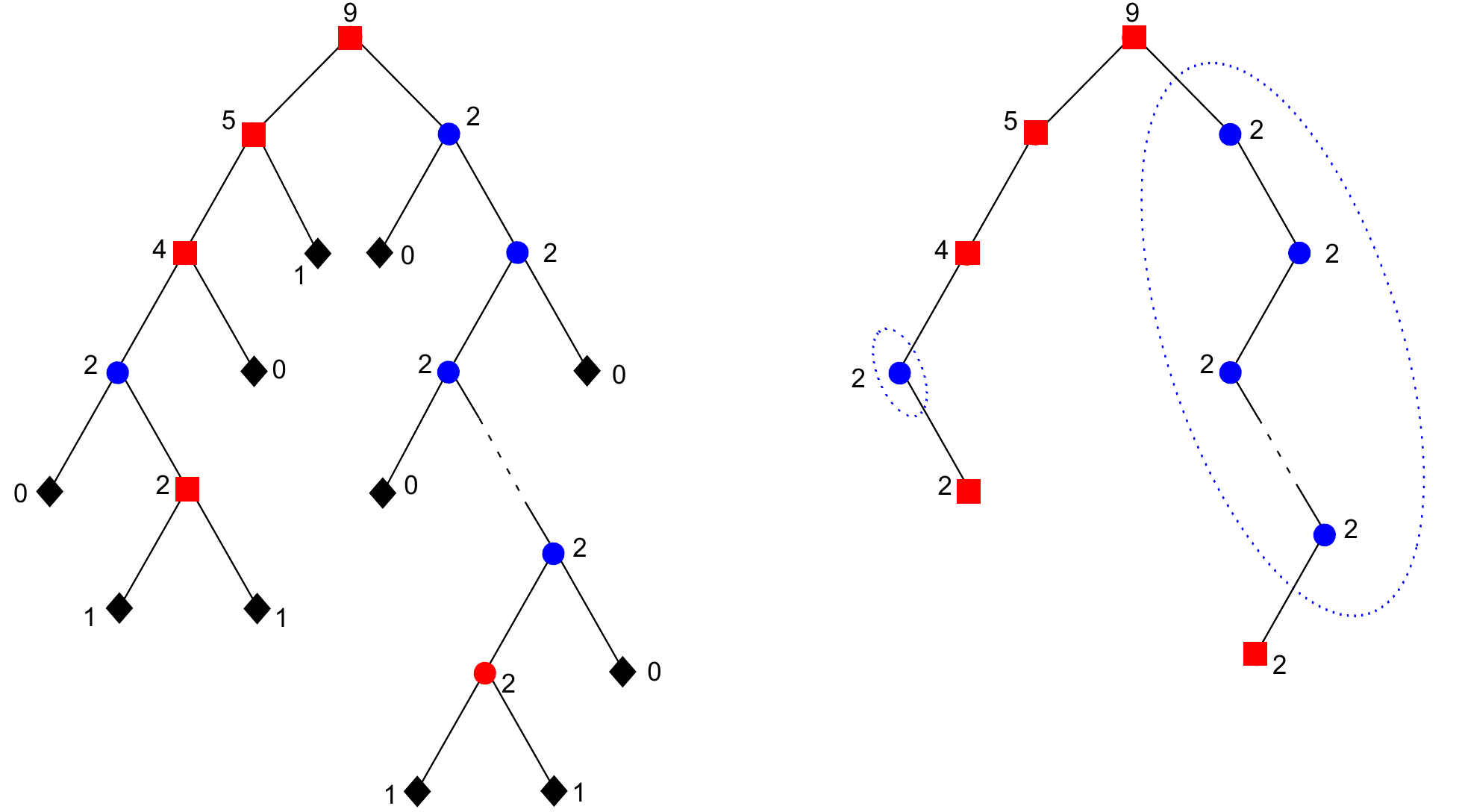}\end{center}
\caption{\label{tree} \MS{The left figure shows the subdivision tree $\mathcal{T}_k$ rooted at the interval $\mathcal{I}_k$, where, for each node $I$, the number $\var(f,I)$ of sign variations is given (e.g.~$\var(P,\mathcal{I}_k)=9$). Special nodes and terminal nodes are indicated by red squares and black diamonds, respectively. Blue dots indicate \emph{ordinary} nodes, which are neither special nor terminal. The right figure shows the subtree $\mathcal{T}_k^*$ obtained by removing all terminal nodes. $\mathcal{T}_k^*$ partitions into special nodes and chains of ordinary nodes connecting two consecutive special nodes.}}
\end{figure}
\MS{
We further provide an alternative bound on the number of iterations needed by the algorithm $\textsc{ANewDSC}$, which is stated in terms of the number of sign variations $v_k:=\operatorname{var}(f,\mathcal{I}_k)$ of $f$ on the initial intervals $\mathcal{I}_k$ instead of the values $|\mathcal{M}(I_{k})|$:
\begin{theorem}\label{thm:sizetreealt}
Let $K = \log n+\log(\Gamma+\log M(\sigma_{P}^{-1}))$. The size $\abs{\mathcal{T}_k}$ of the subdivision tree $\mathcal{T}_k$ rooted at the interval $\mathcal{I}_k$ is $O(\operatorname{var}(P,\mathcal{I}_k)\cdot K),$ and the size of the subdivision forest $\mathcal{T}$ is 
\[ O\left( \sum\nolimits_{k=0}^{2\gamma+2} \left(1 + \operatorname{var}(P,\mathcal{I}_k) \cdot K \right) \right).
\]
\end{theorem}
\begin{proof}
Let $I_1,\ldots,I_s$ be a sequence of intervals produced by the algorithm such that $I_s\subset I_{s-1}\subset\cdots\subset I_1\subset\mathcal{I}_k$, and $v=\operatorname{var}(P,I_1)=\cdots=\operatorname{var}(P,I_s)$ for some $v\in\N$. Notice that $v\le \operatorname{var}(P,\mathcal{I}_k)$ according to Theorem~\ref{subad}. We first show that $s=O(K)$. For $j=1,\ldots,s$, let $L_n(I_j)$ and $A_n(I_j)$ be the Obreshkoff lens and the Obreshkoff area of $I_j$, respectively, as defined in Figure~\ref{fig:Obreshkoff}. Then, from Theorem~\ref{Obreshkoff}, we conclude that each $L_n(I_j)$ contains at most $v$ roots of $P$, whereas each $A_n(I_j)$ contains at least $v$ roots of $P$. Using the same argument in as in the proof of Lemma~\ref{lemma:bound on smax} shows that either $s=O(K)$ or there exists an $s_1=O(\log(\Gamma+\log M(\sigma_P^{-1})))$ such that $I_1$ does not share any endpoint with $I_{s_1}$, and the distance between any two endpoints of $I_1$ and $I_{s_1}$ is at least as large as the width $w(I_{s_1})$ of $I_{s_1}$. Hence, after $\ell=O(\log n)$, further subdivision steps, the one-circle region $\Delta(I_{s_1+\ell})$ of $I_{s_1+\ell}$ is completely contained in the lens $L_n(I_1)$, and thus the one-circle region $\Delta(I_j)$ of any interval $I_j$, with $j\ge s_2:=s_1+\ell=O(K)$, is contained in $L_n(I_1)$. Since $L_n(I_1)$ contains at most $v$ roots, we conclude that $|\mathcal{M}(I_{j})|\le v$ for all $j\ge s_2$. The same argument also shows that either $s=O(K)$ or that there exists an $s_3$, with $s_3\le s$ and $s-s_3=O(K)$ such that the one-circle region $\Delta(I_j)$ of each interval $I_j$, with $j\le s_3$, contains the Obreshkoff area $A_n(I_s)$. Thus, we have $|\mathcal{M}(I_j)|\ge v$ for all $j\le s_3$ as $A_n(I_s)$ contains at least $v$ roots. We can assume that $s_2\le s_3$ as, otherwise, $s\le s_2+(s-s_3)=O(K)$. Since $|\mathcal{M}(I_j)|\le v$ for all $j\ge s_2$ and $|\mathcal{M}(I_j)|\ge v$ for all $j\le s_3$, we have $|\mathcal{M}(I_j)|=v$ for $j=s_2,\ldots,s_3$. Then, from Lemma~\ref{lemma:bound on smax}, we conclude that $s_3-s_2=O(K)$, and thus also $s=O(K)$. Now, consider the sub-tree $\mathcal{T}_k^*$ of $\mathcal{T}_k$ obtained from $\mathcal{T}_k$ after removing all intervals $I\in\mathcal{T}_k$ with $\operatorname{var}(P,I)\le 1$; see also Figure~\ref{tree}. Then, $|\mathcal{T}_k|\le 2\cdot|\mathcal{T}_k^*|+1$, and thus it suffices to bound the size of $\mathcal{T}_k^*$. We 
call an interval $I\in\mathcal{T}^*_k$ \emph{special} if it is either the root $\mathcal{I}_k$ of $\mathcal{T}^*_k$ or if none of its children yields the same number of sign variations as $I$. Then, the following argument shows that the number of special intervals is at most $\var(P,\mathcal{I}_k)$. We can assume that $\operatorname{var}(P,\mathcal{I}_k)>1$ as otherwise $\mathcal{T}_k$ has size $1$. Now, let $\mathcal{A}_k\subset \mathcal{A}$ be the set of all active intervals in $\mathcal{T}_k$ produced by the algorithm in a certain iteration, and define 
\[\mu:=\sum_{I\in \mathcal{A}: I\subseteq\mathcal{I}_k\text{ and }\operatorname{var}(P,I)\ge 1} (\operatorname{var}(P,I)-1),\]
then $\mu$ decreases by at least one at each special interval, whereas it stays invariant at all other intervals. Since, $\mu=\operatorname{var}(P,\mathcal{I}_k)-1$ at the beginning, and $\mu\ge 0$ in each iteration, it follows that there can be at most $\operatorname{var}(f,\mathcal{I}_k)$ special intervals in $\mathcal{T}^*_k$. Notice that $\mathcal{T}_k^*$ splits into special intervals and chains of intervals $I_1,\ldots,I_s$ connecting two consecutive special intervals $I$ and $J$, that is, $I\subset I_s\subset\cdots\subset I_1\subset J$, and there exists no special interval $I'$ with $I\subset I'\subset J$. Since $\var(P,I_j)$ is invariant for all $j=1,\ldots,s$, our above considerations show that each such chain has length $O(K)$. Hence, the claim follows.
\end{proof}
\noindent\emph{Remark.} 
For polynomials $P=P_0+\cdots+P_n x^n\in\mathbb{Z}[x]$ with integer coefficients, 
we remark that the above bound can also be stated in terms of the number $v^+(P):=\operatorname{var}(P_0,\ldots,P_n)$ of sign 
variations of the coefficient sequence of $P$, and the number $v^-(P):=\operatorname{var}(P_0,-P_1,P_2,\ldots,(-1)^n\cdot P_n)$ of sign 
variations of the coefficient sequence of $P(-x)$: After removing a suitable factor $x^i$, we are left 
with a polynomial $\bar{P}=P\cdot x^{-i}$, which fulfills $|\bar{P}(0)|\ge 1$, $v^-(\bar{P})=v^-(P)$, and $v^+(\bar{P})=v^+(P)$. Let us now estimate the size of the subdivision forest induced by our algorithm when applied to the polynomial $\bar{P}$. In the definition (\ref{def:startintervals}) of the intervals $\mathcal{I}_k=(s_k^*,s_{k+1}^*)$, we may choose 
$s_{\gamma+1}^*=s_{\gamma}=0$ as this does not harm the requirements from (\ref{cond:start1}) posed on the 
intervals $\mathcal{I}_k$. 
Now, because of the sub-additivity of the function $\var(\bar{P},.)$, it follows that $\sum_{k=0}^{\gamma+1}\var(\bar{P},\mathcal{I}_k)\le \var(\bar{P},(s_0^*,s_{\gamma+1}^*))$ and $\sum_{k=\gamma+1}^{2\gamma+2}\var(\bar{P},\mathcal{I}_k)\le \var(\bar{P},(s_{\gamma+1}^*,s_{2\gamma+2}))$. The intervals $I^-:=(s_0^*,s_{\gamma+1}^*)$ and $I^+:=(s_{\gamma+1}^*,s_{2\gamma+2})$ are contained in $\R_{<0}$ and in $\R_{>0}$, respectively. Hence, it holds that $\var(\bar{P},I^-)\le v^-(\bar{P})$, and that $\var(\bar{P},I^+)\le v^+(\bar{P})$. Theorem~\ref{thm:sizetreealt} then implies that the subdivision forest has size $O(\gamma+K\cdot(v^-(P)+v^+(P)))=O((v^-(P)+v^+(P)+1)\log(n\tau))$, where $\tau$ bounds the bit size of the coefficients of $P$. In particular, we obtain:
\begin{theorem}
\label{thm:treesizesparse}
Let $P\in\Z[x]$ be a square-free polynomial of degree $n$ and with integer coefficients of bit size less than $\tau$. Let $k$ be the number of non-zero coefficients of $P$. Then, for isolating all real roots of $P$, \Kurt{$\AD$ generates a tree of size  $O(k\log(n\tau))$.}
\end{theorem}
\Kurt{Notice that, in the special case, where $P$ is a sparse integer polynomial with only $(\log(n\tau))^{O(1)}$ non-vanishing coefficients, our algorithm generates a tree of size $(\log(n\tau))^{O(1)}$. An illustrative example of the latter kind are Mignotte polynomials of the form $P=x^n-(a\cdot x-1)^2$, with $a$ an integer of bit size less than $\tau$. In order to isolate the real roots of $P$, our algorithm generates a tree of size $(\log(n\tau))^{O(1)}$, whereas bisection methods, such as the classical Descartes method, generate a tree of size $\Omega(n\tau)$.}
}

\subsection{Bit Complexity}\label{sec:bit-complexity}

In order to bound the bit complexity of our algorithm, we associate a root $z_i$ of $P$ with every interval $I$ in the subdivision forest and argue that the cost (in number of bit operations) of  processing $I$ is 
\begin{equation}\label{bit bound}       \tilde{O}(n(n+\tau_P+n \log M(z_{i})+\log M(P'(z_{i})^{-1}))). \end{equation}
The association is such that each root of $P$ is associated with at most $O(s_{\max} \log n+\log\Gamma)$ intervals, and hence, the total bit complexity can be bounded by summing the bound in (\ref{bit bound}) over all roots of $P$ and multiplying by $s_{\max} \log n+\log\Gamma$. Theorem~\ref{main:theorem} results.

We next define the mapping from $\mathcal{T}$ to the set of roots of $P$. 
Let $I$ be any non-terminal interval in $\mathcal{T}$.
We define a path of intervals starting in $I$. Assume we have extended the path to an interval $I'$. The path ends in $I'$ if $I'$ is strongly splitting or if $I'$ is terminal; see the introduction of Section~\ref{sec:size-subdivision-tree} for the definitions. If $I'$ has a child $I''$ with $\mathcal{M}(I')=\mathcal{M}(I'')$, the path continues to this child. If $I'$ has two children $J_1$ and $J_2$ with $\mathcal{M}(J_1)\cup\mathcal{M}(J_2)=\mathcal{M}(I')$, and both $\mathcal{M}(J_1)$ and $\mathcal{M}(J_2)$, are nonempty (and hence $\max(\abs{\mathcal{M}(J_1)},\abs{\mathcal{M}(J_1)}) < \abs{\mathcal{M}(I')}$), the path continues to the child with smaller value of $\abs{\mathcal{M}(*)}$. Ties are broken arbitrarily, but consistently, i.e., all paths passing through $I'$ make the same decision. Let $J$ be the last interval of the path starting in $I$. Then, the one-circle region of $J$ contains at least one root that is not contained in the one-circle region of any child of $J$. We call any such root $z\in\mathcal{M}(J)\subset\mathcal{M}(I)$ \emph{associated with $I$}.
With terminal intervals $I$ that are different from any $\mathcal{I}_k$, we associate the same root as with the parent interval. With terminal intervals $\mathcal{I}_k$, we associate an arbitrary root. More informally, with an interval $I$, we associate a root $z_i\in\mathcal{M}(I)$, which is either "discarded" or isolated when processing the last interval of the path starting in $J$.

The path starting in an interval has length at most $s_{\max} \cdot\log n$ as there are at most $s_{\max}$ intervals $I$ with the same set $\mathcal{M}(I)$, and $\abs{\mathcal{M}(*)}$ shrinks by a factor of at least $1/2$ whenever the path goes through a splitting node. There are at most $2\log\Gamma+1$ intervals $\mathcal{I}_k$ with which any root can be associated, and each root associated with an interval $I\subsetneq \mathcal{I}_k$ cannot be associated with any interval $I'\subsetneq \mathcal{I}_{k'}$ with $k\neq k'$ as the corresponding one-circle regions are disjoint.
As a consequence, any root of $P$ is associated with at most $s_{\max}\cdot \log n+2\log\Gamma+1=O(s_{\max}\log n+\log\Gamma)$ intervals. 

We next study the complexity of processing an interval $I$. We 
first derive a lower bound for $|P|$ at the subdivision points that are considered when processing $I$. We introduce the following notation: For an interval $I=(a,b)\in\mathcal{T}$, we call a point $\xi$ \emph{special with respect to $I$} (or just special if there is no ambiguity) if $\xi$ is
\begin{compactenum}
\item[(P1)] an endpoint of $I$, that is, $\xi=a$ or $\xi=b$.
\item[(P2)] an admissible point $m^{*}\in \multipoint{m(I)}{w(I)\cdot 2^{-\lceil \log n+2 \rceil}}$ as computed in the $1$-Test.
\item[(P3)] an admissible point $\xi_{j}^{*}\in\multipoint{\xi_{j}}{w(I)\cdot 2^{-\lceil 5+\log n\rceil}}$ as computed in (\ref{Newtonmultipoint1}) in the Newton-Test.
\item[(P4)] an admissible point $a_{j_{1},j_{2}}^{*}\in \multipoint{a_{j_{1},j_{2}}}{2^{-\lceil 5+\log n\rceil}\cdot\frac{w(I)}{N_{I}}}$ or an admissible point $b_{j_{1},j_{2}}^{*}\in\multipoint{b_{j_{1},j_{2}}}{2^{-\lceil 5+\log n\rceil}\cdot\frac{w(I)}{N_{I}}}$ as computed in (\ref{Newtonmultipoint2}) in the Newton-Test.
\item[(P5)] an admissible point $m_{\ell}^{*}\in\multipoint{m_{\ell}}{2^{-\lceil 2+\log n\rceil}\cdot\frac{w(I)}{N_{I}}}$ or an admissible point $m_{r}^{*}\in\multipoint{m_{r}}{2^{-\lceil 2+\log n\rceil}\cdot\frac{w(I)}{N_{I}}}$ as computed in (\ref{Boundarymultipoint}) in the Boundary-Test.
\end{compactenum}
For intervals $I$ with $\var(P,I)=0$, we have only special points of type (P1), and for intervals with $\var(P,I)=1$, we have only special points of type (P1) and type (P2). For other intervals, we consider all types. The following lemma provides a lower bound for the absolute value of $P$ at special points.

\begin{lemma}\label{lem:inequalityspecial}
Let $I\in\mathcal{T}$ be an arbitrary interval, and let $\xi$ be a special point with respect to $I$. If $\Delta(I)$ contains a root of $P$, then
\begin{align}\label{special:lowerbound}
|P(\xi)|>2^{-40n\log n-2\tau_{P}}\cdot M(z_{i})^{-5n}\cdot \min(1,\sigma(z_i,P))^{5}\cdot \min(1,|P'(z_{i})|)^{5}
\end{align}
for all $z_{i}\in\Delta(I)$. If $\Delta(I)$ contains no root, then $\xi$ fulfills the above inequality for all roots contained in $\Delta(J)$, where $J$ is the parent of $I$.\footnote{If $I=\mathcal{I}_k$ for some $k$ and $\Delta(I)$ contains no root, then $z_i$ can be chosen arbitrarily.}
\end{lemma}
\begin{proof}
We will prove the claim via induction on the depth $k$ of an interval $I$, where the depth of the intervals $\mathcal{I}_{k}$ is one. According to (\ref{cond:start1}), the endpoints of $\mathcal{I}_{k}$ fulfill inequality (\ref{special:lowerbound}). Now, if $\xi$ is a special point (with respect to $\mathcal{I}_{k}$) of type (P2), then 
$$|P(\xi)|\ge \frac{1}{4}\cdot \max \set{|P(x)|}{x\in \multipoint{m(\mathcal{I}_{k})}{w(\mathcal{I}_{k})\cdot 2^{-\lceil \log n+2 \rceil}}}.$$
Since at least one of the points in $\multipoint{m(\mathcal{I}_{k})}{w(\mathcal{I}_{k})\cdot 2^{-\lceil \log n+2 \rceil}}$ has distance more than $w(\mathcal{I}_{k})\cdot 2^{-\lceil \log n+3 \rceil}>\frac{1}{8n}$ to all roots of $P$, it follows that $|P(\xi)|>\frac{|P_{n}|}{(8n)^{n}}> 2^{-8n\log n}$, where we use that $|P_n|\ge 1/4$ and $w(\mathcal{I}_{k})\ge 1$. An analogous argument yields $|P(\xi)|>\frac{1}{(64n)^{n+1}}\ge 2^{-8n\log n}$ if $\xi$ is a special point of type (P3) to (P5), where we additionally use $N_{\mathcal{I}_{k}}=4$ for all $k$.\smallskip

For the induction step from $k$ to $k+1$, suppose that $I=(a,b)$ is an interval of depth $k+1$ with parent interval $J=(c,d)$ of depth $k$. We distinguish the following cases.\smallskip

\noindent\emph{The point $\xi$ is a special point of type (P1):}
The endpoints of $I$ are either subdivision points (as constructed in Steps (Q) or (L) in our algorithm) or endpoints of some interval $J'\in\mathcal{T}$ with $I\subseteq J'$. Hence, they are special points with respect to an interval $J'$ that contains $I$. Thus, from our induction hypothesis (the depth of $J'$ is smaller than or equal to $k$) and the fact that $\Delta(I)\subseteq\Delta(J')$, it follows that the inequality (\ref{special:lowerbound}) holds for all admissible points $\xi$ of type (P1).\smallskip

\noindent\emph{The point $\xi$ is a special point of type (P2)}:
Since $\var(P,J)\ge 2$, the two-circle region of $J$ contains at least two roots of $P$, and thus the disk $\Delta:=\Delta_{2w(J)}(m(I))$ with radius $2w(J)$ centered at the midpoint $m(I)$ of $I$ contains at least two roots. This shows that $\sigma(z_i,P)<2w(J)$ for any root $z_{i}\in\Delta$. With $\epsilon:=w(I)\cdot 2^{-\lceil \log n+2 \rceil}$ and $K:=\frac{w(J)}{w(I)}\cdot 2^{\lceil\log n+3\rceil}$, we can now use Lemma~\ref{lem:boundonmax} to show that
\begin{align}\label{inequ1:special}
|P(\xi)|>2^{-4n-1}\cdot K^{-\mu(\Delta)-1}\cdot\sigma(z_i,P)\cdot |P'(z_{i})|>2^{-8n\log n}\cdot \left(\frac{w(J)}{w(I)}\right)^{-\mu(\Delta)-1}\cdot\sigma(z_i,P)\cdot |P'(z_{i})|,
\end{align}
where $z_{i}$ is any root in $\Delta$, and $\mu(\Delta)$ denotes the number of roots contained in $\Delta$. If the subdivision step from $J$ to $I$ is linear, then $w(J)/w(I)\in [4/3,4]$, and thus the bound in (\ref{special:lowerbound}) is fulfilled. Otherwise, we have $w(J)/w(I)\in [N_{I},4N_{I}]=[N_{J}^{2},(2N_{J})^{2}]$. In addition, $w(J)\le 4w(\mathcal{I}_{k})/N_{J}$, where $k$ is the unique index with $J\subseteq\mathcal{I}_{k}$. We conclude that $N_J\le 4w(\mathcal{I}_k)/w(J)$, and thus
\begin{align}\label{ineq:w(J)}
\frac{w(J)}{w(I)}\le (2N_J)^2\le \left(\frac{8w(\mathcal{I}_{k})}{w(J)}\right)^{2}\le \frac{2^{12}\cdot M(z_{i})^{2}}{w(J)^2}
\end{align}
since $w(\mathcal{I}_{k})\le 8M(x)^{2}$ for all $x\in \mathcal{I}_{k}$. Furthermore, since 
\begin{align*}
|P'(z_{i})|&=n|P_{n}|\cdot\prod_{z_{j}\in\Delta:z_{j}\neq z_{i}}|z_{i}-z_{j}|\cdot\prod_{z_{j}\notin\Delta}|z_{i}-z_{j}|\le n\cdot (2w(J))^{\mu(\Delta)-1}\cdot \prod_{z_{j}\notin\Delta}|z_{i}-z_{j}|\\
&\le n\cdot 2^{n}\cdot w(J)^{\mu(\Delta)-1}\cdot \prod_{z_{j}\notin\Delta}M(z_{i}-z_{j})\le n\cdot 2^{n}\cdot w(J)^{\mu(\Delta)-1}\cdot \frac{\Mea(P(z_{i}-x))}{|P_n|}\\
&\le n\cdot 2^{n}\cdot w(J)^{\mu(\Delta)-1}\cdot 2^{\tau_{P}}2^{n}M(z_{i})^{n}<2^{4n+\tau_{P}}\cdot M(z_{i})^{n}\cdot w(J)^{\mu(\Delta)-1},
\end{align*}
it follows that 
\begin{align*}
w(J)^{-\mu(\Delta)-1}=w(J)^{-\mu(\Delta)+1}w(J)^{-2}<\frac{2^{8n+\tau_{P}}\cdot M(z_{i})^{n}}{\sigma(z_i,P)^{2}\cdot |P'(z_{i})|},
\end{align*}
where we used that $2w(J)>\sigma(z_i,P)$ in order to bound $w(J)^{-2}$.
Hence, it follows from (\ref{ineq:w(J)}) that 
\begin{align*}
\left(\frac{w(J)}{w(I)}\right)^{-\mu(\Delta)-1}&\ge \left(2^{12}\cdot M(z_{i})^{2}\cdot w(J)^{-2}\right)^{-\mu(\Delta)-1}\\
&\ge 2^{-12(n+1)}\cdot M(z_{i})^{-2(n+1)}\cdot 2^{-16n-2\tau_{P}}\cdot M(z_{i})^{-2n}\sigma(z_i,P)^{4}\cdot |P'(z_{i})|^{4}\\
&>2^{-32n-2\tau_{P}}\cdot M(z_{i})^{-5n}\cdot \sigma(z_i,P)^{4}\cdot |P'(z_{i})|^{4}.
\end{align*}
Plugging the latter inequality into (\ref{inequ1:special}) eventually yields
\[
|P(\xi)|>2^{-8n\log n}\cdot \left(\frac{w(J)}{w(I)}\right)^{-\mu(\Delta)-1}\cdot\sigma(z_i,P)\cdot |P'(z_{i})|>2^{-40n\log n-2\tau_{P}}\cdot M(z_{i})^{-5n}\cdot \sigma(z_i,P)^{5}\cdot |P'(z_{i})|^{5}.
\]
Thus, $\xi$ fulfills the bound (\ref{special:lowerbound}).\smallskip

\noindent\emph{The point $\xi$ is a special point of type (P3)}: The same argument as in the preceding case also works here.
Namely, each disk $\Delta:=\Delta_{2w(J)}(\xi_j)$ with radius $2w(J)$ centered at the point $\xi_j$ contains at least two roots, and thus we can use Lemma~\ref{lem:boundonmax} with $\epsilon:=w(I)\cdot 2^{-\lceil \log n+5 \rceil}$ and $K:=\frac{w(J)}{w(I)}\cdot 2^{\lceil\log n+6\rceil}$.\smallskip

\noindent\emph{The point $\xi$ is a special point of type (P4)}:
The Newton-Test is only performed if the $0$-Test and the $1$-Test have failed. Hence, we must have $\var(P,I)\ge 2$, and thus, each disk $\Delta:=\Delta_{2w(I)}(x_0)$ with radius $2w(I)$ centered at any point $x_0\in I$  contains at least two roots. We use this fact $x_0=a_{j_1,j_2}$ and $x_0=b_{j_1,j_2}$ and obtain, 
using Lemma~\ref{lem:boundonmax} with $\epsilon:=\frac{w(I)}{N_I}\cdot 2^{-\lceil \log n+5 \rceil}$ and $K:=N_I\cdot 2^{\lceil\log n+4\rceil}$, 
\begin{align}\label{inequ4:special}
|P(\xi)|>2^{-4n-1}\cdot K^{-\mu(\Delta)-1}\cdot\sigma(z_i,P)\cdot |P'(z_{i})|>2^{-8n\log n}\cdot N_I^{-\mu(\Delta)-1}\cdot\sigma(z_i,P)\cdot |P'(z_{i})|.
\end{align}
If the subdivision step from $J$ to $I$ is linear, then $N_I\le N_J\le  4w(\mathcal{I}_{k})/w(J)$, where $k$ is the unique index with $J\subseteq\mathcal{I}_{k}$. If the step from $J$ to $I$ is quadratic, then $N_I= N_J^2\le  (4w(\mathcal{I}_{k})/w(J))^2$. Now, the same argument as in the type (P2) case (see (\ref{ineq:w(J)}) and the succeeding computation) shows that 
\begin{align*}
N_I^{-\mu(\Delta)-1}&\ge \left(\frac{4w(\mathcal{I}_k)}{w(J)}\right)^{-\mu(\Delta)-1}\ge \left(2^{10}\cdot M(z_i)^2\cdot w(J)^{-2}\right)^{-\mu(\Delta)-1}\\
&\ge 2^{-32n-2\tau_{P}}\cdot M(z_{i})^{-5n}\cdot \sigma(z_i,P)^{4}\cdot |P'(z_{i})|^{4},
\end{align*}
and thus
\[
|P(\xi)|>2^{-8n\log n}\cdot N_I^{-\mu(\Delta)-1}\cdot\sigma(z_i,P)\cdot |P'(z_{i})|>2^{-40n\log n-2\tau_{P}}\cdot M(z_{i})^{-5n}\cdot \sigma(z_i,P)^{5}\cdot |P'(z_{i})|^{5}.
\]\smallskip

\noindent\emph{The point $\xi$ is a special point of type (P5)}: The same argument as in the previous case works since the 
Boundary-Test is only applied if $\var(P,I)\ge 2$.
\end{proof}

We can now derive our final result on the bit complexity of $\AD$:

\begin{theorem}[Restatement of Theorem~\ref{thm: main}]\label{main:theorem} Let $P = P_n x^n + \ldots + P_0 \in \R[x]$ be a real polynomial with $1/4 \le \abs{P_n} \le 1$. 
The algorithm $\AD$ computes isolating intervals for all real roots of $P$ with a number of bit operations bounded by
\begin{align}\label{result:bitcomplexity:isolation} 
&\tilde{O}(n\cdot(n^2+n \log \Mea(P)+\sum_{i=1}^n \log M(P'(z_i)^{-1}))) \\
                              &\qquad= \tilde{O}(n(n^2 + n \log \Mea(P) + \log M(\Disc(P)^{-1}))).
\end{align}
The coefficients of $P$ have to be approximated with quality
\[ \tilde{O}(n+\tau_{P}+\max_{i}(n\log M(z_{i})  + \log M(P'(z_{i})^{-1}) )). \]
\end{theorem}

\begin{proof}
We first derive an upper bound on the cost for processing an interval $I\in\mathcal{T}$. Suppose that $\Delta(I)$ contains at least one root: When processing $I$, we consider a constant number of special points $\xi$ with respect to $I$. Since each of these points fulfills the inequality (\ref{special:lowerbound}), we conclude from Lemma~\ref{lem:apxmultipointeval} that the computation of all special points $\xi$ uses 
\begin{align}\label{complexityatnode}
\tilde{O}(n(n+\tau_P+n\log M(z_i)+\log M(\sigma(z_i,P)^{-1})+\log M(P'(z_i)^{-1})))
\end{align}
bit operations, where $z_i$ is any root contained in $\Delta(I)$. We remark that when applying Lemma~\ref{lem:apxmultipointeval}, we used (\ref{cond:start1}) which implies that $\log M(x)\le 2(1+\log M(z_i))$ for all $x\in I$ and all $z_i\in\Delta(I)$.

In addition, Corollary~\ref{cor:cost0test} and Corollary~\ref{cor:cost1test} yield the same complexity bound as stated in (\ref{complexityatnode}) for each of the considered $0$-Tests and $1$-Tests. Since we perform only a constant number of such tests for $I$, the bound in (\ref{complexityatnode}) applies to all $0$-Tests and $1$-Tests.

It remains to bound the cost for the computation of the values $\tilde{\lambda}_{j_1,j_2}$ in the Newton-Test: According to the remark following Step (2.2) in the description of the Newton-Test, it suffices to evaluate $P$ and $P'$ at the points $\xi_{j_1}^*$ and $\xi_{j_2}^*$ to an absolute precision of
\[
O(\log n+\log N_{I}+\log M(P(\xi^{*}_{j_{1}})^{-1}+\log M(P(\xi^{*}_{j_{2}})^{-1})+\log M(w(I)^{-1})+\log M(w(I))).
\]
Thus, according to Lemma~\ref{lem:singlepolyeval} and Lemma~\ref{lem:inequalityspecial}, the total cost for this step is bounded by
\[
\tilde{O}(n(n+\tau_P+n\log M(z_i)+\log M(\sigma(z_i,P)^{-1})+\log M(P'(z_i)^{-1}))
\]
bit operations, where $z_i$ is any root contained in $\Delta(I)$. Here, we used the fact that $2w(I)>\sigma(z_i,P)$ (notice that $\var(P,I)\ge 2$, and thus, the two-circle region of $I$ contains at least two roots) and that $\log N_I=O(\log M(\sigma(z_i,P)^{-1})+\log M(z_i))$ as shown in the proof of Lemma~\ref{lem:inequalityspecial}.
In summary, for any interval $I$ whose one-circle region contains at least one root, the cost for processing $I$ is bounded by (\ref{complexityatnode}). A completely analogous argument further shows that, for intervals $I$ whose one-circle region does not contain any root, the cost for processing $I$ is also bounded by (\ref{complexityatnode}), where $z_i$ is any root in $\Delta(J)$ and $J\in\mathcal{T}$ is the parent of $I$.

Since we can choose an arbitrary root $z_i\in\mathcal{M}(I)$ (or $z_i\in\mathcal{M}(J)$ for the parent $J$ of $I$ if $\mathcal{M}(I)$ is empty) in the above bound, we can express the cost of processing an interval $I$ in terms of the root associated with $I$. Since any root of $P$ has at most $O(s_{\max} \log n+\log\Gamma)$ many roots associated with it, the total cost for processing all intervals is bounded by 
\[
\tilde{O}(n\cdot(n^2+n \log \Mea(P)+\sum_{i=1}^n\log M(\sigma(z_i,P)^{-1})+\sum_{i=1}^n \log M(P'(z_i)^{-1})))
\]
bit operations, where we used that $\sum_{i=1}^n \log M(z_i)=\log\frac{\Mea(P)}{|P_n|}$, $\tau_P\le n+\log\Mea(P)$, and that the factor $s_{\max} \log n+\log\Gamma$ is swallowed by $\tilde{O}$. We can further discard the sum $\sum_{i=1}^n\log M(\sigma(z_i,P)^{-1})$ in the above complexity bound. Namely, if $z_k$ denotes the root with minimal distance to $z_i$, then (we use inequality (\ref{measure of shifted polynomial}))
$$\sigma(z_i,P)\ge \frac{|P'(z_i)|}{|P_n|\cdot\prod_{j\neq k,i}|z_j-z_i|}\ge\frac{|P'(z_i)|}{\Mea(P(z-z_i))}\ge\frac{|P'(z_i)|}{2^{\tau_P}\cdot 2^{n}\cdot M(z_i)^n},$$
and thus, $\sum_{i=1}^n\log M(\sigma(z_i,P)^{-1})=O(n^2+n \log \Mea(P)+\sum_{i=1}^n\log M(P'(z_i)^{-1})$.
Since the cost for the computation of $\Gamma$ is bounded by $\tilde{O}(n^2\Gamma_P)=\tilde{O}(n^2\cdot M(\log \Mea( P)))$ bit operations, the bound follows.

For the alternative bound, we use inequalities (\ref{upper bound on sum of inverse derivatives}) and (\ref{taup in terms of measure}).
\end{proof}

For the special case where the input polynomial $p$ has integer coefficients, we can specify the above complexity bound to obtain the following result:

\begin{corollary}[Restatement of Theorem~\ref{thm: main integer case}]\label{complexity:integer}
For a polynomial $p(x)=p_n\cdot x^n+\cdots+p_0\in\Z[x]$ with integer coefficients of absolute value $2^\tau$ or less, the algorithm $\AD$ computes isolating intervals for all real roots of $p$ with $\tilde{O}(n^3+n^2\tau)$ bit operations. \Kurt{If $p$ has only $k$ non-vanishing coefficients, the bound becomes $\tilde{O}(n^2(k+\tau))$ bit operations.}
\end{corollary}
\begin{proof}
We first compute a $t\in\N$ with $2^{t-1}\le|p_n|< 2^t$. Then, we apply $\AD$ to the polynomial $P:=2^{-t}\cdot p$ whose leading coefficient has absolute value between $1/2$ and $1$. The complexity bound now follows directly from Theorem~\ref{main:theorem}, where we use that $\tau_P\le\tau$, $\Mea(P)\le\Mea(p) \le (n+1)2^\tau$ (inequality (\ref{upper bound on measure})), $\Disc(P)=2^{-(2n-2)t}\cdot\Disc(p)$, and the fact that the discriminant of an integer polynomial is integral.

\Kurt{For a polynomial with $k$ non-vanishing coefficients, $\AD$ needs $O(k\cdot (\log n+\log(\Gamma+\log M(\sigma_p^{-1}))))=O(k\cdot\log(n\tau))$ iterations to isolate the real roots of $p$, where $k$ is defined as the number of non-vanishing coefficients of $p$, by Theorem~\ref{thm:treesizesparse}. The bound stated follows.}
\end{proof}


\section{Root Refinement}\label{sec:Root Refinement}

In the previous sections, we focused on the problem of isolating all real roots of a square-free polynomial $P\in\R[x]$. Given sufficiently good approximations of the coefficients of $P$, our algorithm $\AD$ returns isolating intervals $I_1$ to $I_m$ with the property that $\var(P,I_k)=1$ for all $k=1,\ldots,m$. This is sufficient for some applications (existence of real roots, computation of the number of real roots, etc.); however, many other applications also need very good approximations of the roots. In particular, this holds for algorithms to compute a cylindrical algebraic decomposition, where we have to approximate polynomials whose coefficients are polynomial expressions in the root of some other polynomial.\footnote{For instance, when computing the topology of an algebraic curve defined as the real valued zero set of a bivariate polynomial $f(x,y)\in\Z[x,y]$, many algorithms compute the real roots $\alpha$ of the resultant polynomial $R(x):=\operatorname{res}(f,f_y;y)\in\Z[x]$ first and then isolate the real roots of the fiber polynomials $f(\alpha,y)\in\R[y]$. The second step requires very good approximations of the root $\alpha$ in order to obtain good approximations of the coefficients of $f(\alpha,y)$.} 

In this section, we show that our algorithm $\AD$ can be easily modified to further refine the intervals $I_k$ to a width less than $2^{-\kappa}$, where $\kappa$ is any positive integer. Furthermore, our analysis in Sections~\ref{analysis:refinement1} and~\ref{analysis:refinement2} shows that the cost for the refinement is the same as for isolating the roots plus $\tilde{O}(n\cdot\kappa)$. Hence, as a bound in $\kappa$, the latter bound is optimal (up to logarithmic factors) since the amortized cost per root and bit of precision is
logarithmic in $n$ and $\kappa$.\\

Throughout this section, we assume that $z_1$ to $z_m$ are exactly the real roots of $P$ 
and that $I_k=(a_k,b_k)$, with $k=1,\ldots,m$, are corresponding isolating intervals as computed by $\AD$. In particular, it holds that $\var(P,I_k)=1$.
According to Theorem~\ref{Obreshkoff}, the Obreshkoff lens $L_n$ of each interval $I_k$ is also isolating for the root $z_k$. Hence, from the proof of~\cite[Lemma~5]{NewDsc} (see also~\cite[Figure~3.1]{NewDsc}), we conclude that 
\begin{align}\label{distance-to-other-roots}
|x-z_j|>\frac{\min(|x-a_k|,|x-b_k|)}{4n}\quad\text{for all }x\in I_k\text{ and all }j\neq k.
\end{align}   

\subsection{The Refinement Algorithm}\label{algo:refinement}

We modify $\AD$ so as to obtain an efficient algorithm for root refinement. The modification is based on two observations, namely that we can work with a simpler notion of multipoints and that we can replace the $0$-Test and the $1$-Test with a simpler test based on the sign of $P$ at the endpoints of an interval. In $\AD$, we used:
\begin{itemize}
\item[(A)] computation of an admissible point $m^*\in m[\epsilon]$, where 
$$\multipoint{m}{\epsilon} := \set{m_{i}:=m+(i-\lceil n/2\rceil)\cdot \epsilon}{i=0,\ldots,2\cdot \lceil n/2\rceil}$$ is a multipoint of size $2\cdot\lceil n/2\rceil+1$.
\item[(B)] execution of the $0$-Test/$1$-Test for an interval $(a',b')\subset (a,b)$, where $a'$ and $b'$ are admissible points of corresponding multipoints contained in $I$.\footnote{Notice that this step also uses the computation of admissible points.}
\end{itemize}
The reason for putting more than $n$ points into a multipoint was to guarantee, that at least one constituent point has a reasonable distance from all roots contained in the interval. Now, we are working on intervals containing only one root, and hence, can use multipoints consisting of only two points. An interval known to contain at most one root of $P$ contains no root if the signs of the polynomial at the endpoints are equal and contains a root if the signs are distinct (this assumes that the polynomial is nonzero at the endpoints). We will, therefore, work with the following modifications when 
 processing an interval $I\subset I_k$:
\begin{itemize}
\item[(A')] computation of an admissible point $m^*\in m[\epsilon]'$, where 
$$m[\epsilon]' =\{m_1',m_2'\}:= \{m-\lceil n/2\rceil\epsilon,m+\lceil n/2\rceil\epsilon\}\subset m[\epsilon]$$ consists of the first and the last point from $m[\epsilon]$ only.\footnote{In fact, one can show that choosing two arbitrary points from $m[\epsilon]$ does not affect any of the following results.}
\item[(B')] execution of a \emph{sign-test} on an interval $I'=(a',b')\subset (a,b)$ (i.e.,~the computation of $\sgn (f(a')\cdot f(b'))$), where $a'$ and $b'$ are admissible points in some $m[\epsilon]'$.\footnote{More precisely, we compute $s:=\sgn (P(a')\cdot P(b'))$. If $s>0$, then $I'$ contains no root. If $s<0$, then $I'$ isolates the root $z_k$. Since $\var(P,I_k)=1$, it follows that $s>0$ if and only if $\var(P,I')=0$, and $s<0$ if and only if $\var(P,I')=1$.} 
\end{itemize}

We now give details of our refinement method which we denote $\textsc{Refine}$. As input, \textsc{Refine} receives isolating intervals $I_1$ to $I_m$ for the real roots of $P$ as computed by $\AD$ and a positive integer $\kappa$. It returns isolating intervals $J_k$, with $J_k\subset I_k$ and width $w(J_k)<2^{-\kappa}$. \medskip

\begin{mdframed}[frametitle={{\bf Algorithm: Refine}}]
{\color{black}

\noindent{\bf Input:}
A polynomial $P(x)$ as in (\ref{def:P}), isolating intervals $I_j$ for the real roots of $P$ with $\var(P,I_j)=1$ for $j=1,\ldots,m$, and a positive integer $\kappa$.\smallskip

\noindent{\bf Output:}
Isolating intervals $I_j'$ for the real roots of $P$ with $I_j'\subseteq I_j$ and $w(I_j')<2^{-\kappa}$ for $j=1,\ldots,m$.\smallskip

\begin{itemize}
\item[(1)] $\mathcal{A}:=\{(I_{j},4)\}_{j=1,\ldots,m}$ and $\mathcal{O}:=\emptyset$.
\item[(2)] while $\mathcal{A}\neq \emptyset$ do
\begin{itemize}
\item[(2.1)] Choose an arbitrary pair $(I,N_{I})$ from $\mathcal{A}$, with $I=(a,b)$, and remove $(I,N_I)$ from $\mathcal{A}$ 
\item[(2.2)] Run the Boundary-Test and the Newton-Test with input $P$ and $I$, where the steps in (A) and (B) are replaced by the respective modifications in (A') and (B'). 
\item[(2.3)] If one of the tests in Step (2.2) returns True and an interval $I'\subseteq I$, then add $I$ to $\mathcal{O}$ if $w(I')<2^{-\kappa}$, and add $(I',N_{I'})=(I',N_I^2)$ to $\mathcal{A}$ if $w(I')\ge 2^{-\kappa}$. Then, go to Step (2.1).
\item[]\hfill\emph{(quadratic step)}
\item[(2.4)] Compute an admissible point
$m^* \in \multipoint{m(I)}{\frac{w(I)}{2^{\ceil{2 + \log n}}}}'$ using the algorithm {\bf Admissible Point}. 
\begin{itemize}
\item[(2.4.1)] $I':=(a,m^*)$, $I'':=(m^*,b)$, and $N_{I'}:=N_{I''}:=\max(4,\sqrt{N_{I}})$.
\item[(2.4.2)] If $P(a')\cdot P(m^*)<0$ and $w(I')<2^{-\kappa}$, add $I'$ to $\mathcal{O}$. 
\item[(2.4.3)] If $P(a')\cdot P(m^*)<0$ and $w(I')\ge 2^{-\kappa}$, add $(I',N_{I'})$ to $\mathcal{A}$.
\item[(2.4.4)] If $P(a')\cdot P(m^*)>0$ and $w(I'')<2^{-\kappa}$, add $I''$ to $\mathcal{O}$.
\item[(2.4.5)] If $P(a')\cdot P(m^*)>0$ and $w(I'')\ge 2^{-\kappa}$, add $(I'',N_{I''})$ to $\mathcal{A}$.
\end{itemize}\hfill\emph{(linear step)} 
\end{itemize}
\item[(3)] return $\mathcal{O}$
\end{itemize}
}
\end{mdframed}\bigskip
The reader may notice that, in comparison to $\AD$, there are only a constant number of polynomial evaluations at each node, and thus, there is no immediate need to use an algorithm for fast approximate multipoint evaluation.\footnote{However, we will later show how to make good use of approximate multipoint evaluation in order to improve the worst case bit complexity.} Namely, when processing a certain interval $I\in\mathcal{A}$, we have to compute admissible points $m^*\in m[\epsilon]'$ for a constant number of $m[\epsilon]'$, and each $m[\epsilon]'$ consists of two points (i.e.,~$m'_1=m-\lceil n/2\rceil$ and $m_2'=m+\lceil n/2 \rceil$) only. For computing an admissible point $m^*$, we evaluate $P(x)$ at $x=m_1'$ and $x=m_2'$ to an absolute error less than $2^{-L}$, where $L=1,2,4,8\ldots$. We stop as soon as we have computed a 4-approximation for at least one of the values $P(m_1')$ and $P(m_2')$. The cost for each such computation is bounded by 
\begin{align}\label{cost:admissible:outoftwo}
\tilde{O}(n(\tau_p+n\log \max(M(m_1'),M(m_2'))+\log M(\max(|P(m_1'|,|P(m_2')|)^{-1})))
\end{align}
bit operations; see the proof of Lemma~\ref{lem:singlepolyeval}.

\subsection{Analysis}\label{analysis:refinement1}
We first derive bounds on the number of iterations that \textsc{Refine} needs to refine an isolating interval $I_k$ to a size less than $2^{-\kappa}$.

\begin{lemma}\label{number-of-iterations}
For refining an interval $I_k$ to a size less than $2^{-\kappa}$, \textsc{Refine} needs at most 
$$s_{\max,k}\cdot |\mathcal{M}(I_k)|=O((\log n+\log(\log M(z_k)+\kappa)))\cdot |\mathcal{M}(I_k)|$$ 
iterations, where $s_{\max,k}$ has size $O(\log n+\log(\log M(z_k)+\kappa))=O(\log n+\log(\Gamma+\kappa))$ and $\mathcal{M}(I_k)$ is the set of roots contained in the one-circle region of $I_k$. The total number of iterations to refine all intervals $I_k$ to a size less than $2^{-\kappa}$ is $\tilde{O}(n(\log n+\log(\Gamma+\kappa)))$.
\end{lemma}

\begin{proof}
Similar as in the proof of Lemma~\ref{lemma:bounds1}, we first derive upper and lower bounds for the values $N_I$ and $w(I)$, respectively, where $I\subset I_k$ is an \emph{active} interval produced by \textsc{Refine}. 
According to property (\ref{cond:start1}), we have $w(I)\le w(I_k)\le 4\cdot M(z_k)^2$. Hence, it follows that either $N_I=4$ or $w(I)\le w(I_k)/\sqrt{N_I}\le 4\cdot M(z_k)^2/\sqrt{N_I}$, and thus, (notice that $w(I)\ge 2^{-\kappa}$)
$$N_I\le 16\cdot \frac{M(z_k)^4}{w(I)^2}\le 2^{4(\Gamma+1)+2\kappa}.$$
Furthermore, for \emph{each} interval $I\subset I_k$ produced by \textsc{Refine}, we have
\[
\min(2^{\Gamma},4 M(z_k)^2)\ge w(I)\ge \frac{w(J)}{N_J}\ge 2^{-3\kappa-4(\Gamma+1)},  
\]
where $J$ is the parent interval of $I$ of size $w(J)\ge 2^{-\kappa}$. The bound for the number of 
iterations is then an immediate consequence of Lemma~\ref{newtonsuccess} and of our considerations in the proof 
of Lemma~\ref{lemma:bound on smax}. Namely, exactly the same argument as in the proof of 
Lemma~\ref{lemma:bound on smax} shows that the \emph{maximal length of any path between splitting nodes}, denoted $s_{\max,k}'$, is $O(\log n +\log(\log M(z_k)+\kappa))$,\footnote{For \textsc{Refine}, a node $I$ is splitting if either $I$ is terminal (i.e.,~$w(I)<2^{-\kappa}$) or $\mathcal{M}(I)\neq \mathcal{M}(I')$ for the child $I'$ of $I$. If $I_k^*$ denotes the first node whose one-circle region isolates the root $z_k$ (i.e.~$|\mathcal{M}(I_k^*)|=1$), then it follows that the path connecting $I_k^*$ with $J_k$ has length less than or equal to $s_{\max,k}'$.} and thus the path from $I_k$ to the refined interval $J_k\subset I_k$ of size less 
than $2^{-\kappa}$ has length $s_{\max,k}'\cdot |\mathcal{M}(I_k)|$.  
\end{proof}

In the next step, we estimate the cost for processing an active interval $I$.

\begin{lemma}\label{cost-at-node}
For an active interval $I\subset I_k$ of size $w(I)\ge \sigma(z_k,P)/2$, the cost for processing $I$ is bounded by
$$
\tilde{O}(n(n+\tau_P+n\log M(z_i)+\log M(P'(z_i)^{-1}))),
$$
where $z_i$ is any root contained in the one-circle region of $I$. If $w(I)<\sigma(z_k,P)/2$, the cost for processing $I$ is bounded by
$$
\tilde{O}(n(\kappa+n+\tau_P+n \log M(z_k)+\log M(P'(z_k)^{-1}))).
$$
\end{lemma}
\begin{proof}
Suppose that $w(I)\ge \sigma(z_k,P)/2$, and let $\xi\in m[\epsilon]'$ be an admissible point that is computed when processing $I$. For at least one of the two points (w.l.o.g.~say $m_1'$) in $m[\epsilon]'$, the distance to the root $z_k$ as well as the distance to both endpoints of $I$ is at least $\lceil n/2\rceil\cdot \epsilon\ge n\cdot \epsilon/2$. Hence, from inequality (\ref{distance-to-other-roots}) we conclude that the distance from $m_1'$ to any root of $P$ is at least $\epsilon/8$. Now, exactly the same argument as in the proof of Lemma~\ref{lem:boundonmax} (with $x_{i_0}:=m_1'$) shows that 
\[
|P(\xi)|>2^{-6n-1}\cdot K^{-\mu(\Delta)-1}\cdot\sigma(z_i,P)\cdot|P'(z_i)|\quad\text{for all }z_i\in\Delta,
\] 
where $K\ge 2\cdot\lceil n/2\rceil$ is any positive real value, such that the disk $\Delta:=\Delta_{K\cdot\epsilon}(m)$ contains at least two roots of $P$. Since $w(I)\ge \sigma(z_k,P)/2$, it further follows that the disk $\Delta_{2w(I)}(m(I))$ contains at least two roots. Thus, we can use the same argument as in the proof of Lemma~\ref{lem:inequalityspecial} (type (P2)-(P5) cases) to prove that the inequality (\ref{special:lowerbound}) holds for $\xi$. In addition, inequality (\ref{special:lowerbound}) also holds for the endpoints of $I_k$ (as already proven in the analysis of the root isolation algorithm), and thus, by induction, it holds for the endpoints of any node $I\subset I_k$.
Hence, when processing $I$, there are a constant number of approximate polynomial evaluations with a precision bounded by 
\begin{align}\label{precisionbound1}
O(n\log n+\tau_P+n\log M(z_i)+\log M(\sigma(z_i,P)^{-1})+\log M(P'(z_i)^{-1}))\nonumber\\
O(n\log n+\tau_P+n\log M(z_i)+\log M(P'(z_i)^{-1})),
\end{align} 
where we again used that $\log M(\sigma(z_k,P)^{-1})=O(n\log M(z_i)+\tau_P+\log M(P'(z_i)^{-1}))$.
This proves the first part; see the proof of Theorem~\ref{main:theorem} and Lemma~\ref{lem:singlepolyeval}.

For the second part, we now assume that $w(I)<\sigma(z_k,P)/2$. Let $\xi\in m[\epsilon]'$ be an admissible point that is considered when processing $I$. Then, the disk $\Delta_{w(I)}(m(I))$ contains the root $z_k$ but no other root of $P$. Hence, for any $x\in I$, it holds that 
\begin{align*}
|P(x)|=|P_n|\cdot |x-z_k|\cdot\prod_{i\neq k}|x-z_i|\ge |P_n|\cdot|x-z_k|\cdot \prod_{i\neq k}\frac{|z_k-z_i|}{4}=|x-z_k|\cdot \frac{|P'(z_k)|}{ 2^{2(n-1)}}.
\end{align*}
Since the distance of at least one of the two points in $m[\epsilon]'$ (w.l.o.g.~say $m_1'$) to the root $z_k$ is larger than or equal to $n\epsilon/2\ge w(I)/(8N_I)$, it follows that
\begin{align*}
|P(\xi)|&\ge \frac{|P(m_1')|}{4}\ge\frac{1}{4}\cdot \frac{w(I)}{8N_I}\cdot 2^{-2(n-1)}\cdot |P'(z_k)|\\
&\ge \frac{w(I)^3}{2^9 M(z_k)^4}\cdot 2^{-2(n-1)}\cdot |P'(z_k)|
>2^{-4\Gamma-3\kappa-2n-7}\cdot |P'(z_k)|,
\end{align*}
where we used the bounds for $w(I)$ and $N_I$ as computed in the proof of Lemma~\ref{number-of-iterations}. Furthermore, the endpoints of $I$ fulfill the inequality (\ref{special:lowerbound}), and thus all approximate polynomial evaluations (when processing $I$) are carried out with an absolute precision of \begin{align}\label{precisionbound2}
& O(\log M(w(I)^{-1})+n\log n+\tau_P+n\cdot \log M(z_k)+\log M(\sigma(z_k,P)^{-1})+\log M(P'(z_k)^{-1}))=\nonumber\\
& O(\kappa+n\log n+\tau_P+n\cdot \log M(z_k)+\log M(P'(z_k)^{-1}))
\end{align}
This proves the second claim.
\end{proof}

Combining Lemma~\ref{number-of-iterations} and Lemma~\ref{cost-at-node} now yields the following result:

\begin{theorem}\label{complexity-refinement1}
The cost for refining $I_k$ to an interval of size less than $2^{-\kappa}$ is bounded by
$$
\tilde{O}(n\kappa+\sum_{i:z_i\in\mathcal{M}(I_k)}n(n+\tau_P+n\log M(z_i)+\log M(P'(z_i)^{-1}))).
$$
The cost for refining all isolating intervals to a size less than $2^{-\kappa}$ is bounded by
\begin{align}\label{totalcomplexity-refinement1}
\tilde{O}(n^2\kappa+n(n^2+n\log \Mea(P)+\sum_{i=1}^{n}\log M(P'(z_i)^{-1}))).
\end{align}
\end{theorem} 

\begin{proof}
We split the total cost into those for refining the interval $I_k$ to a size less than $\sigma(z_k,P)/2$ and into those for the additional refinement steps until the interval has size less than $2^{-\kappa}$. For the latter cost, we remark that $|\mathcal{M}(I)|=1$ if $I$ is an isolating interval for $z_k$ of width $w(I)<\sigma(z_k,P)/2$. Hence, there are at most $s_{\max,k}'$ refinement steps of $I$, each of cost $\tilde{O}(n(\kappa+n+\tau_P+n\log M(z_k)+\log M(P'(z_k)^{-1})))$. It remains to bound the cost for refining $I_k$ to a width of less than $\sigma(z_k,P)/2$. According to Lemma~\ref{cost-at-node}, the cost for processing an interval $I$ of width $w(I)\ge\sigma(z_k,P)/2$ is bounded by $\tilde{O}(n(n+\tau_P+n\log M(z_i)+\log M(P'(z_i)^{-1})))$, where we can choose an arbitrary root $z_i\in\mathcal{M}(I)$.
If we choose the root $z_i$ that is associated\footnote{Essentially, we use the the same definition as in Section~\ref{sec:bit-complexity}. More precisely, we say that a root $z_i$ is associated with $I$ if $z_i\in\mathcal{M}(I)$ and the number of children $I'\subset I$ with $z_i\in\mathcal{M}(I')$ is minimal for all roots in $\mathcal{M}(I)$. Notice that each root $z_i\in\Delta(I_k)$ is associated with at most $s_{\max,k}'$ intervals; see Lemma~\ref{number-of-iterations}.} with $I$, then each root in $\mathcal{M}(I_k)$ is considered at most $s_{\max,k}'$ many times. Thus, the first complexity bound follows.
The bound (\ref{totalcomplexity-refinement1}) for the total cost for refining all intervals follows immediately from the first bound and from the fact that the one-circle regions of the intervals $I_k$ are pairwise disjoint.    
\end{proof}

\subsection{Asymptotic Improvements}\label{analysis:refinement2}

In this section, we show that our complexity bound (\ref{totalcomplexity-refinement1}) for refining all intervals can be further improved, that is, the term $n^2\kappa$ can be replaced by $n\kappa$. We achieve this result by using fast approximate multipoint evaluation. 
For an integer $l$, with $0\le l\le\log\kappa+1$, we consider all active intervals in the refinement process whose width is larger than or equal to $2^{-2^l}$. We call each such interval an \emph{$l$-active} interval. 
We start with $l=0$ and proceed in rounds: For a fixed $l$, let (w.l.o.g.) $I_1',\ldots,I_{m(l)}'$, with $m(l)\le m$ and $I_k'\subset I_k$, be all $l$-active intervals. 
The crucial idea is now to carry out the polynomial evaluations for each of the $l$-active intervals "in parallel" by using fast approximate multipoint evaluation. That is, instead of computing admissible points $m_k^*$ of $m_k[\epsilon_k]'$ for each interval $I_k'$ independently, we aggregate these evaluations in one multipoint evaluation. We continue refining all $l$-active intervals in this way until all intervals have size less than $2^{-2^l}$. Once this happens, we proceed in the same manner with $l:=l+1$.
Notice that after a few iterations (for a fixed $l$) some of the $l$-active intervals might become smaller than $2^{-2^l}$, whereas other intervals are still $l$-active. Intervals which become smaller than $2^{-2^l}$ are then not considered anymore in this round.
The cost for each multipoint evaluation is comparable to the cost of the most expensive individual evaluation multiplied by a logarithmic factor. Furthermore, in each iteration, a constant number of multipoint evaluations is sufficient because for each interval $I_k'$, there are only constantly many evaluations. Hence, the cost for each iteration is bounded by
\begin{align}\label{bound-multipoint}
\tilde{O}(n(2^l+n+\tau_P+n\cdot \log M(z_i)+M(P'(z_i)^{-1})),
\end{align}
where $z_i$ is any root in the one-circle region of the interval $I_{k_0}'$, and $I_{k_0}'$ is the interval for which the highest precision is needed; see (\ref{precisionbound1}) and (\ref{precisionbound2}), and use that $w(I_{k_0}')\ge 2^{-2^l}$.
We now distinguish the following three cases:
\begin{itemize}
\item[(1)] $l=0$: The cost in (\ref{bound-multipoint}) is bounded by
\[
\tilde{O}(n(n+\tau_P+n\cdot \log M(z_i)+M(P'(z_i)^{-1})).
\]
  We allocate the cost to a root $z_i$ that is associated to the interval $I_{k_0}'$.
    
\item[(2)] $l>0$ and there exists an interval $I_k'$ with $\mathcal{M}(I_k')>1$: Since $2^{-2^{l-1}}>w(I_k')\ge 2^{-2^l}$, each root $z_i$ in $\mathcal{M}(I_k)$ has separation $\sigma(z_i,P)<2^{-2^{l-1}}$, and thus, replacing the term $2^l$ in (\ref{bound-multipoint}) by $\log M(\sigma(z_i,P)^{-1})$ yields
\[
\tilde{O}(n(\log M(\sigma(z_i,P)^{-1})+n+\tau_P+n\cdot \log M(z_j)+M(P'(z_j)^{-1})),
\] 
where $z_i$ is some root in $\mathcal{M}(I_k')$ and $z_j$ is some root in $\mathcal{M}(I_{k_0}')$. We allocate the cost to roots $z_i$ and $z_j$ that are associated to the intervals $I_k'$ and $I_{k_0}'$, respectively.
\item[(3)] $l>0$ and $\mathcal{M}(I_k')=1$ for all $k=1,\ldots,m(l)$: We allocate the cost to a root $z_i$ that is associated to $I_{k_0}'$. 
\end{itemize}

It remains to sum up the cost over all iterations. The sum over all iterations of type (1) and (2) is bounded by
\begin{align*}
\tilde{O}(n(n^2+n\log \Mea(P)+\sum_{i=1}^n\log M(\sigma(z_i,P)^{-1})+\sum_{i=1}^{n}\log M(P'(z_i)^{-1})))=\\
\tilde{O}(n(n^2+n\log \Mea(P)+\sum_{i=1}^{n}\log M(P'(z_i)^{-1})))
\end{align*}
because the cost of an iteration is allocated to a certain root $z_i$ only a logarithmic number of times. For the sum over all iterations of type (3), we remark that, for a certain $l$, there can be at most 
$\max\nolimits_{k=1,\ldots,m} s_{\max,k}'$ iterations of type (3). Namely, the number of iterations to refine a certain interval $I_k'$ with $\mathcal{M}(I_k')=1$ to a size less than $2^{-\kappa}$ is bounded by $s_{\max,k}'$. Hence, the sum of the first term $n\cdot 2^l$ in (\ref{bound-multipoint}) over all $l$ is bounded by $\max\nolimits_{k=1,\ldots,m} s_{\max,k}\cdot n\cdot\kappa$. The sum over the remaining term is again bounded by
\[
\tilde{O}(n(n^2+n\log \Mea(P)+\sum_{i=1}^{n}\log M(P'(z_i)^{-1})))
\]  
because the cost of an iteration is allocated to a certain root $z_i$ only a logarithmic number of times. 
We summarize:

\begin{theorem}[Restatement of Theorem~\ref{main: refinement}]\label{mainresult2}
Let $P = P_n x^n + \ldots + P_0 \in \R[x]$ be a real polynomial with $1/4 \le \abs{P_n} \le 1$, and let $\kappa$ be a positive integer. Computing isolating intervals of size less than $2^{-\kappa}$ for all real roots needs a number of bit operations bounded by
\begin{align}\label{result:bitcomplexity:total} 
& \tilde{O}(n\cdot(\kappa+n^2+n \log \Mea(P)+\sum_{i=1}^n \log M(P'(z_i)^{-1}))) \\
&\qquad = \tilde{O}(n\cdot(\kappa+n^2+n \log \Mea(P)+  \log M(\Disc(P)^{-1}))). \end{align}
The coefficients of $P$ have to be approximated with quality
$$\tilde{O}(\kappa+n+\tau_{P}+\max_{i}(n\log M(z_{i})  + \log M(P'(z_{i})^{-1}) )).
$$
For a polynomial $P$ with integer coefficients of size less than $2^\tau$, computing isolating intervals of size less than $2^{-\kappa}$ for all real roots needs $\tilde{O}(n(n^2+n\tau+\kappa))$ bit operations.
\end{theorem}

\section{Conclusion}

We have introduced a novel subdivision algorithm, denoted $\AD$, to compute isolating intervals for the real roots of a square-free polynomial with real coefficients. The algorithm can also be used to further refine the isolating intervals to an arbitrary small size.

In our approach, we combine the Descartes method with Newton iteration and approximate (but certified) arithmetic. As a result, $\AD$ uses an almost optimal number of iterations, and the precision demand as well as the working precision are directly related to the actual geometric locations of the roots; hence, the algorithm adapts to the actual hardness of the input. The bit complexity of our method matches that of Pan's method from 2002, which is the best algorithm known and goes 
back to Sch\"onhage's splitting circle method from 1982. By comparison, our approach is completely different from Pan's method and, in addition, it is simpler. Because of its simpleness, we consider our algorithm to be well suited for an efficient implementation. Furthermore, it can be used to isolate the roots in a given interval only, whereas Pan's method has to compute all complex roots at the same time.

\Kurt{The first author, A. Kobel, and F. Rouillier are currently working on an implementation of $\AD$. More precisely, they are considering a randomized version of the algorithm, where admissible (subdivision) points at chosen randomly and not via approximate multipoint evaluation as proposed in this paper (see Section~\ref{sec:multipoint}). They expect the randomized version to show good practical behavior. It may even have an expected bit complexity comparable to the algorithm presented in this paper.}

\bibliography{bib}

\begin{thebibliography}{10}

\bibitem{abbott-quadratic}
J.~Abbott.
\newblock {Quadratic Interval Refinement for Real Roots}.
\newblock {\em Communications in Computer Algebra}, 28:3--12, 2014.
\newblock Poster presented at the International Symposium on Symbolic and
  Algebraic Computation (ISSAC), 2006.

\bibitem{akritas-strzebonski:comparison:05}
A.~G. Akritas and A.~Strzebo{\'n}ski.
\newblock A comparative study of two real root isolation methods.
\newblock {\em Nonlinear Analysis:Modelling and Control}, 10(4):297--304, 2005.

\bibitem{Alesina-Galuzzi}
A.~Alesina and M.~Galuzzi.
\newblock A new proof of {V}incent's theorem.
\newblock {\em L'Enseignement Mathematique}, 44:219--256, 1998.

\bibitem{Bini-Fiorentino}
D.~Bini and G.~Fiorentino.
\newblock {Design, Analysis, and Implementation of a Multiprecision Polynomial
  Rootfinder}.
\newblock {\em Numerical Algorithms}, 23:127--173, 2000.

\bibitem{DBLP:journals/jsc/BurrK12}
M.~Burr and F.~Krahmer.
\newblock Sqfreeeval: An (almost) optimal real-root isolation algorithm.
\newblock {\em Journal of Symbolic Computation}, 47(2):153--166, 2012.

\bibitem{Collins}
G.~E. Collins.
\newblock Continued fraction real root isolation using the {H}ong bound.
\newblock {\em Journal of Symbolic Computation}, 2014.
\newblock in press.

\bibitem{Collins-Akritas}
G.~E. Collins and A.~G. Akritas.
\newblock Polynomial real root isolation using {D}escartes' rule of signs.
\newblock In {\em Symposium on symbolic and algebraic computation (SYMSAC)},
  pages 272--275, 1976.

\bibitem{du-sharma-yap:sturm:07}
Z.~Du, V.~Sharma, and C.~Yap.
\newblock Amortized bounds for root isolation via {S}turm sequences.
\newblock In {\em Symbolic Numeric Computation (SNC)}, pages 113--130, 2007.

\bibitem{eigenwillig-phd}
A.~Eigenwillig.
\newblock {\em Real Root Isolation for Exact and Approximate Polynomials Using
  {D}escartes' Rule of Signs}.
\newblock PhD thesis, Saarland University, 2008.

\bibitem{DBLP:conf/casc/EigenwilligKKMSW05}
A.~Eigenwillig, L.~Kettner, W.~Krandick, K.~Mehlhorn, S.~Schmitt, and
  N.~Wolpert.
\newblock A {D}escartes algorithm for polynomials with bit-stream coefficients.
\newblock In {\em Computer Algebra in Scientific Computation (CASC)}, pages
  138--149, 2005.

\bibitem{Eigenwillig-Sharma-Yap}
A.~Eigenwillig, V.~Sharma, and C.~K. Yap.
\newblock Almost tight recursion tree bounds for the descartes method.
\newblock In {\em International Symposium on Symbolic and Algebraic Computation
  (ISSAC)}, pages 71--78, 2006.

\bibitem{Pan:survey}
I.~Z. Emiris, V.~Y. Pan, and E.~P. Tsigaridas.
\newblock Algebraic algorithms.
\newblock In {\em Computing Handbook, Third Edition: Computer Science and
  Software Engineering}, pages 10: 1--30. 2014.

\bibitem{DBLP:journals/jsc/Fortune02}
S.~Fortune.
\newblock An iterated eigenvalue algorithm for approximating roots of
  univariate polynomials.
\newblock {\em Journal of Symbolic Computation}, 33(5):627--646, 2002.

\bibitem{Gourdon}
X.~Gourdon.
\newblock {\em Combinatoire, Algorithmique et G\'{e}om\'{e}trie des Polynomes.}
\newblock PhD thesis, \'{E}cole Polytechnique, 1996.

\bibitem{snc-benchmarks09}
M.~Hemmer, E.~P. Tsigaridas, Z.~Zafeirakopoulos, I.~Z. Emiris, M.~I. Karavelas,
  and B.~Mourrain.
\newblock Experimental evaluation and cross-benchmarking of univariate real
  solvers.
\newblock In {\em Symbolic Numeric Computation (SNC)}, pages 45--54, 2009.

\bibitem{Johnson-Krandick}
J.~R. Johnson and W.~Krandick.
\newblock Polynomial real root isolation using approximate arithmetic.
\newblock In {\em International Symposium on Symbolic and Algebraic Computation
  (ISSAC)}, pages 225--232, 1997.

\bibitem{Kamath}
N.~Kamath.
\newblock Subdivision {A}lgorithms for {C}omplex {R}oot {I}solation:
  {E}mpirical {C}omparisons.
\newblock Master's thesis, Kellogg College, University of Oxford, 2010.

\bibitem{DBLP:journals/corr/abs-1104-1362}
M.~Kerber and M.~Sagraloff.
\newblock Root refinement for real polynomials using quadratic interval
  refinement.
\newblock {\em Journal of Computational and Applied Mathematics}, 280:377--395,
  2015.
\newblock a preliminary version appeared in the proceedings of the
  International Symposium on Symbolic and Algebraic Computation (ISSAC), 2011.

\bibitem{DBLP:journals/jc/Kirrinnis98}
P.~Kirrinnis.
\newblock Partial fraction decomposition in (z) and simultaneous newton
  iteration for factorization in {C}[z].
\newblock {\em Journal of Complexity}, 14(3):378--444, 1998.

\bibitem{Kobel-Rouillier-Sagraloff}
A.~Kobel, F.~Rouillier, and M.~Sagraloff.
\newblock personal communication.

\bibitem{DBLP:journals/corr/abs-1304.8069}
A.~Kobel and M.~Sagraloff.
\newblock Fast approximate polynomial multipoint evaluation and applications.
\newblock {\em CORR}, abs/1304.8069, 2013.

\bibitem{DBLP:journals/corr/KobelS14}
A.~Kobel and M.~Sagraloff.
\newblock On the complexity of computing with planar algebraic curves.
\newblock {\em Journal of Complexity}, 31(2):206--236, 2014.

\bibitem{McNamee:2002}
J.~M. McNamee.
\newblock A 2002 update of the supplementary bibliography on roots of
  polynomials.
\newblock {\em Journal of Computational and Applied Mathematics},
  142(2):433--434, 2002.

\bibitem{McNamee2007}
J.~M. McNamee.
\newblock {\em Numerical Methods for Roots of Polynomials}.
\newblock Number~1 in Studies in Computational Mathematics. Elsevier Science,
  2007.

\bibitem{McNamee-Pan2013}
J.~M. McNamee and V.~Y. Pan.
\newblock {\em Numerical Methods for Roots of Polynomials}.
\newblock Number~2 in Studies in Computational Mathematics. Elsevier Science,
  2013.

\bibitem{MSW-rootfinding2013}
K.~Mehlhorn, M.~Sagraloff, and P.~Wang.
\newblock From approximate factorization to root isolation with application to
  cylindrical algebraic decomposition.
\newblock {\em Journal of Symbolic Computation}, 66:34--69, 2015.
\newblock A preliminary version appeared in the proceedings of the
  International Symposium on Symbolic and Algebraic Computation (ISSAC), 2013.

\bibitem{Mignotte}
M.~Mignotte.
\newblock {\em Mathematics for Computer Algebra}.
\newblock Springer, 1992.

\bibitem{Obreschkoff:book}
N.~Obreshkoff.
\newblock {\em Verteilung und Berechnung der Nullstellen reeller Polynome}.
\newblock VEB Deutscher Verlag der Wissenschaften, 1963.

\bibitem{Obrechkoff:book-english}
N.~Obreshkoff.
\newblock {\em Zeros of Polynomials}.
\newblock Marina Drinov, Sofia, 2003.
\newblock Translation of the Bulgarian original.

\bibitem{Pan:alg}
V.~Pan.
\newblock {Univariate Polynomials: Nearly Optimal Algorithms for Numerical
  Factorization and Root Finding}.
\newblock {\em Journal of Symbolic Computation}, 33(5):701--733, 2002.

\bibitem{pantsi:ISSAC13}
V.~Pan and E.~Tsigaridas.
\newblock On the boolean complexity of real root refinement.
\newblock In {\em International Symposium on Symbolic and Algebraic Computation
  (ISSAC)}, pages 1--8, 2013.

\bibitem{Pan:history}
V.~Y. Pan.
\newblock Solving a polynomial equation: Some history and recent progress.
\newblock {\em SIAM Review}, 39(2):187--220, 1997.

\bibitem{DBLP:journals/jc/Renegar87}
J.~Renegar.
\newblock On the worst-case arithmetic complexity of approximating zeros of
  polynomials.
\newblock {\em Journal of Complexity}, 3(2):90--113, 1987.

\bibitem{Ritzmann19861}
P.~Ritzmann.
\newblock A fast numerical algorithm for the composition of power series with
  complex coefficients.
\newblock {\em Theoretical Computer Science}, 44(0):1 -- 16, 1986.

\bibitem{rouillier-zimmermann:roots:04}
F.~Rouillier and P.~Zimmermann.
\newblock Efficient isolation of [a] polynomial's real roots.
\newblock {\em Journal of Computational and Applied Mathematics}, 162:33--50,
  2004.

\bibitem{NewDsc}
M.~Sagraloff.
\newblock When {N}ewton meets {D}escartes: A simple and fast algorithm to
  isolate the real roots of a polynomial.
\newblock In {\em International Symposium on Symbolic and Algebraic Computation
  (ISSAC)}, pages 297--304, 2012.

\bibitem{sagraloff-sparse-2014}
M.~Sagraloff.
\newblock A near-optimal algorithm for computing real roots of sparse
  polynomials.
\newblock In {\em International Symposium on Symbolic and Algebraic Computation
  (ISSAC)}, pages 359--366, 2014.

\bibitem{Sagraloff2014DSC}
M.~Sagraloff.
\newblock On the complexity of the {D}escartes method when using approximate
  arithmetic.
\newblock {\em Journal of Symbolic Computation}, 65(0):79 -- 110, 2014.

\bibitem{Yap-Sagraloff-Bolzano}
M.~Sagraloff and C.~K. Yap.
\newblock A simple but exact and efficient algorithm for complex root
  isolation.
\newblock In {\em International Symposium on Symbolic and Algebraic Computation
  (ISSAC)}, pages 353--360, 2011.

\bibitem{Schoenberg}
I.~Schoenberg.
\newblock \"{U}ber variationsvermindernde lineare {T}ransformationen.
\newblock {\em Mathematische Zeitschrift}, pages 321--328, 1930.

\bibitem{schonhage:fundamental}
A.~Sch{\"o}nhage.
\newblock The fundamental theorem of algebra in terms of computational
  complexity, 1982; updated 2004.
\newblock Manuscript, Department of Mathematics, University of T{\"u}bingen.

\bibitem{sharma}
V.~Sharma.
\newblock Complexity of real root isolation using continued fractions.
\newblock {\em Theoretical Computer Science}, 409:292--310, 2008.

\bibitem{te-cf:08}
E.~P. Tsigaridas and I.~Z. Emiris.
\newblock On the complexity of real root isolation using continued fractions.
\newblock {\em Theoretical Computer Science}, 392(1-3):158--173, 2008.

\bibitem{JvH2008}
J.~van~der Hoeven.
\newblock Fast composition of numeric power series.
\newblock Technical Report 2008-09, Universit\'e Paris-Sud, France, 2008.
\newblock \url{http://www.texmacs.org/joris/fastcomp/fastcomp.html}.

\bibitem{Yakoubsohn2000}
J.-C. Yakoubsohn.
\newblock Finding a cluster of zeros of univariate polynomials.
\newblock {\em Journal of Complexity}, 16(3):603 -- 638, 2000.

\bibitem{yap-fundamental}
C.~K. Yap.
\newblock {\em Fundamental Problems of Algorithmic Algebra}.
\newblock Oxford University Press, 2000.

\end{thebibliography}
\bibliographystyle{abbrv}
\end{document}